\documentclass[aps,reprint,floatfix]{revtex4-1}
\usepackage{amsmath,amssymb,graphicx,mathtools}
\usepackage{float}

\usepackage[T1]{fontenc}
\usepackage{colortbl}

\makeatletter
\DeclareRobustCommand{\cev}[1]{%
  \mathpalette\do@cev{#1}%
}

\makeatother

\newcounter{theoremcounter}
\stepcounter{theoremcounter}

\newtheorem{theorem}{Theorem}

\newtheorem{corollary}{Corollary}
\newtheorem{definition}{Definition}

\newenvironment{proof}[1][Proof]{\begin{trivlist}
\item[\hskip \labelsep {\bfseries #1}]}{\end{trivlist}}

\newcommand{\qed}{\hfill $\blacksquare$}

\newcommand{\bs}{\boldsymbol}
\newcommand{\bra}[1]{\left\langle #1\right|}
\newcommand{\ket}[1]{\left|#1\right\rangle}
\newcommand{\braket}[2]{\left\langle #1|#2\right\rangle}
\newcommand{\expval}[3]{\left\langle #1\middle|#2\middle|#3\right\rangle}
\newcommand{\ketbra}[2]{\ket{#1}\bra{#2}}

\DeclareMathOperator{\Tr}{Tr}

\newcommand{\Mod}[1]{\ \mathrm{mod}\ #1}

\makeatletter
\newcommand{\vast}{\bBigg@{4}}
\newcommand{\Vast}{\bBigg@{5}}
\makeatother

\begin{document}

\title{Improved Strong Simulation of Universal Quantum Circuits}
\author{Lucas Kocia}
\affiliation{Sandia National Laboratories, Livermore, California 94550, U.S.A.}
\begin{abstract}
  We find a scaling reduction in the stabilizer rank of the twelve-qubit tensored \(T\) gate magic state. This lowers its asymptotic bound to \(2^{\sim 0.463 t}\) for multi-Pauli measurements on \(t\) magic states, improving over the best previously found bound of \(2^{\sim 0.468 t}\). We numerically demonstrate this reduction. This constructively produces the most efficient strong simulation algorithm of the Clifford+\(T\) gateset to relative or multiplicative error.  We then examine the cost of Pauli measurement in terms of its \emph{Gauss sum rank}, which is a slight generalization of the stabilizer rank and is a lower bound on its asymptotic scaling. We demonstrate that this lower bound appears to be tight at low \(t\)-counts, which suggests that the stabilizer rank found at the twelve-qubit state can be lowered further to \(2^{\sim 0.449 t}\) and we prove and numerically show that this is the case for single-Pauli measurements. Our construction directly shows how the reduction at \(12\) qubits is iteratively based on the reduction obtained at \(6\), \(3\), \(2\), and \(1\) qubits. This explains why novel reductions are found at tensor factors for these number of qubit primitives, an explanation lacking previously in the literature. Furthermore, in the process we observe an interesting relationship between the T gate magic state's stabilizer rank and decompositions that are Clifford-isomorphic to a computational sub-basis tensored with single-qubit states that produce minimal unique stabilizer state inner products -- the same relationship that allowed for finding minimal numbers of unique Gauss sums in the odd-dimensional qudit Wigner formulation of Pauli measurements.
\end{abstract}
\maketitle

\section{Introduction}
Universal quantum computation can be achieved using the Clifford+\(T\) gateset, stabilizer states and Pauli measurement. Classical \emph{strong} simulation is defined as the task of calculating the probabilities of quantum circuits composed of this gateset to relative error. It is \(\#P\)-hard~\cite{Huang18_2} and so is expected to require exponential resources on a classical computer. Improvements on the exponential factor in the runtime of classical algorithms are important for simulating near-term noisy intermediate-scale quantum (NISQ) computers~\cite{Preskill18} and establishing the limit of their classical simulation~\cite{Harrow17}. Consequently, a great deal of effort has recently focused on the development and characterization of faster classical simulation algorithms~\cite{Stahlke14,Pashayan15,Bravyi16_2,Bravyi16_1,Dax17,Howard17,Bennink17,Howard18,Huang18,Bu19,Qassim19,Rall19,Pashayan20,Seddon20,Huang20} of universal quantum circuits to keep pace with experimental advances in quantum near-term devices.

The output of any circuit composed of Clifford and \(T\) gates can be written as a Pauli-based computation with \(k\) \(T\) gate magic states and stabilizer states on \(n\) qubits~\cite{Bravyi05}:
\begin{equation}
  \label{eq:probofapp}
  \Pi_k = \Tr \left[ \hat \Pi \hat \rho_k \right],
\end{equation}
\(\hat \rho_k = \ketbra{\Psi_k}{\Psi_k}\) for \(\ket{\Psi_k} = \ket{T}^{\otimes k}\ket{0}^{\otimes (n-k)}\), \(\ket T = (\ket 0 + \sqrt{i} \ket 1)/\sqrt{2}\) is the \(T\) gate magic state, and \(\hat \Pi = \prod_{i=1}^n (\hat I + p_i \hat P_i)\), \(p_i = \pm \frac{1}{2}\), for \(\hat P_i\) a (multi-qubit) Pauli operator.

Tensoring \(T\) gate magic states extends Eq.~\ref{eq:probofapp} to quantum universality. In other words, taking the limit \(k \rightarrow \infty\) approximates any quantum probability. 

Classical simulation of \(\Pi_k\) (Eq.~\ref{eq:probofapp}) is frequently performed in terms of stabilizer state inner products. The number of these inner products scales exponentially with \(k\) and so their number should be targeted for reduction in the development of any classical simulation that aspires to be more efficient. One way to minimize the number of inner products is by finding a decomposition of \(\ket{T}^{\otimes k}\) into fewer stabilizer states. Since tensor products of stabilizer states \(\ket 0\) are stabilizer states themselves, a stabilizer decomposition of \(\ket{T}^{\otimes k}\) can trivially be extended to a stabilizer decomposition of \(\ket{\Psi_k}\) with the same number of states. The minimum such decomposition of \(\ket{T}^{\otimes k}\) is called the stabilizer rank, \(\chi_k\), of \(\ket{T}^{\otimes k}\).

The tensor product of two stabilizer states is a stabilizer state, which means there exists a \emph{trivial tensor bound} on the stabilizer rank for all integer powers of a state: \(\chi_t \le (\chi_k)^{t/k} \) where \(t\) is a multiple of \(k\). This bound is sub-multiplicative; it is possible that the actual stabilizer rank $\chi_{t}$ for any multiple of \(k\) is strictly less than this trivial bound. If so, this implies that $\ket{\Psi}^{\otimes t}$ has a more efficient stabilizer decomposition and thereby establishes a more efficient asymptotic scaling bound.

Looking for better \(\chi_k\) scaling that beats prior trivial tensor bounds has required Monte Carlo numerical searchs of the stabilizer space~\cite{Bravyi16_2} and produced the following novel reductions for the \(T\) gate magic state: \(\chi_1 = 2\), \(\chi_2 = 2\), \(\chi_3 = 3\), and \(\chi_6 \le 7\)~\cite{Howard18}. The trivial tensor bound implies the following upper bounds: \(\chi_t \le (\chi_1)^t = 2^t\) for arbitrary \(t\), \(\chi_t \le (\chi_2)^{t/2} = 2^{0.5 t}\) for even \(t\), \(\chi_t \le (\chi_3)^{t/3} = 2^{\sim 0.53 t}\) for \(t\) a multiple of \(3\), and \(\chi_t \le (\chi_6)^{t/6} = 2^{\sim 0.47 t}\) for \(t\) a multiple of \(6\). It is clear that the last bound provides the most favorable such scaling.

To find a better asymptotic scaling via lower stabilizer decompositions requires finding \(\chi_k\) for larger \(k\). Unfortunately, the stabilizer rank grows at least linearly with \(k\) and the number of stabilizer states increase as \(2^{(1/2+o(1))k^2}\)~\cite{Aaronson04}. Therefore any numerical search must contend with a prohibitive search space that, for instance, grows to \(>10^{10^3}\) states by \(k=12\). Monte Carlo searches implemented on current hardware cease converging appreciably at \(k>7\). Therefore, a non-numerical method is especially desirable.

Here we take an algebraic approach to minimizing the number of inner products, which does not suffer from the scaling problems inherent in numerical searches. 

We establish the following results:
\begin{enumerate}
\item The six-qubit T gate magic state produces the best previously-known stabilizer rank scaling of \((\chi_6)^{t/6} = 2^{\sim 0.468 t}\) Gaussian eliminations of \(t \times t\) symmetric matrices with entries in \(\mathbb Z/2 \mathbb Z\) for \(t\) T gates. We find a reduction in the stabilizer rank of the twelve-qubit \(T\) gate magic state which improves this asymptotic scaling to \(2^{\sim 0.463 t}\). We numerically demonstrate this performance (see Fig~\ref{fig:strongsimscaling}) and freely release code that performs universal Pauli-based computations with this improved scaling.

\item We demonstrate an alternative stabilizer decomposition that also produces a minimal number of stabilizer state inner products \(\xi_k\) for single-Pauli measurements (i.e. restricting \(\hat \Pi\) in Eq.~\ref{eq:probofapp} to be a single-Pauli measurement: \((\hat I + \hat P)/2\)), if only the unique Gauss sums are found and counted. This alternative approach exhibits the same novel reductions in its scaling as \(\chi_k\) from tensoring products of \mbox{one-}, \mbox{two-}, \mbox{three-}, six-qubit T gates to produce any higher multiple while requiring less polynomial overhead since entanglement does not exist for single-Pauli measurements on tensored states. As a result, the new method is far easier to find reductions with.

\begin{itemize}
\item The Gauss sum scaling from the new method produces an asymptotic lower bound on the stabilizer rank scaling for single-Pauli measurement (see Theorem~\ref{th:asymptoticlowerbound}) which empirically appears to be tight.
\item It is important to restrict to decompositions of linear combinations of Gauss sums, not products of Gauss sums, to achieve this tight lower bound, as discussed in Section~\ref{sec:singlepaulimeas}.
\end{itemize}
\item Since finding the optimal Gauss sums is demonstrably easier than the stabilizer rank, we leverage this to push past prior stabilizer rank searches and find a new novel reduction at \(k=12\). This produces a new bound of \((\xi_{12})^{t/12} \le 2^{\sim 0.449 t}\) Gaussian eliminations (see Fig~\ref{fig:scaling}). If the lower bound in Theorem~\ref{th:asymptoticlowerbound} is tight and extends to multi-Pauli measurement, this implies that the reduction we found of \((\chi_{12})^{t/12} \le 2^{\sim 0.463 t}\) can be reduced further to at least \((\xi_{12})^{t/12}\).
\begin{itemize}
\item Our constructive proof directly shows how the reduction of \(k=12\) is iteratively based on the reduction obtained from \(k=6\), \(k=3\), \(k=2\), and \(k=1\). This explains why the novel reductions are found from these particular values of \(k\), an explanation lacking previously in the stabilizer inner product approach.
\end{itemize}
\item Since our proof is constructive, we explicitly give the Gauss sum decompositions that saturate this new bound for single-Pauli measurements. We release code of this decomposition that validates its scaling.
\end{enumerate}

This approach to simulating universal quantum probabilities builds upon the foundations of prior advances in stabilizer rank-based approaches~\cite{Bravyi16_1, Bravyi16_2, Qassim19, Huang20} and alternative quasiprobability approaches~\cite{Stahlke14,Pashayan15,Bennink17,Howard17,Huang18,Seddon19,Seddon20}. The latter approach tackles the sampling version of this problem---the task of sampling from the probabilities of quantum circuits composed of a universal gateset to additive error---which belongs to the \(BQP\) complexity class and is sometimes called \emph{weak simulation}. In particular, in this work, we will demonstrate a striking relationship between the Gauss sums that are produced when calculating expectation values in the Wigner and generalized quasiprobability representations that are frequently used in these latter approaches~\cite{Pashayan15,Bennink17,Kocia18_2,Seddon19,Kocia20}, and stabilizer-based decompositions that are similar to dyadic frames~\cite{Seddon20} and their related decompositions.

\section{Novel Stabilizer Rank Reduction}
\label{sec:stabrankreduction}

Clifford transforming the T gate magic state to \(e^{\frac{\pi i}{8}} S\ket{T} = \cos\frac{\pi}{8} \ket 0 + \sin \frac{\pi}{8} \ket 1 \equiv \ket H\), where \(S = \text{diag}\{1, i\}\), produces a state with only real coefficients. This implies that its minimal stabilizer decomposition can be written with only real coefficients on real-valued stabilizer states as well, which reduces the stabilizer search space. One such decomposition found numerically~\cite{Bravyi16_1} is,
\begin{eqnarray}
\label{eq:sixqubitmagicstatedecomp}
\ket{H^{\otimes 6}} &\propto& (-16 + 12\sqrt{2}) \ket{B_{6,0}} + (96 - 68\sqrt{2})\ket{B_{6,6}} \nonumber\\
                    && + (10 - 7\sqrt{2})\ket{E_6} + (-14 + 10\sqrt{2})\ket{O_6} \nonumber\\
                    && +(7 - 5 \sqrt{2})Z^{\otimes 6} \ket{K_6} + (10 - 7\sqrt{2})\ket{\phi'_0}\\
                    && + (10 - 7 \sqrt{2})\ket{\phi''_0}, \nonumber
\end{eqnarray}
where \(\ket B_{6,0} = \ket{0}^{\otimes 6}\), \(\ket B_{6,6} = \ket{1}^{\otimes 6}\), \(\ket{E_6}\) (\(\ket{O_6}\)) is the sum of all computational basis states with even (odd) Hamming weight, and the definitions for the rest of the stabilizer states can be found in~\cite{Bravyi16_1}.

It follows that \(\ket{H^{\otimes 12}} = \ket{H^{\otimes 6}} \otimes \ket{H^{\otimes 6}}\) and so it can be decomposed into \(7^2 = 49\) stabilizer states. These are the stabilizer states found by tensoring the stabilizer states in Eq.~\ref{eq:sixqubitmagicstatedecomp} with themselves.

Two stabilizer states in this decomposition (with their coefficients) are
\begin{equation}
\label{eq:twelvequbitstabstate1}
  (-16 + 12\sqrt{2}) (96 - 68\sqrt{2}) ( \ket{B_{6,0} \otimes B_{6,6}} + \ket{B_{6,6}\otimes B_{6,0}}).
\end{equation}

This linear combination also happens to equal a stabilizer state itself; \(\ket{B_{6,0} \otimes B_{6,6}} {=} \ket{0^{\otimes 6}1^{\otimes 6}}\) and \(\ket{B_{6,6} \otimes B_{6,0}} {=} \ket{1^{\otimes 6}0^{\otimes 6}}\) transform to \(\ket{0^{\otimes 5} \otimes 1\otimes 0^{\otimes 6}}\) and \(\ket{0^{\otimes 11} \otimes 1}\), respectively, by the Clifford transformation \(\prod_{j=2}^{6} C_{7,j+6}C_{1,j}\), where \(C_{i,j}\) is a controlled-not gate in control qubit \(i\) and target qubit \(j\), and so are Clifford isomorphic to a Bell state. 

\begin{figure}[ht]
  \includegraphics[scale=0.55]{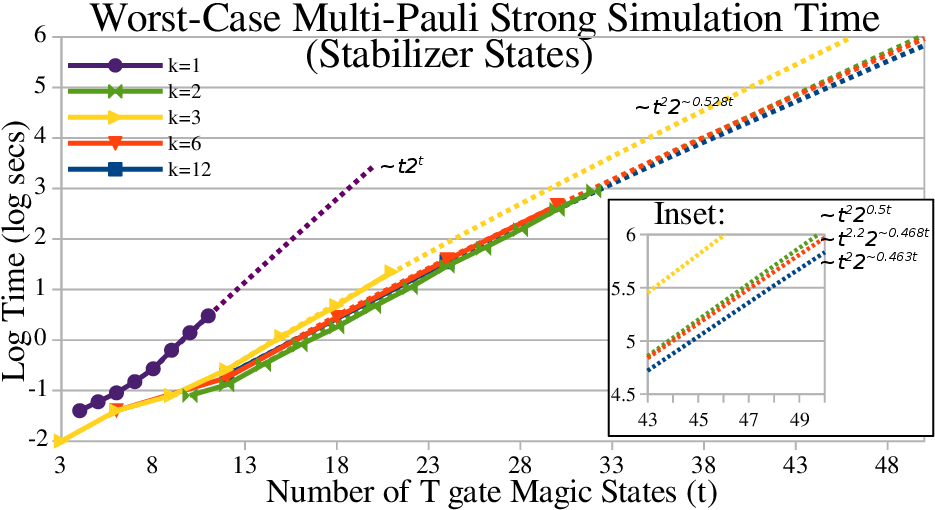}
  \caption{100 stabilizer states are sampled to calculate a random \(n\)-Pauli expectation value to some fixed relative error \(> 10\%\).}
  \label{fig:strongsimscaling}
\end{figure}

Two other stabilizer states in this decomposition (with their coefficients) are
\begin{equation}
\label{eq:twelvequbitstabstate2}
  (10 - 7\sqrt{2})(-14 + 10\sqrt{2}) (\ket{E_6} \ket{O_6} + \ket{O_6} \ket{E_6}).
\end{equation}

Six stabilizer generators of the stabilizer group of \(\ket{E_6}\) are \(XXXXXX\), \(XXIIII\), \(IXXIII\), \(IIXXII\), \(IIIXXI\), and \(ZZZZZZ\). \(\ket{O_6}\) has the same generators with the exception that the last one has a minus sign. This means that there exists a Clifford transformation that transforms both of them to computational basis states. Their tensor products must also be Clifford-isomorphic to orthogonal computational states. Hence, Eq.~\ref{eq:twelvequbitstabstate2} is also equal to a stabilizer state since it is Clifford-isomorphic to a Bell state.

Reducing these these two pairs of stabilizer states to two stabilizer states brings down the number of stabilizer states in the trivial tensored decomposition from \(49\) down to \(47\). This scales asymptotically as \(2^{\sim 0.463 t}\) for \(t\) a multiple of \(12\). The explicit form of the stabilizer decomposition for \(\ket{T^{\otimes 12}}\) can be found in Appendix~\ref{sec:minimalstabdecomp}. Numerical results are presented in Figure~\ref{fig:strongsimscaling} and details can be found in Appendix~\ref{sec:strongsimnumericaldetails}.

The tensored decompositions from this reduction can undergo similar reductions due to the presence of equally-weighted stabilizer states that are Clifford-isomorphic to Bell states. However, the improvement is asymptotically smaller and can be numerically found to not lower the exponential factor by ten-thousandth (i.e. any improvement still rounds to \(2^{\sim 0.463 t}\)).

We note that finding such a reduction from the trivial tensored stabilizer decomposition of smaller decompositions is in general not expected. Indeed, a similar effort of using a \(k=3\) decomposition (such as Eq.~\ref{eq:3qubitTgatedecomp}) to infer a minimal stabilizer decomposition for \(k=6\) (such as Eq.~\ref{eq:6qubitdecomposition}) would require using more elaborate properties. This is why the \(k=6\) reduction in the past was actually found by a numerical search. Moreover, there is no indication in the literature that the reductions found at \(k=2\), \(k=3\) and \(k=6\) have a dependent relationship on each other, let alone an iterative relationship. As a result, there is no reason to believe beforehand that a scaling improvement can be found at \(k=12\). However, in the following sections, we will constructively prove that a reduction exists and argue that it can improved further.

\section{Single-Pauli Measurement Lower Bound}
\label{sec:singlepaulimeas}

In the next section we will restrict to the case of single-Pauli measurements to ease our formal analysis of finding a lower bound on the stabilizer rank. In this section, we describe why this is a useful endeavor when it is appropriately restricted.

As defined after Eq.~\ref{eq:probofapp}, \(\ket T = 2^{-1/2} (\ket 0 + \sqrt{i} \ket 1)\). Since the stabilizer rank of \(\ket T\), \(\chi_1 = 2\), this is a minimal stabilizer decomposition. Single-Pauli measurements, \(2^{-1} \bra{T}^{\otimes t}(I \pm P) \ket{T}^{\otimes t}\), where \(P \equiv \prod_{i=1}^n P_i\) for \(P_i\) a single-qubit Pauli, can be decomposed into the following linear combination of \(t\)-qubit stabilizer state inner products:
\begin{equation}
  \label{eq:vanillalineardecomp}
  2^{-t/2-1}\left(1 + \sqrt{i}^{|x|-|x'|} \sum_{\bs x,\, \bs x' \in \mathbb F_2^n} \expval{\bs x}{P}{\bs x'}\right),
\end{equation}
where \(|x|\) denotes the Hamming weight of \(x\) in binary. This linear combination consists of \(((\xi_1)^t)^2 = 2^{2t}\) \(t\)-qubit inner products. However, since \(P\) is separable into single-qubit Pauli operators, then so is the overall inner product. This permits the sum to be factored:
\begin{equation}
  2^{-t/2-1}\left(1 + \sqrt{i}^{|x|-|x'|} \prod_{i=1}^t \sum_{x,\, x' \in \mathbb F_2} \expval{x}{P_i}{x'}\right).
\end{equation}
Every set of \(x_1^2=4\) one-qubit inner products, \(\sum_{x,\, x' \in \mathbb F}\), can be calculated separately and the \(t\) results can be multiplied together to produce the full result. This tallies up to \(4t\) one-qubit inner products compared to Eq.~\ref{eq:vanillalineardecomp}'s \(2^{2t}\) \(t\)-qubit inner products. An exponential reduction!

Nevertheless, it is the first decomposition consisting of \(2^{2n}\) inner products that is of interest to use as a lower bound on general multi-Pauli measurement.

This is because multi-Pauli measurements, \(\prod_{i=1}^t (I + \sigma_i \mathbb P_i)\), where \(\sigma_i = \pm 1\) and \(\mathbb P_i \equiv \prod_{j=1}^t P_i\), no longer generally permit the factoring reduction possible with single-Pauli measurements since they no longer keep the initially separable \(T\)-gate tensored magic states separable.

As a result, the stabilizer decomposition that produces the lowest possible \emph{linear} combination of stabilizer state inner products for single-Pauli measurement corresponds to a more useful (i.e. tighter) lower bound on the stabilizer rank of general multi-qubit Pauli measurements. Therefore, here we will not allow for independent calculations of products of stabilizer state inner products when trying to produce such a lower bound.

The same sort of restriction will hold in the next section where we generalize from stabilizer state inner products to Gauss sums.

\section{Gauss Sum Rank Bound}

Gauss sums are proportional to the determinant of their covariance matrix:
\begin{eqnarray}
  \label{eq:Gausssum}
  G_m(\bs A, \bs v, c) &\equiv& \sum_{x_1, \ldots, x_k=0}^1 \exp \left[ \frac{\pi i}{2^m} (\bs x \bs A \bs x^T + 2 \bs v \bs x^T + c) \right] \nonumber\\
  &\propto& 2^{\det \bs A/2},
\end{eqnarray}
where \(2^m \bs A^{-1}\) has even diagonal entries. The proportionality constant in Eq.~\ref{eq:Gausssum} consists of an eighth root of unity. See Definition 2.1 and Remark 2.2 in~\cite{Fisher02} for a full identity.

Calculating these sums for \(n\) qubits requires Gaussian elimination to find the determinant of the \(k\times k\) covariance matrix \(\bs A\) with entries in \(\mathbb Z/ 2 \mathbb Z\) in Eq.~\ref{eq:Gausssum}. Therefore it scales as \(\mathcal O(k^3)\).

\(n\)-qubit stabilizer states can be expressed as a quadratic form defined on a \(m\)-dimensional affine space in the computational basis~\cite{Dehaene03,Nest08} where \(m \le n\):
\begin{equation}
  \ket{\Psi} = 2^{-\frac{m}{2}}\sum_{x \in \mathcal A_{G,h}} (-1)^{q(x)}i^{l(x)} \ket{x},
\end{equation}
where \(l(x)\) is a linear term, \(q(x) = \sum_{i\ne j} q_{ij} x_i x_j + c_i x_i\), where \(q_{ij}\), \(c_i \in \mathbb Z_2\). \(\mathcal A_{G,u}\) is an affine space, \(\mathcal A_{G,u} = \{ G u \oplus h | u \in \mathbb Z_2^n\}\) where \(\oplus\) is mod \(2\) addition, with \(G\) an \(m \times n\) matrix with entries in \(\mathbb Z_2\) and \(h \in \mathbb Z_2^n\). \(l(x)\) and \(q(x)\) are functions of \(x\) in terms of its coordinates (i.e. components in the affine space), \((x_1, \ldots, x_m)\) for \(x_i \in \{0,1\}\), whereas \(\ket{x}\) is expressed in the basis of the affine space, \(\ket{x} = \ket{h \oplus x_1 g_1 \oplus \ldots \oplus x_m g_m}\), in terms of \(n\) qubits.

As a result, stabilizer state inner products are Gauss sums~\cite{Bravyi16_2}. This association is injective; each stabilizer state inner product can be associated with a Gauss sum while any Gauss sum can be associated with many stabilizer state inner products. This flexibility can be used to algebraically lower the number of Gauss sums without bothering to associate them consistently with a state's stabilizer decomposition.

  Given a state \(\ket \Psi\) on \(n\) qubits and a set of projectors \(\{\hat \Pi_i\}_{i=1}^\mu\), choose a state decomposition \(\ket \Psi = \sum_{i=1}^\nu c_i \ket{\phi_i}\) (not necessarily consisting of stabilizer states) such that there exists some Clifford gate \(U_C\) and set of stabilizer states \(\{\ket{\varphi_i}\}_{i=1}^\nu\) and
  \begin{equation}
    \label{eq:stabstate_correspondence}
    \expval{\phi_j}{\hat \Pi_i}{\phi_k} = \expval{\varphi_j}{\hat \Pi_i \hat U_C}{\varphi_k}  \, \forall j,\, k \in \{1, \ldots, \nu\}.
  \end{equation}
  This establishes a one-to-one correspondence between the states \(\ket{\phi_i}\) and stabilizer states \(\ket{\varphi_i}\) that allows us to define a more general notion of a minimal number of Gauss sums:
\begin{definition}
  \label{def:GaussSumRank}
  We define the \emph{Gauss sum rank} for an \(n\)-qubit state \(\ket{\Psi}\), \(\xi \equiv \xi(\ket{\Psi}, \{\hat \Pi_i\}_{i=1}^\mu)\), to be the minimum number of non-zero unique Gauss sums with maximum dimension \(n\) considered over all such decompositions \(\{\phi_i\}\) and \(\{\varphi_i\}\) given by Eq.~\ref{eq:stabstate_correspondence}, where the cost of specifying the unique Gauss sums and their multiples is \(\mathcal O(n^3)\). Gauss sums are defined to be non-unique if they are term-wise identical in the computational basis.
\end{definition}
Note that this definition includes the case where \(\{\ket{\phi_i}\}_{i=1}^\nu\) are stabilizer states too.

While the definition of a stabilizer rank \(\chi(\Psi) \equiv \chi\) is only dependant on a state \(\Psi\), the definition of a Gauss sum rank \(\xi(\Psi, \Pi) \equiv \xi\) is dependent on a state \(\Psi\) and a set of observables \(\{\Pi_i\}_i\). \(\chi_k\) stabilizer states can produce from \(\chi_k\) to \(\chi_k^2\) non-zero inner products when evaluating the expectation value of some operator. It follows that \(\xi_k \le \chi_k^2 \, \forall \, \Pi\). 

For \(k\)-separable measurements \(\{\hat \Pi_i\}_i\), the Gauss sum always obeys the same trivial tensor bound as the stabilizer rank: \(\xi_t \le (\xi_k)^{t/k} \) where \(t\) is a multiple of \(k\). Similarly, the cost of simulation of \(\Pi_k\) scales the same: \(\xi_k\) stabilizer inner products, each requiring \(k^3\) operations to calculate, for a total cost of \(\xi_k k^3\) operations compared to \(\chi_k^2 k^3\) in the stabilizer decomposition~\cite{Bravyi16_2}. However, in the stabilizer case, the quadratic power of \(\chi_k\) can be decreased to a linear power by taking advantage of the fact that stabilizer states form a two-design~\cite{Dankert09}. This necessarily involves sampling random stabilizer states \(\{\ket{\psi_a}\}_{a\in \mathcal S}\) and scales with additional cost: \(L(\epsilon, p_f) \equiv\epsilon^{-2} \log(p_f^{-1})\) to obtain the quantum probability (Eq.~\ref{eq:probofapp}) with relative error \(\epsilon\) with probability at least \((1-p_f)\)~\cite{Bravyi16_1}. 
\begin{equation}
\label{eq:stabstatesamplingexpval}
  \expval{\Psi}{\hat \Pi}{\Psi} = \sum_{a=1}^{L(\epsilon, \,p_f)} \left( \left( \sum_{i=1}^{(\chi_k)^{t/k}} c_i \braket{\psi_a}{\phi_i} \right) \times \text{c.c.} \right)
\end{equation}

The result of this transformation is that only \(L(\epsilon, p_f) \chi_k k^3\) operations are necessary to evaluate the expectation value \(\Pi_k\) to relative error. A very similar quadratic reduction can be made for the Gauss sum rank cost for single-Pauli measurements, as we will see shortly, which will allow us to use it as a lower bound on the stabilizer rank cost.

If we restrict ourselves to single-Pauli measurements, to find an upper bound on the Gauss sum rank of the T gate magic state, we can exploit two properties that these states have:
\begin{enumerate}
\item they possess equiprobable stabilizer decompositions
  \begin{itemize}
  \item this produces optimal conditions for their expectation values to be repeated
  \end{itemize}
\item these stabilizer decompositions are taken to one another under Paulis
  \begin{itemize}
  \item as a result, the expectation values of \(\nu\) states only produces \(\nu\) non-zero Gauss sums instead of up to \(\nu^2\)
  \end{itemize}
\end{enumerate}
In particular, these are properties of stabilizer decompositions that can be written in terms of stabilizer states Clifford-isomorphic to computational basis states tensored with single-qubit stabilizer states on the Bloch sphere equator,
\begin{theorem}[Gauss Sum Rank of Single-Pauli]
\label{th:PauliExpectationScaling}
Consider a Clifford gate \(U_C\) such that 
\begin{equation}
  \hat U_C \ket{\Psi} = \sum_{i=0}^{\nu} c_i \prod_j^{m} \ket{\phi_{ij}} \prod_k^{n-m} \ket{\psi_{ik}},
\end{equation}
where \(\{\ket{\phi_{ij}}\}\) are orthogonal single-qubit stabilizer states, \(\{\ket{\psi_{ik}}\} \equiv \frac{1}{\sqrt{2}}(\ket 0 + i^{\mu_{ik}} \ket 1)\) are single-qubit stabilizer states on the Bloch sphere equator, \(c_i \in \mathbb R\) and \(\mu_{ik} \in \mathbb Z\). Then, for \(\hat \rho = \ketbra{\Psi}{\Psi}\) and \(\hat P\) a Pauli operator, the Gauss sum rank of \(\Tr \left[ \hat P \hat \rho \right]\) is at most \(\nu\) and the decomposition is made up of Gauss sums with equal absolute value magnitudes.
\end{theorem}
\begin{proof}
  It is clear that for \(\hat P = \hat I\), there are only \(\nu\) non-zero orthogonal stabilizer state inner products that are non-zero, namely the ones with the same bras and kets. Any \(\hat P \ne \hat I \) can be taken to act on the kets. The ket stabilizer states can be taken to the computational basis states by a Clifford operator, \(\hat U_C\). Clifford gates act on Paulis, \(\hat P\), by taking them to other Paulis, \(\hat P'\). Paulis act on computational basis states, \(\ket{\alpha_i}\), by taking them to other computational basis states, \(\ket{\alpha'_i}\). As a result, each ket stabilizer state, \(\ket{\phi_i}\) and \(\ket{\psi_j}\), which previously only had non-zero inner product with one bra stabilizer state, can still only have non-zero inner product with one bra stabilizer state:
  \begin{eqnarray}
    && \sum_{i,j}^{\nu} \prod_{k,l} \prod_{m,n}\expval{\phi_{ik},\psi_{im}}{\hat U^\dagger_C \hat P \hat U_C}{\phi_{jl},\psi_{jn}} c^*_i c_j \nonumber\\
    &=& \sum_{i,j}^{\nu} \prod_{k,l} \prod_{m,n}\expval{\phi_{ik},\psi_{im}}{\hat P'}{\phi_{jl},\psi_{jn}} c^*_i c_j \\
    &=& \sum_{i,j}^{\nu} \prod_{k,l} \prod_{m,n}\braket{\phi_{ik},\psi_{im}}{\phi'_{jl},\psi'_{jn}} c^*_i c_j
  \end{eqnarray}
Therefore, the total sum consists of only \(\nu\) non-zero inner products.

Paulis take stabilizer states that lie on the Bloch sphere equator to other stabilizer states on the Bloch sphere equator. As a result, the absolute value of the final inner product sum remains invariant under the Paulis.
\qed
\end{proof}
This was originally noticed in~\cite{Kocia18_2} (see Theorem 3) for odd-dimensional qudits, and is extended here to qubits.

As a result, it can also be shown that the Gauss sum rank cost for single-Pauli measurements is a lower bound on the stabilizer rank cost in the asymptotic limit:

\begin{theorem}[Asymptotic Bound on Stabilizer Rank]
\label{th:asymptoticlowerbound}
In the limit that \(t \rightarrow \infty\), the cost of evaluating single-Pauli measurement expectation values in terms of the Gauss sum rank (where independent calculations of products of stabilizer state inner products is not permitted as discussed in Section~\ref{sec:singlepaulimeas}) is \(\mathcal O (\xi_t)\) and lower bounds cost of evaluating the expectation values in terms of the stabilizer rank cost, which is \(\mathcal O (\chi_t)\). %, where both are calculated to a relative error \(\epsilon\) with failure probability \(p_f\).
\end{theorem}
\begin{proof}
We are interested in calculating \(\expval{\Psi}{(I+P)}{\Psi} = ||(I+P)\ket{\Psi}||^2\).
\begin{eqnarray}
  && \expval{\Psi}{(I+P)}{\Psi} \nonumber\\
&=& \expval{\Psi}{I}{\Psi} + \expval{\Psi}{P}{\Psi}\\
&=& \sum_{i=1}^{\le \xi} c_{i}^* c_{f(i)} \expval{\varphi_{i}}{I}{\varphi_{f(i)}} \\
&&+ \sum_{i=1}^{\le \xi} (c'_{i})^* c'_{f(i)} \expval{\varphi'_{i}}{P' U_C}{\varphi'_{f(i)}} \nonumber\\
%&=& \sum_{a=1}^L \left[ \sum_{i=1}^{\le \xi} c_{i}^* c_{f(i)} \expval{\varphi_{i}}{I}{\psi_a}\expval{\psi_a}{I}{\varphi_{f(i)}} \right.\\
%&& \left. + \sum_{i=1}^{\le \xi} (c'_{i})^* c'_{f(i)} \expval{\varphi'_{i}}{P' U_C}{\psi_a}\braket{\psi_a}{\varphi'_{f(i)}} \right] \nonumber
\end{eqnarray}
where \(f:\mathbb Z \rightarrow \mathbb Z\) is a map that associates the \(i\)th state with the one state it has a non-zero inner product with under \(\hat P\).% and \(\{\psi_j\}_j\) are random stabilizer states. In the last line the sum indexed by \(j\) is over \(L = \mathcal O(\epsilon^{-2} \log(p_f^{-1}))\) random stabilizer states for multiplicative error \(\epsilon\) and failure probability \(p_f\), which uses the two-design property of stabilizer states~\cite{Dankert09}.

This contains \(2 \xi_t\) diagonal Gauss sums.

Similarly, a decomposition of \(\ket{\Psi}\) into a minimal decomposition of \(\chi_t\) stabilizer states would produce \(L \chi_t\) Gauss sums above, for \(L = \mathcal O(\epsilon^2 \log p_f^{-1})\) for multiplicative error \(\epsilon\) and failure probability \(p_f\), using their two-design property~\cite{Dankert09}.

Since \(\mathcal O(\epsilon^{-2} \log(p_f^{-1}) t^3) \chi_t \underset{t\rightarrow \infty}{\rightarrow} \chi_t\) and the Gauss sum calculation above to relative error scales as \(\xi_t\), it follows that since the Gauss sum rank includes the stabilizer rank in its definition that the cost of evaluating the single-Pauli measurement expectation values in terms of the Gauss sum asymptotic rank \(\xi_t\) is a lower bound on the cost in terms of the stabilizer asymptotic rank \(\chi_t\).\qed
\end{proof}

The T gate magic state \(\ket{T^{\otimes k}}\) for \(k=1\) and \(k=2\) possesses such stabilizer decompositions and they are also minimal rank stabilizer decompositions.

For \(k=1\):
\begin{eqnarray}
\label{eq:1qubittgatemagicstate}
  && \Tr \left[ \hat P \hat \rho_{T} \right] \nonumber\\
  &=& \frac{1}{2} \left[ \left(\expval{0}{\hat P}{0} + \expval{1}{\hat P}{1}\right) \right. \\
  && \left.+  \left(\sqrt{i} \expval{0}{\hat P}{1} + \sqrt{-i} \expval{1}{\hat P}{0}\right) \right] \nonumber\\
  &=& \frac{\omega}{2} \left\{ \left[ \delta(\gamma + \delta ) + (-1)^{\beta}\delta(\gamma + \delta ) \right] \right. \nonumber\\
  && \left. + \left[ \sqrt{i} (-i)^{\delta} \delta(\alpha + \beta ) + \sqrt{-i} i^{\delta}\delta(\alpha + \beta ) \right] \right\} \nonumber
\end{eqnarray}
where \(\hat P = \omega(\alpha \hat I + \beta \hat Z + \gamma \hat X +\delta \hat Y)\), for \(\omega\) a fourth root of unity and \(\alpha, \, \beta, \, \gamma, \, \delta \in \{0,1\}\) such that \(\alpha + \beta + \gamma + \delta = 1\) (i.e. \(\hat P\) is a single-qubit Pauli operator). The Kronecker delta function arguments are to be considered modulo two here and below.

The Gauss sums are formed from a minimal stabilizer decomposition that are computational basis states: \(\ket{0}\) and \(\ket{1}\). Therefore, according to Theorem~\ref{th:PauliExpectationScaling}, for a given \(\hat P\), only two of the Gauss sums are non-zero. The stabilizer rank for \(k=1\) is also two. 

The last equality in Eq.~\ref{eq:1qubittgatemagicstate} shows that \(\Tr \left[ \hat P \hat \rho_T \right]\) consists of four terms that can each be written as Kronecker delta functions multiplied by phases. For the sake of dimensional analysis, these Kronecker delta functions should be considered in terms of their Gauss sum identities: \(\delta(x) = 2^{-1} \sum_{y \in \mathbb Z/2\mathbb Z} \exp (\pi i xz)\). Each line consists of a set of two Gauss sums that are chosen so that only such set is non-zero for a given \(P\).

A very similar decomposition can be made for \(k=2\). Setting \(\ket{\phi_1} = \frac{1}{\sqrt{2}} (\ket {00} + i\ket {11})\) and \(\ket{\phi_2} = \frac{1}{\sqrt{2}} ( \ket{01} + \ket{10} )\), such that \(\ket {T^{\otimes 2}} = (\ket{\phi_1} + \sqrt{i}\ket{\phi_2})/\sqrt{2}\), we can write
\begin{widetext}
\begin{eqnarray}
  \label{eq:2qubittgatemagicstate}
  &&  \Tr \left[ \hat P \hat \rho_{T^2} \right] \\
  &=& \frac{1}{2} \left[ \left(\expval{\phi_1}{\hat P}{\phi_1} + \expval{\phi_2}{\hat P}{\phi_2}\right) \right.\nonumber\\
  && \quad \left. + \sqrt{i} \expval{\phi_1}{\hat P}{\phi_2} + \left(\sqrt{-i} \expval{\phi_2}{\hat P}{\phi_1} \right) \right] \nonumber\\
  &=& \frac{\omega}{2} \left\{ \left[ i^{\delta_2-\gamma_1} G_1 \left( 2(\beta_1 + \beta_2 + \gamma_1 + \delta_2) \right) + (-1)^{\beta_1+\delta_1} i^{\delta_1+\delta_2} G_1 \left( 2(\beta_1 + \delta_1 + \beta_2 + \delta_2) \right) \right] \delta( \gamma_1+\delta_1 -\gamma_2-\delta_2 ) \right.\nonumber\\
  && \left. \quad + \left[ \sqrt{i} (-1)^{\beta_2+\gamma_2} i^{\delta_1 -\gamma_2 } (-i) G_1(1,\beta_2 + \gamma_2 + \beta_1 + \delta_1) + \sqrt{-i} i^{\delta_1 + \delta_2} G_1(1,\beta_1 + \beta_2 + \delta_1 + \delta_2) \right] \delta(\gamma_1 + \delta_1 - \gamma_2 - \delta_2 +1 ) \right\} \nonumber
\end{eqnarray}
\end{widetext}
for \(\hat P = \prod_{i=1}^k\omega(\alpha_i \hat I + \beta_i \hat Z_i + \gamma_i \hat X_i +\delta_i \hat Y_i)\), where \(\omega\) is a fourth root of unity and \(\alpha_i, \, \beta_i, \, \gamma_i, \, \delta_i \in \{0,1\}\) such that \(\alpha_i + \beta_i + \gamma_i + \delta_i = 1\) (i.e. \(\hat P\) is a Pauli operator). We will use this same definition for \(\hat P\) for \(k>2\).

The last equality of Eq.~\ref{eq:2qubittgatemagicstate} consists of \(4\) terms that can be written as two-dimensional Gauss sums (again, the Kronecker delta functions should be considered in their Gauss sum form for this dimensional analysis). Again, each line consists of a set of two Gauss sums that are chosen so that only one set is non-zero for a given \(\hat P\). Similarly as for \(k=1\), the stabilizer rank for \(k=2\) is also two. As before, the non-zero Gauss sums are formed from a minimal stabilizer decomposition: \(\ket{\phi_1}\) and \(\ket{\phi_2}\).

These two examples are particularly simple because the Gauss sums in their decomposition come from the inner products of stabilizer states that satisfy the constraints given by Theorem~\ref{th:PauliExpectationScaling}. In general, for \(k>2\), it has been found that stabilizer decompositions that saturate the stabilizer rank are not even orthogonal, let alone Clifford-isomorphic to computational basis states tensored with separable stabilizer states on the Bloch sphere equator~\cite{Bravyi16_2}. As discussed, the only known way to find these general non-orthogonal decompositions is by searching the prohibitively large stabilizer space.

Fortunately, a corollary can be made to Theorem~\ref{th:PauliExpectationScaling} that will allow us to keep finding novel Gauss sum reductions for \(k>2\):

\begin{corollary}[T Gate Magic State Gauss Sums]
\label{co:PauliExpectationScaling}
  Theorem~\ref{th:PauliExpectationScaling} holds if the single-qubit stabilizer states that lie on the Bloch sphere equator are replaced with single-qubit T gate magic states, where the resultant Gauss sums are greater in dimension by the number of the single-qubit T gate magic states.
\end{corollary}
\begin{proof}
The expectation value of \(P\), for \(P \in \{I, Z\}\), with respect to the T gate magic state, can be written as the off-diagonal expectation value of \(P\) with respect to two stabilizes states, or equivalently, the expectation value of \(P S\), for \(S = \text{diag}\{1,i\}\) the phase shift gate, with respect to a single stabilizer state:
\begin{eqnarray}
  \expval{T}{P}{T} &=& (\bra{0} + \bra{1}) (P) (\ket{0}+\ket{1})/2
\end{eqnarray}
The expectation value can be written in terms of a single stabilizer state and is a two-dimensional Gauss sum.

On the other hand, the expectation value of \(P\), for \(P \in \{X, Y\}\), with respect to the T gate magic state, can be written as the expectation value of \(\sqrt{i} P S\), for \(S = \text{diag}\{1,i\}\) the phase shift gate, with respect to a single stabilizer state:
\begin{eqnarray}
  \expval{T}{P}{T} &=& (\bra{0} + (-i) \bra{1}) (\sqrt{i} P S) (\ket{0}+ i\ket{1})/2.
\end{eqnarray}
This produces a diagonal two-dimensional Gauss sum.\qed
\end{proof}

Here we take the approach suggested by this Corollary. We use the decidedly non-minimal stabilizer decomposition at \(k=3\) from tensoring together the \(k=1\) and \(k=2\) decompositions, which have the benefit of consisting of stabilizer states that are Clifford-isomorphic to computational basis states. We will minimize the total number of Gauss sums by finding repetitions from the Gauss sums found by making the alternative decomposition in Corollary~\ref{co:PauliExpectationScaling}.

\section{Gauss Sum Scaling of Single-Pauli Measurement}

\subsection{\(k=3\)}

Due to the qubit permutation symmetry inherent in any tensored state, three different equiprobable stabilizer decompositions for \(k=3\) can be obtained from tensoring \(k=1\) and \(k=2\):
  \begin{eqnarray}
\label{eq:decomposition1}
    \ket{T}^{\otimes 3} &=& \frac{1}{2} C_{12} \left[ \ket{+'} \ket{00} + i\ket{+} \ket{11} \right.\\
                        &&\qquad \left. + \sqrt{i}\ket{+'} \ket{01} + \sqrt{i} \ket{+}\ket{10} \right] \nonumber\\
                        &=& \frac{1}{2} C_{23} \left[ \ket{0}\ket{+'}\ket{0} + i\ket{1}\ket{+}\ket{1} \right.\\
                        &&\qquad \left. + \sqrt{i}\ket{0}\ket{+}\ket{1} + \sqrt{i}\ket{1}\ket{+'}\ket{0}\right] \nonumber\\
                        &=& \frac{1}{2} C_{31} \left[ \ket{00}\ket{+'} + i\ket{11}\ket{+} \right.\\
                        &&\qquad \left. + \sqrt{i}\ket{01}\ket{+} + \sqrt{i}\ket{10}\ket{+'}\right] , \nonumber
  \end{eqnarray}
where \(\ket{+} \equiv (\ket{0} + \ket{1})/\sqrt{2}\) and \(\ket{+'} \equiv (\ket{0} + i \ket{1})/\sqrt{2}\).

As can be seen above, the four stabilizer states in these three decompositions are all Clifford-isomorphic to computational basis states tensored with single-qubit stabilizer states since their tensor components have this property. Therefore, by Theorem~\ref{th:PauliExpectationScaling}, they produce up to four three-dimensional Gauss sums of equal magnitude for the expectation value of any Pauli operator \(P\).

Similarly, the three-qubit T gate magic state can also be rewritten in three ways to be Clifford-isomorphic to computational basis states tensored with single-qubit T gate magic states:
\begin{eqnarray}
\label{eq:TgateMagicStateGaussSums1}
  \ket{T}^{\otimes 3} &=& \frac{1}{2} ( \ket{00T} {+} \sqrt{i}\ket{01T} {+} \sqrt{i}\ket{10T} {+} i\ket{11T})\\
  &=& \frac{1}{2} ( \ket{T00} + \sqrt{i}\ket{T01} + \sqrt{i}\ket{T10} + i\ket{T11})\\
\label{eq:TgateMagicStateGaussSums3}
  &=& \frac{1}{2} ( \ket{0T0} + \sqrt{i}\ket{0T1} + \sqrt{i}\ket{1T0} {+} i\ket{1T1}).
\end{eqnarray}

Therefore, by Corollary~\ref{co:PauliExpectationScaling}, \emph{given a Pauli}, this can be rewritten as four three-dimensional Gauss sums for its expectation value.

We will find the decompositions given by Eq.~\ref{eq:decomposition1}, Eq.~\ref{eq:TgateMagicStateGaussSums1} and Eq.~\ref{eq:TgateMagicStateGaussSums3} to be particularly useful. We can rewrite these terms parametrized by \(x\) and \(y\), and write out the sum over computational basis states in the Gauss sum or T gate magic state explicitly in terms of a third variable \(z\):
  \begin{eqnarray}
\label{eq:decomposition1_2}
    \ket{T}^{\otimes 3} &=& \frac{1}{2 \sqrt{2}} C_{12} \sum_{\substack{x,y,\\z=0}}^1 {i^{(x-1)^2(y+z)+xy/2+x(y-1)^2/2}}\\
                        &&\qquad \qquad \qquad \times \ket{4z{+}3x{+}y{+}2xy} \nonumber\\
\label{eq:TgateMagicStateGaussSums1_2}
                        &=& \frac{1}{2 \sqrt 2} \sum_{\substack{x,y,\\z=0}}^1 {i^{(x-1)^2(y)+xy/2+x(y-1)^2/2+z/2}}\\
                        &&\qquad \qquad \qquad \times \ket{z{+}6x{+}3y{-}5xy},  \nonumber\\
\label{eq:TgateMagicStateGaussSums3_2}
                        &=& \frac{1}{2 \sqrt 2} \sum_{\substack{x,y,\\z=0}}^1 {i^{(x-1)^2(y)+xy/2+x(y-1)^2/2+z/2}}\\
                        &&\qquad \qquad \qquad \times \ket{4z{+}3x{+}y{+}2xy},  \nonumber
  \end{eqnarray}
where Eq.~\ref{eq:decomposition1}, Eq.~\ref{eq:TgateMagicStateGaussSums1} and Eq.~\ref{eq:TgateMagicStateGaussSums3} are rewritten as Eq.~\ref{eq:decomposition1_2}, Eq.~\ref{eq:TgateMagicStateGaussSums1_2} and Eq.~\ref{eq:TgateMagicStateGaussSums3_2}, respectively.

A straightforward decomposition for \(k=3\) by tensoring together the \(k=1\) and \(k=2\) decompositions, corresponding to using the decomposition given by Eq.~\ref{eq:decomposition1_2}, can be found in Appendix~\ref{sec:3qubitdecomp} in Eq.~\ref{eq:3qubittgate1}. This produces \(\xi_3^2 = (\xi_1 \xi_2)^2 = 16\) Gauss sums by the trivial tensor bound. In Eq.~\ref{eq:3qubittgate1}, these sets are organized into four sets. As before, only one set is non-zero for a given \(\hat P\). 

By inspection, we can reduce the number of unique Gauss sums for any given Pauli by instead using the decomposition given by Eq.~\ref{eq:TgateMagicStateGaussSums1} for the first and second sets of four Gauss sums, and the decomposition given by Eq.~\ref{eq:TgateMagicStateGaussSums3} for the fourth set of four Gauss sums. This produces repeated Gauss sums that we can take account of with factors of \(2\):
\begin{widetext}
\begin{eqnarray}
  \label{eq:3qubittgatemagicstate}
  && \Tr \left[ \hat P \hat \rho_{T^3} \right]  \\
  &=& \frac{\omega}{8} \left\{\sum_{y=0}^{1} \sum_{x=0}^{y(\gamma_1+\gamma_2)+(y+1)(\gamma_1+\delta_2)} 2^{1-y(\gamma_1-\gamma_2)^2-(y-1)^2(\gamma_1-\delta_2)^2} \left[ \mathcal G_1(x,y) + \mathcal G_2(x,y) \right] \right. \nonumber\\
  && \quad + 2 \left. \sum_{x=0}^{1+\beta_3} \sum_{y=0}^{1+\alpha_3} \mathcal G_3(x,y) + \sum_{y=0}^1 \sum_{x=y(\delta_1 + \delta_2)}^{1+y(\gamma_1 + \gamma_2)} 2^{(y-1)^2[(x)^2(\gamma_1+\delta_2) + (x-1)^2(\gamma_2+\delta_1)]} \mathcal G_4(x,y) \right\} \nonumber
\end{eqnarray}
where
\begin{eqnarray}
  \mathcal G_1(x,y) &=& e^{\frac{\pi i}{2} ((\delta_2-\gamma_1)(y+1)^2 + 2(\beta_1 + \beta_2 + \gamma_1 + \delta_2)x(y+1)^2 + 2(\delta_1 + \delta_2 + \beta_1 + \beta_2) x y + (2\beta_1+3\delta_1+\delta_2)y} G_1( 2 \beta_3)\\
                    && \qquad \times  \delta( \gamma_1+\delta_1 -\gamma_2-\delta_2) \delta(\gamma_3 + \delta_3), \nonumber\\
  \mathcal G_2(x,y) &=& e^{\frac{\pi i}{2} ((\delta_2-\gamma_1)(y+1)^2 + 2(\beta_1 + \beta_2 + \gamma_1 + \delta_2)x(y+1)^2 + 2(\delta_1 + \delta_2 + \beta_1 + \beta_2) x y + (2\beta_1+3\delta_1+\delta_2)y} \sqrt{-i} G_1(\gamma_3+\delta_3) \\
                    && \qquad \times \delta( \gamma_1+\delta_1 -\gamma_2-\delta_2) \delta(\alpha_3 + \beta_3), \nonumber\\
  \mathcal G_3(x,y) &=& \sqrt{-i} e^{\frac{\pi i}{2} (2\beta_3 y + (x+1)^2(\delta_1 + \gamma_2 + 2 \beta_2 ) + x(\delta_1 + \delta_2))} G_1(1,\beta_1 + \beta_2 + \delta_1 + (x+1) \gamma_2 + x \delta_2) \\
                    && \qquad \times \delta(\gamma_1 + \delta_1 - \gamma_2 - \delta_2 +1) \delta(\gamma_3 + \delta_3), \nonumber\\
                    && \text{and} \nonumber\\
  \mathcal G_4(x,y) &=& (-i) e^{\frac{\pi i}{2} (y(\delta_1+\delta_2) + (y+1)^2 (\delta_1+\gamma_2+2\beta_2)+x^2+2(\beta_1+\beta_2+\delta_1+\gamma_2(y+1)+\delta_2y)x)}\\
                    && \qquad \times G_1(0 , \gamma_3 + \delta_3 + \beta_1+\beta_2+\delta_1+\gamma_2(y+1)+\delta_2y + x ) \delta(\gamma_1 + \delta_1 - \gamma_2 - \delta_2 +1) \delta(\alpha_3 + \beta_3). \nonumber
\end{eqnarray}
\end{widetext}
The sum ranges are to be considered modulo two here and below.

Like the naively tensored expression in Appendix~\ref{sec:3qubitdecomp}, Eq.~\ref{eq:3qubittgatemagicstate} still consists of four sets of four Gauss sums, where only one set is non-zero for a given \(\hat P\). However, now, of the four Gauss sums in each set, at least one is repeated. Hence, at most only three non-zero Gauss sums have to be evaluated in each set for a given \(\hat P\). More precisely, the first two lines correspond to sets containing either one or two repeated Gauss sums. The third line corresponds to a set containing two or four repeated Gauss sums. The fourth line corresponds to a set containing one repeated Gauss sum.

As a consequence, the number of non-zero Gauss sums that must be summed over to evaluate \(\Tr \left[ \hat P \hat \rho_{T^3} \right]\) for all \(\hat P\) is \(\xi_3 = 3\). This matches the number found by finding a stabilizer decomposition of \(\ket \Psi\) that saturates its stabilizer rank: \(\xi_3 \le \chi_3 = 3\)~\cite{Bravyi16_2}.

As an aside, in order for us to meaningfully declare that we have fewer Gauss sums, it is important that the cost of the parts of the equations that determine the appropriate factor of 2 scale the same as the phase prefactors did when all the Gauss sums were included. Similarly, the three Guass sums that are summed over must also scale that same. This means that the cost must scale cubically with the number of qubits, the cost of Gaussian elimination in evaluating Gauss sums. In particular, for Eq.~\ref{eq:3qubittgatemagicstate}, this means that the arguments to the delta functions in the power of \(2\) and the equation determining the range of the sums over \(x\) and \(y\) must contain a number of terms that scales cubically with the Pauli operator parameters (\(\alpha_i\), etc.).

In practice we actually find that the prefactor scales linearly instead of cubically with the number of qubits. This is because the states are not entangled and we are only taking a single-Pauli measurement. This means there is no need for the expensive row-reduction step in Gaussian elimination.

\subsection{\(k=6\)}

While we found a reduction in the decomposition for \(k=3\) formed by tensoring the \(k=1\) and \(k=2\) decompositions by inspection, we proceed to now find reductions in \(k=6\) and \(k=12\) over the trivial tensor bound by making use of a property observe in the largest set of Gauss sums that are non-zero for a given \(P\) in reduced \(k=3\) decomposition.

A naive decomposition of the six-qubit T gate magic state can be written by tensoring the decomposition for the three-qubit T gate magic state with itself. This is shown in Eq.~\ref{eq:6qubittgate1} in Appendix~\ref{sec:6qubitdecomp}.

The lines with the largest number of Gauss sums contain nine of them. We can parametrize these sets in terms of \(2\times 2\) matrices of Gauss sums that are permuted by \(\sigma_{\hat P',i}\) and \(\sigma'_{\hat P',i}\): permutations dependent on \(\hat P'\) and \(i\), the index of the sets of nine Gauss sums in Eq.~\ref{eq:6qubittgate1}:
\begin{eqnarray}
\label{eq:permutingGausssums}
  \mathcal G_i = \left(\begin{array}{cc} \sigma_{\hat P',i}(G_{i_1},& G_{i_2})\\ \sigma'_{\hat P',i}(G_{i_3},& G_{i_2})\end{array}\right),
\end{eqnarray}
\(\sigma\) permutes the entries in each row. There are three unique Gauss sums in these two sets. We denote the repeated Gauss sum by \(G_{i_2}\) and the other two by \(G_{i_1}\) and \(G_{i_3}\). In Eq.~\ref{eq:6qubittgate1}, the Gauss sums are indexed by \(x\)-values row-wise and by \(y\)-values column-wise.

Due to the prefactor equation (in terms of \(x\) and \(y\)) being a parametrization of the two-dimensional Gauss sums with arguments that \(0\) or \(1\), it follows that the Gauss sums can only advance by quarters in the unit circle. Hence, those that are not repeated must either be negatives of each other or the same: \(G_1 = \pm G_3\). We can use this to simplify the number of unique Gauss sums when we take the tensor product of the sets with themselves.

It follows that
\begin{eqnarray}
 && \mathcal G_i \otimes \mathcal G_j \\
&=& \left( \begin{array}{cc} \begin{array}{cc} G_{i_1} \otimes \sigma(G_{j_1}, & G_{j_2})\\ G_{i_1} \otimes \sigma'(G_{j_3},& G_{j_2}) \end{array} & \begin{array}{cc} G_{i_2} \otimes \sigma(G_{j_1}, & G_{j_2})\\ G_{i_2} \otimes \sigma'(G_{j_3},& G_{j_2}) \end{array} \\ \begin{array}{cc} G_{i_3} \otimes \sigma(G_{j_1}, & G_{j_2})\\ G_{i_3} \otimes \sigma'(G_{j_3},& G_{j_2}) \end{array} & \begin{array}{cc} G_{i_2} \otimes \sigma(G_{j_1}, & G_{j_2})\\ G_{i_2} \otimes \sigma'(G_{j_3},& G_{j_2}) \end{array}
\end{array} \right). \nonumber
\end{eqnarray}
The right half of the matrix clearly contains only three unique Gauss sums since it only involves tensor products between \(G_{i_2}\) and the three unique Gauss sums \(G_{j_{\{1,2,3\}}}\). On the other hand, since the left half involves the tensor products between the two Gauss sums \(G_{i_{\{1,3\}}}\) with \(G_{j_{\{1,2,3\}}}\), it seems to contain six unique Gauss sums. In total, this produces nine Gauss sums, as expected from tensoring two sets of three Gauss sums. However, using the fact that \(G_1 = \pm G_3\), it can be shown that two pairs in the left hand-side are equal to each other: \(G_{i_1} \otimes G_{j_1} = G_{i_3} \otimes G_{j_3}\) and \(G_{i_1} \otimes G_{j_3} = G_{i_3} \otimes G_{j_1}\).

\begin{figure}[ht]
  \includegraphics[scale=0.55]{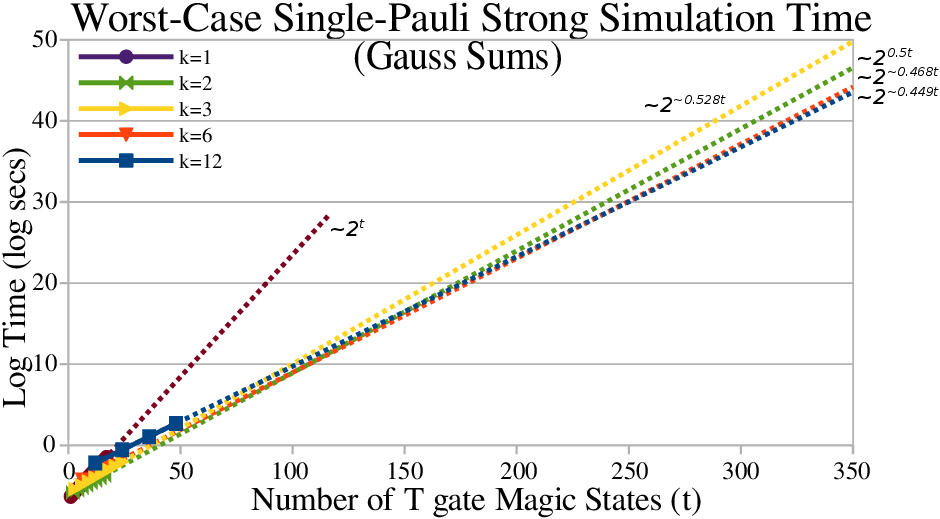}
  \caption{A plot of the computational time to calculate random single-Pauli measurements using the \(k=1\), \(k=2\), \(k=3\), \(k=6\) and \(k=12\) Gauss sum rank. Fitting is shown to \((\xi_1)^t = 2^t\), \((\xi_2)^{t/2} = 2^{0.5 t}\), \((\xi_6)^{t/6} = 2^{\sim 0.468 t}\), and \((\xi_{12})^{t/12} = 2^{\sim 0.449t}\), with respect to \(t\), the number of \(T\) gate magic states in Eq.~\ref{eq:probofapp} for \(t\) a multiple of \(1\), \(2\), \(6\) and \(12\), respectively.}
  \label{fig:scaling}
\end{figure}

Therefore, the nine Gauss sums reduce to seven unique Gauss sums. Every Gauss sums is repeated twice except for \(G_{i_2} \otimes G_{j_2}\), which is repeated four times. This additional repetition can be described by a factor of two taken to the power of the products of the two arguments in the factors of two in Eq.~\ref{eq:6qubittgate2}, thereby conserving the linear cost of calculating this factor. Otherwise, the unique Gauss sums can be picked out by adding a factor of \(x\) or \(y\) (\(\Mod 2\)) to the range of the second two sums. This simpification can be seen by comparing the first line of Eq.~\ref{eq:6qubittgate1} compared to Eq.~\ref{eq:6qubittgate2}:
\begin{widetext}
\begin{eqnarray}
  && \sum_{y=0}^{1} \sum_{x=0}^{\substack{y(\gamma_1+\gamma_2)+\\(y+1)(\gamma_1+\delta_2)}} \sum_{y'=0}^{1} \sum_{x'=0}^{\substack{y'(\gamma_4+\gamma_5)+\\(y'+1)(\gamma_4+\delta_5)}} 2^{1-y(\gamma_1-\gamma_2)^2-(y-1)^2(\gamma_1-\delta_2)^2} 2^{1-y'(\gamma_4-\gamma_5)^2-(y'-1)^2(\gamma_4-\delta_5)^2} \\
  && \quad \qquad \qquad \qquad \qquad \qquad \qquad \qquad \times \left[ \mathcal G_1(x,y) + \mathcal G_2(x,y) \right] \left[ \mathcal G_1(x',y') + \mathcal G_2(x',y') \right] \nonumber\\
  &=& \sum_{y=0}^{1} \sum_{x=0}^{\substack{y(\gamma_1+\gamma_2)+\\(y+1)(\gamma_1+\delta_2)}} \sum_{y'=0+(\gamma_4+\delta_5)x}^{1-(\gamma_4+\gamma_5)x} \sum_{x'=0}^{\substack{y'(\gamma_4+\gamma_5)+\\(y'+1)(\gamma_4+\delta_5)}} 2\times 2^{\left[1-y(\gamma_1-\gamma_2)^2-(y-1)^2(\gamma_1-\delta_2)^2\right]\left[1-y'(\gamma_4-\gamma_5)^2-(y'-1)^2(\gamma_4-\delta_5)^2\right]} \\
  && \quad \qquad \qquad \qquad \qquad \qquad \qquad \qquad \qquad \qquad \times \left[ \mathcal G_1(x,y) + \mathcal G_2(x,y) \right] \left[ \mathcal G_1(x',y') + \mathcal G_2(x',y') \right]. \nonumber
\end{eqnarray}
\end{widetext}

This simplification reduces all the largest sets of Gauss sums, which contain nine, down to seven.

If \(G_{i_k} = G_{j_k} \forall k\), which is not generally true, one might imagine two more pairs can be equated since \(G_1 \otimes G_2 = G_2 \otimes G_1\) and \(G_3 \otimes G_2 = G_2 \otimes G_3\). However, these cannot be singled out with arguments that are bilinear in the Greek letters and \(x\) or \(y\) and do not sum up to repetitions that number a power of \(2\) and hence would require additional coefficients that would thereby increase the number of Gauss sums in linear combination.

Therefore, in the worst case, the six-qubit T gate magic state expectation value requires calculating \(\xi_6=7\) Gauss sums, the same as \(\chi_6\).

\subsection{\(k=12\)}

This same simplification can be iteratively used again in the largest sets obtained from tensoring together the decomposition for the six-qubit T gate with itself, reducing the number from \(7 \times 7 = 49\) Gauss sums down to \(31\). As a result, the formerly next-largest set consisting of \(6 \times 7=42\) Gauss sums, becomes the largest set and determines the worst-case scaling of \((\xi_{12})^{t/12} = 2^{\sim 0.449 t}\). See Eq.~\ref{eq:12qubittgate} in Appendix~\ref{sec:12qubitdecomp} for the resulting expression after this reduction.

We numerically validated these Gauss sum decompositions, and thereby their scaling, for both representative and random cases of single-Pauli measurements, by comparing measurement probabilities with those obtained from a naive Hilbert space treatment (see Fig.~\ref{fig:scaling}). We make this code available at https://s3miclassical.com/gitweb/.

\section{Discussion}

The Gauss sum rank for single-Pauli measurements on the T gate magic state seems to be a tight lower bound on the T gate magic state's stabilizer rank for \(k=1\), \(k=2\), \(k=3\) and \(k=6\), where it is known. If this trend holds, then it follows that the improved scaling stabilizer rank we found from the twelve-qubit tensor component of \(2^{\sim 0.463 t}\) can be lowered further to at least \(2^{\sim 0.449 t}\). This corresponds to a stabilizer rank for twelve qubits of at most \(42\) instead of \(47\).

A note of caution: the Gauss sum reductions we have found are technically upper bounds on the Gauss sum rank --- we have not proven that they are the lowest reductions possible. If it is possible to lower the Gauss sum rank for single-Pauli measurements on the T gate magic state further, the Gauss sum rank will no longer be a tight lower bound on the stabilizer rank. However, we see no way to lower the Gauss sum rank decompositions further within the contraints of Definition~\ref{def:GaussSumRank} of the Gauss sum rank.

\begin{table*}[t!]
  \begin{tabular}{p{30pt}|p{30pt}|p{30pt}|p{30pt}|p{30pt}|p{30pt}|p{30pt}|p{30pt}|p{30pt}|p{30pt}|p{30pt}|p{30pt}|p{30pt}|p{30pt}|p{30pt}}
    \hline
    \(k\) & $1$ & $2$ & $3$ & $4$ & $5$ & $6$ & $7$ & $8$ & $9$ & $10$ & $11$ & $12$ & $13$ & $14$ \\
    \hline \hline
    \multicolumn{8}{l}{qubit:} \\
    \hline
    \(\chi_k\) & $2$ & $2$ & $3$ & $4$ & $6$ & $7$ & $12$ & \multicolumn{4}{c}{\cellcolor[gray]{0.8}\emph{inaccessible to Monte Carlo}}& \multicolumn{1}{c}{\cellcolor[gray]{0.8}$\le 47$} & \multicolumn{2}{c}{\cellcolor[gray]{0.8}{}}\\
    \hline
    \(\xi_k\) & $2$& $2$ & $3$ & $4$ & $6$ & $7$ & $12$ & $\le 14$ & $\le 21$ & $\le 28$ & {$\le 42$} & {$\le 42$} & $\le 84$& $\le 84$\\
    \hline
    \(\xi_k^{t/k}\) & {$2^t$}& {$2^{0.5 t}$}& {$2^{\sim 0.528t}$}& & & {${2^{\sim 0.468t}}$} & & & & & & $2^{\sim 0.449 t}$ & \\
    \hline
    \multicolumn{8}{l}{qutrit:} \\
    \hline
    \(\chi_k\) & $3$ & $3$ & $8$ & \multicolumn{11}{c}{\cellcolor[gray]{0.8}\emph{inaccessible to Monte Carlo}}\\
    \hline
    \(\xi_k\) & $3$& $3$ & $8$ & $9$ & $24$ & $24$ & $\le 72$ & $72$ & $\le 216$ & $216$ & {$\le 486$} & {$\le 486$} & $\le 1458$& $\le 1458$\\
    \hline
    \(\xi_k^{t/k}\) & {$3^t$}& {$3^{0.5 t}$}& {$3^{\sim 0.631t}$}& & & {$3^{\sim 0.482t}$} & & & & & {$3^{\sim 0.512 t}$} & {$3^{\sim 0.469 t}$} & & \\
    \hline
  \end{tabular}
  \caption{Upper bound of qubit and qutrit \(T\) gate magic state stabilizer ranks \(\chi_k\) are tabulated and compared to qubit and qutrit single-Pauli measurement Gauss sum ranks \(\xi_k\) along with their tensor upper bounds (\(\chi_k^{t/k}\) and \(\xi_k^{t/k}\) respectively). The Gauss sum rank for qutrits is for the Wigner formulation of the expectation value~\cite{Kocia20}. The reductions in the qubit scaling for \(k=1\), \(k=2\), \(k=3\), \(k=6\), and \(k=12\) are observed for qutrits as well. \((\chi_k)^{t/k}\) is only listed at the \(k\) values at which there is a reduction over the trivial tensor bound.}
  \label{tab:Gausssumresults}
\end{table*}

The deviation of the search for the minimum number of Gauss sums at \(t>2\) was inspired by the approach taken in the Wigner-Weyl-Moyal (WWM) formalism for odd-prime-dimensional qudits where a similar reduction was found at \(t=1\), \(t=2\), \(t=3\), \(t=6\) and \(t=12\) (see Table~\ref{tab:Gausssumresults}). This parallel can most easily be seen by examining the Gauss sums for \(k=3\) produced by the straightforward tensoring of the \(k=1\) and \(k=2\) decompositions, given by Eq.~\ref{eq:decomposition1_2}. The transformations to using the alternative decompositions correspond to rewriting the first and second sets of four Gauss sums as Gauss sums over \(y\) instead of \(z\), leaving the third set of four Gauss sums as is, and rewriting the fourth set of Gauss sums over \(X \equiv x+y\) instead of \(z\). These transformations produce linear combinations of Gauss sums that, once multiplied by their coefficients, are always non-negative if they are not cancelled. Thus they can be interpreted as classical probability distributions. These transformations are similar to performing controlled-not Clifford transformations on the ``intermediate'' phase space coordinates in the qutrit WWM formalism to obtain reductions in the number of quadratic Gauss sums~\cite{Kocia18_2,Kocia20}.

Definition~\ref{def:GaussSumRank} was inspired by the fact that in the odd-dimensional WWM formalism quadratic Gauss sums are the computational primitive as well, but are not restricted to be stabilizer inner products of a consistent minimal stabilizer decomposition.

For even \(d\), the WWM formalism requires three Grassmann generators~\cite{Kocia17_2}. All the results presented here could have equally well been presented in the WWM formalism: the stabilizer states would correspond to phase space and the Gauss sums presented here would be replaced by Grassman Gauss (``Grauss'') sums~\cite{Kocia17_2}. Each such inner product is a non-contextual (order \(\hbar^0\)) term with a contextual (order \(\hbar^1\)) phase, just like for the Feynman path integral expanded in powers of \(\hbar\) to produce the van Vleck-Morette-Gutzwiller propagator~\cite{Van28,Morette51,Gutzwiller67}.

This relationship raises an interesting open question: does the existence of \(\xi\) Gauss sums for \(\Tr(\hat \Pi \hat \rho)\) necessarily mean that the stabilizer rank of \(\ket{\Psi}\) is \(\chi=\xi\)? We have seen that this is the case here, at least for \(k=1\) to \(k=7\), but we leave open the question of whether there is a stabilizer decomposition for \(k=12\) that matches our minimal Gauss sum decomposition.

A related open question is whether the Gauss sum rank results shown here for single-Pauli measurements can be extended to multi-Pauli measurements. We leave this for future work.

Previous studies~\cite{Gross07} and the WWM formalism have established a correspondence between coherent states and stabilizer states, and Gaussian integrals and Gauss sums. If this correspondence truly holds, then one would also expect there to be such a relationship between minimal coherent state decompositions, which correspond to stationary phase approximation points, and minimal stabilizer state decompositions, and hence that \(\xi_k = \chi_k\).

There are, of course, special properties that the magic states possess that may make them more fit to have \(\xi_k = \chi_k\). For instance, for the qubit case the \(T\) gate magic state has equal magnitude coefficients for its \(k=1\) and \(k=2\) minimal stabilizer decompositions that are eighth roots of unity, an important property that we took advantage of in finding their minimal number of unique Gauss sums when traced over. Are these properties necessary for \(\xi_k\) to be equal to \(\chi_k\) and thereby is this equivalence somehow restricted to this subset of states? We leave these questions for future study.

Another natural question arises: Can we use the same property we used when tensoring together \(k=3\) to produce our reduction at \(k=6\) and again at \(k=12\) to produce a reduction at \(k=24\)? It turns out that the same method cannot be used because the second largest group of Gauss sums in \(k=3\) becomes the dominant one now at \(k=12\): the set of Gauss sums consisting of \emph{two} pairs of repeating Gauss sums. When tensored with the sets containing one repeated Gauss sum at \(k=6\) and two repeated Gauss sums at \(k=12\), these produce sets with \(42\) unique Gauss sums. However, because of the first tensor product containing two pairs of repeated Gauss sums, the important property that is used in our reductions does not hold -- there are not necessarily two Gauss sums that are negatives of each other. Hence, further reductions require finding some other property of this new largest set after tensoring.

\section{Conclusion}

In conclusion, we improve on the best known stabilizer rank of \(\xi_t \le 2^{\sim 0.468 t}\) down to \(\xi_t \le 2^{\sim 0.463 t}\) for multi-Pauli measurement expectation values. We derive an alternative asymptotic lower bound for single-Pauli measurement expectation value that relies on grouping equivalent Gauss sums formed by inner products between orthogonal stabilizer states Clifford-isomorphic to a computational basis and particular single-qubit states called the Gauss sum rank \(\xi_k\). Using this method we find a new asymptotic scaling bound of \(\xi_t \le 2^{\sim 0.449 t}\) from the twelve-qubit decomposition. This work suggests that this twelve-qubit stabilizer decomposition can be further improved to match this single-Pauli bound since the bound is tight for lower \(t\)-counts.

\noindent---\newline
This material is based upon work supported by the U.S. Department of Energy, Office of Science, Office of Advanced Scientific Computing Research, under the Quantum Computing Application Teams program. 
Sandia National Laboratories is a multimission laboratory managed and operated by National Technology \& Engineering Solutions of Sandia, LLC, a wholly owned subsidiary of Honeywell International Inc., for the U.S. Department of Energy's National Nuclear Security Administration under contract {DE-NA0003525}. This paper describes objective technical results and analysis. Any subjective views or opinions that might be expressed in the paper do not necessarily represent the views of the U.S. Department of Energy or the United States Government. SAND2020-14255 O

\acknowledgments
The author would like to thank Mohan Sarovar for fruitful discussions in the process of this research.

\bibliography{biblio}{}
\bibliographystyle{unsrt}
\appendix

\section{Minimal Stabilizer Decompositions}
\label{sec:minimalstabdecomp}

A decomposition of \(\ket{T^{\otimes 6}}\) into seven stabilizer states is: 
\begin{eqnarray}
\label{eq:6qubitdecomposition}
\ket{T^{\otimes 6}} &=&c_{b_{60}} \ket{\Psi_{b_{60}}} + c_{b_{66}} \ket{\Psi_{b_{66}}} +  c_{e_6} \ket{\Psi_{e_6}} + c_{o_6} \ket{\Psi_{o_6}} \nonumber\\
&& + c_{k_6} \ket{\Psi_{k_6}} + c_{\phi'} \ket{\Psi_{\phi'}} + c_{\phi''} \ket{\Psi_{\phi''}},
\end{eqnarray}
where
\begin{eqnarray}
  c_{b_{60}} &=& (-16+12\sqrt{2}) \cos^6\left(\pi/8\right) e^{\pi i 6/8} 2^3,\\
  c_{b_{66}} &=& (96 - 68\sqrt{2}) \cos^6\left(\pi/8\right) e^{\pi i 6/8} 2^3,\\
  c_{e_6} &=& (10 - 7 \sqrt{2}) \cos^6\left(\pi/8\right) e^{\pi i 6/8} 2^{\frac{5}{2}},\\
  c_{o_6} &=& (-14 + 10 \sqrt{2}) \cos^6\left(\pi/8\right) e^{\pi i 6/8} 2^{\frac{5}{2}},\\
  c_{k_6} &=& (7 - 5 \sqrt{2}) \cos^6\left(\pi/8\right) e^{\pi i 6/8} 2^{3},\\
  c_{\phi'} &=& c_{\phi''} = (10 - 7 \sqrt{2}) \cos^6\left(\pi/8\right) e^{\pi i 6/8} 2^{\frac{5}{2}},
\end{eqnarray}
and
\begin{eqnarray}
  \ket{\Psi_{b_{60}}} &=& \frac{1}{8} \sum_{x=0}^{2^6-1} \ket{x},\\
  \ket{\Psi_{b_{66}}} &=& \frac{1}{8} \sum_{x=0}^{2^6-1} e^{\frac{\pi i}{4} (4 x+4)} \ket{x},\\
  \ket{\Psi_{e_6}} &=& \frac{1}{4\sqrt 2} \sum_{\substack{x=0\\x'=Gx+h}}^{2^5-1} e^{\frac{\pi i}{4} (x^TJ_6x + 4)} \ket{x'},\\
  \ket{\Psi_{e_6}} &=& \frac{1}{4\sqrt 2} \sum_{\substack{x=0\\x'=Gx+h}}^{2^5-1} e^{\frac{\pi i}{4} (x^TJ_6x + Dx+ 4)} \ket{x'},\\
  \ket{\Psi_{k_6}} &=& \frac{1}{\sqrt 2} \sum_{\substack{x=0\\x'=G_k x+h}}^{1} e^{\frac{\pi i}{4} (2x + 6)} \ket{x'},\\
  \ket{\Psi_{\phi'}} &=& \frac{1}{4\sqrt 2} \sum_{\substack{x=0\\x'=Gx+h}}^{2^5-1} e^{\frac{\pi i}{4} (x^TJ_{\phi'}x)} \ket{x'},\\
  \ket{\Psi_{\phi''}} &=& \frac{1}{4\sqrt 2} \sum_{\substack{x=0\\x'=Gx+h}}^{2^5-1} e^{\frac{\pi i}{4} (x^TJ_{\phi''}x)} \ket{x'},
\end{eqnarray}
where
\begin{equation*}
  J_6 = \left( \begin{array}{ccccc} 0 & 4 & 4& 4& 4\\ 0& 0& 4& 4& 4\\ 0& 0& 0& 4& 4\\ 0& 0& 0& 0& 4\\ 0& 0& 0& 0& 0 \end{array} \right), \hskip 10pt
  J_{\phi'} = \left( \begin{array}{ccccc} 0 & 4 & 0& 0& 4\\ 0& 0& 4& 0& 0\\ 0& 0& 0& 4& 0\\ 0& 0& 0& 0& 4\\ 0& 0& 0& 0& 0 \end{array} \right),
\end{equation*}
\begin{equation}
  J_{\phi''} = \left( \begin{array}{ccccc} 0 & 0 & 4& 4& 0\\ 0& 0& 0& 4& 4\\ 0& 0& 0& 0& 4\\ 0& 0& 0& 0& 0\\ 0& 0& 0& 0& 0 \end{array} \right),
\end{equation}
\begin{equation}
  D = \left( \begin{array}{c} 4 \\ 4 \\ 4\\ 4\\ 4 \end{array}\right),
\end{equation}
\begin{equation}
  G = \left(\begin{array}{cccccc} 1 & 1 & 0 & 0 & 0 & 0\\ 1 & 0 & 1 & 0 & 0 & 0\\ 1 & 0 & 0 & 1 & 0 & 0\\ 1 & 0 & 0 & 0 & 1 & 0\\ 1 & 0 & 0 & 0 & 0 & 1 \end{array} \right), \hskip 5pt
  G_k = \left(\begin{array}{cccccc} 1 & 1 & 1 & 1 & 1 & 1 \end{array} \right),
\end{equation}
and 
\begin{equation}
h = \left(\begin{array}{c} 1 \\ 0 \\ 0 \\ 0 \\0 \\ 0 \end{array}\right),\hskip 10pt h_o = \left(\begin{array}{c} 0 \\ 0 \\ 0 \\ 0 \\ 0 \\ 0 \end{array}\right),\hskip 10pt h_k = \left(\begin{array}{c} 1 \\ 1 \\ 1 \\ 1 \\ 1 \\ 1 \end{array}\right).
\end{equation}

A decomposition of \(\ket{T^{\otimes 12}}\) into forty-seven stabilizer states can be found by tensoring the above decomposition and replacing the terms \({c_{b_{60}} c_{b_{66}} \ket{\Psi_{b_{60}} \otimes \Psi_{b_{66}}} + c_{b_{66}} c_{b_{60}} \ket{\Psi_{b_{66}} \otimes \Psi_{b_{60}}}}\) and \({c_{e_6} c_{o_6} \ket{\Psi_{e_6} \otimes \Psi_{o_6}} + c_{o_6} c_{e_6} \ket{\Psi_{o_6} \otimes \Psi_{e_6}}}\) with \(c_b \ket{\Psi_{b_{60} b_{66}}}\) and \(c_6 \ket{\Psi_{e_6 o_6}}\), respectively, where
\begin{eqnarray}
  \ket{\Psi_{b_{60} b_{66}}} &=& \frac{1}{32\sqrt 2} \sum_{\substack{x=0\\x'=Gx+h_b}}^{2^{11}-1} e^{\frac{\pi i}{4} (Dx+4)} \ket{x'},\\
  \ket{\Psi_{e_6 o_6}} &=& \frac{1}{32 \sqrt 2} \sum_{\substack{x=0\\x'=Gx+h_6}}^{2^{11}-1} e^{\frac{\pi i}{4} (x^TJx)} \ket{x'},
\end{eqnarray}
\begin{eqnarray}
  c_b &=& 2c_{b_{60}} c_{b_{66}}\\
  c_6 &=& 2c_{e_6} c_{o_6},
\end{eqnarray}
\begin{equation}
  \{D_i\}_{i=1,\ldots,11} = \begin{cases} 4 & \text{if}\,\, i>5,\\ 0 & \text{otherwise}, \end{cases}
\end{equation}
\begin{equation}
  \{J_{ij}\}_{i,j=1,\ldots,11} = \begin{cases} 4 & \text{if}\,\, j>i,\\ 0 & \text{otherwise}, \end{cases}
\end{equation}
\begin{equation}
  \{G_{ij}\}_{\substack{i=1,\ldots,11\\j=1,\ldots,12}} = \begin{cases} 1 & \text{if}\,\, j=i+1 \,\, \text{or} \, \, j=1, \\ 0 & \text{otherwise}, \end{cases}
\end{equation}
\begin{equation}
  \{{h_b}_i\}_{i=1,\ldots,12} = \delta_{i,1} \,\, \text{and}\,\,
  \{{h_6}_i\}_{i=1,\ldots,12} = 0.
\end{equation}

\section{Strong Simulation Numerical Simulation Details}
\label{sec:strongsimnumericaldetails}

Numerical implementation follows the pseudocode proved in the appendices of~\cite{Bravyi16_1}. It consists of modules EXPONENTIALSUM(), which performs diagonalization of the quadratic form in the exponential ans scales as \(\mathcal O(n^3)\) for \(n\) qubits, SHRINK() and SHRINKSTAR(), which shrink the affine space spanned by a stabilizer state by one dimension, INNERPRODUCT(), which finds the common affine space shared by two stabilizer states using SHRINK() iteratively if necessary and calculates their inner product using EXPONENTIALSUM(), EXTEND(), which increases the dimensions of the affine space spanned by a stabilizer state, MEASUREPAULI(), which performes a Pauli measurement on a stabilizer state taking it to another stabilizer state using SHRINK() and EXTEND() if necessary, and RANDOMSTABILIZERSTATE(), which generates Haar random stabilizer states and scales as \(\mathcal O(n^3)\). The primary formalism that governs the affine space are its periodic boundary conditions; vectors in the affine space have entries in \(\mathbb Z/ 2 \mathbb Z\). As a result, there is not well-defined notion of adjoint, but dual or reciprocal vectors are well-defined; \(a\) and \(b\) are conjugate if \(a \cdot b = 1\) where the dot product is performed mod \(2\). Therefore, for instance, to determine whether two affine spaces intersect, inner products with the basis vectors of one of the spaces must be taken with the dual vectors to the basis vectors of the other space, etc. Code and further information can be found at https://s3miclassical.com/gitweb/.

\section{Three-Qubit Decomposition}
\label{sec:3qubitdecomp}

A stabilizer decomposition for \(k=3\) that saturates its stabilizer rank \(\chi_3 = 3\) is:
\begin{equation}
\label{eq:3qubitTgatedecomp}
  \ket{T^{\otimes 3}} = c_1 \ket{\psi_1} + c_2 \ket{\psi_2} + c_3 \ket{\psi_3},
\end{equation}
where
\begin{eqnarray}
c_1 &=& -\frac{1-i}{4}(-1-i+\sqrt{2})\sqrt{-i},\\
c_2 &=& -\frac{1+i}{4}(1-i+\sqrt{2})\sqrt{i},\\
c_3 &=& -\frac{+1+i}{4}(-1.0+i+\sqrt{2})\sqrt{i},
\end{eqnarray}
and
\begin{eqnarray}
\ket{\psi_1} &=& \frac{1}{\sqrt{2}}(\ket{011} + i \ket{100}),\\
\ket{\psi_2} &=& \frac{1}{2\sqrt{2}}(i\ket{000}-\ket{001}-\ket{010}-i\ket{011}\nonumber\\
             && +i\ket{100}-\ket{101}-\ket{110}-i\ket{111}),\\
\ket{\psi_3} &=& \frac{1}{2\sqrt{2}}(i\ket{000}+\ket{001}+\ket{010}+i\ket{011}\nonumber\\
             &&+i\ket{100}-\ket{101}-\ket{110}+i\ket{111}).
\end{eqnarray}

Selecting the first decomposition for the expectation value produces the following expression:
\begin{eqnarray}
\label{eq:3qubittgate1}
  && \Tr \left[ \hat P \hat \rho_{T^3} \right]  \\
  &=& \frac{\omega}{4} \left\{ \left[ i^{\delta_2-\gamma_1} G_1 \left( 2(\beta_1 + \beta_2 + \gamma_1 + \delta_2) \right) \right. \right.\\
  && \qquad \left. \left. + (-1)^{\beta_1+\delta_1} i^{\delta_1+\delta_2} G_1 \left( 2(\beta_1 + \delta_1 + \beta_2 + \delta_2) \right) \right] \right.\\
  && \qquad \times \left[ 1 + (-1)^{\beta_3} \right] \delta( \gamma_1+\delta_1 -\gamma_2-\delta_2 ) \delta(\gamma_3 + \delta_3 ) \nonumber\\
  && \quad + \left[ i^{\delta_2-\gamma_1} G_1 \left( 2(\beta_1 + \beta_2 + \gamma_1 + \delta_2) \right) \right.\\
  &&\qquad \left. + (-1)^{\beta_1+\delta_1} i^{\delta_1+\delta_2} G_1 \left( 2(\beta_1 + \delta_1 + \beta_2 + \delta_2) \right) \right] \\
  && \qquad \times \left[ \sqrt{i} (-1)^{\delta} i^{\delta}  + \sqrt{-i} i^{\delta} \right] \delta( \gamma_1+\delta_1 -\gamma_2-\delta_2 ) \delta(\alpha_3 + \beta_3 ) \nonumber\\
  && \quad + \left[ \sqrt{i} (-1)^{\beta_2+\gamma_2} i^{\delta_1 -\gamma_2 } (-i) G_1(1,\beta_2 + \gamma_2 + \beta_1 + \delta_1) \right. \nonumber\\
  && \qquad \left. + \sqrt{-i} i^{\delta_1 + \delta_2} G_1(1,\beta_1 + \beta_2 + \delta_1 + \delta_2) \right] \left[ 1 + (-1)^{\beta_3} \right]\nonumber\\
  && \qquad \times \delta(\gamma_1 + \delta_1 - \gamma_2 - \delta_2 +1 ) \delta(\gamma_3 + \delta_3 ) \nonumber\\
  && \quad + \left[ \sqrt{i} (-1)^{\beta_2+\gamma_2} i^{\delta_1 -\gamma_2 } (-i) G_1(1,\beta_2 + \gamma_2 + \beta_1 + \delta_1) \right.\\
  && \qquad \left. + \sqrt{-i} i^{\delta_1 + \delta_2} G_1(1,\beta_1 + \beta_2 + \delta_1 + \delta_2) \right] \\
  && \qquad \times \left[ \sqrt{i} (-1)^{\delta} i^{\delta}  + \sqrt{-i} i^{\delta} \right] \delta(\gamma_1 + \delta_1 - \gamma_2 - \delta_2 +1 ) \\
  && \qquad \times \left. \delta(\alpha_3 + \beta_3 ) \right\},\nonumber
\end{eqnarray}
where one line is non-zero, as can be seen by inspecting the delta functions on the Pauli parameters \(\alpha\), \(\beta\), \(\gamma\) and \(\delta\).

Let us index the four Gauss sums in each of the sets by \(x,\, y \in \mathbb Z/ 2 \mathbb Z\), where \(x\) indexes the sum over the Gauss sums in the first factor of Eq.~\ref{eq:3qubittgate1} and \(y\) indexes over the sum in the second factor, and let us consider the Gauss sums to be over a third degree of freedom, \(z \in \mathbb Z /2 \mathbb Z\). %The sums over the Gauss sum terms can be parametrized into a product of sums over \(x\) and \(y\) where \(x\) indexes the terms in the first square brackets on every line and the \(y\) indexes the terms in the second square brackets:

\begin{widetext}
\begin{eqnarray}
  && \Tr \left[ \hat P \hat \rho_{T^3} \right]  \\
  &=& \frac{\omega}{4} \left[ e^{\frac{\pi i}{2} ((\delta_2 - \gamma_1)(y + 1)^2 + 2(\beta_1 + \beta_2 + \gamma_1 + \delta_2) x (y + 1)^2 + 2(\delta_1 + \delta_2 + \beta_1 + \beta_2) x y + (2\beta_1 + 3\delta_1 + \delta_2)y)} e^{\pi i \beta_3 z} \delta( \gamma_1+\delta_1 -\gamma_2-\delta_2 ) \delta(\gamma_3 + \delta_3 ) \right. \nonumber\\
  && \quad + \sqrt{-i} e^{\frac{\pi i}{2} ((\delta_2 - \gamma_1)(y + 1)^2 + 2 (\beta_1 + \beta_2 + \gamma_1 + \delta_2) x (y + 1)^2 + 2 (\delta_1 + \delta_2 + \beta_1 + \beta_2) x y + (2\beta_1 + 3\delta_1 + \delta_2) y)}e^{\frac{\pi i}{2}(\gamma_3 + \delta_3) z^2} \delta( \gamma_1+\delta_1 -\gamma_2-\delta_2 ) \delta(\alpha_3 + \beta_3 ) \nonumber\\
  && \quad + \sqrt{-i} e^{\frac{\pi i}{2} (2 \beta_3 z + y (\delta_1 + \delta_2) + (y + 1)^2 (\delta_1 + \gamma_2 + 2 \beta_2))} e^{\frac{\pi i}{2}(x^2 + 2 (\beta_1 + \beta_2 + \delta_1 + \gamma_2 (y + 1) + \delta_2 y)x)} \delta(\gamma_1 + \delta_1 - \gamma_2 - \delta_2 +1 ) \delta(\gamma_3 + \delta_3 ) \nonumber\\
  && \quad \left. (-i) e^{\frac{\pi i}{2} (y (\delta_1 + \delta_2) + (y + 1)^2 (\delta_1 + \gamma_2 + 2 \beta_2))} e^{\frac{\pi i}{2} (x^2 + 2 (\beta_1 + \beta_2 + \delta_1 + \gamma_2 (y + 1) + \delta_2 y) x)} e^{\frac{\pi i}{2}(\gamma_3 + \delta_3) z} \delta(\gamma_1 + \delta_1 - \gamma_2 - \delta_2 +1 ) \delta(\alpha_3 + \beta_3 ) \right],\nonumber
\end{eqnarray}
\end{widetext}

By inspection, we can rewrite the first and second sets of four Gauss sums as Gauss sums over \(y\) instead of \(z\). This corresponds to using the decomposition given by Eq.~\ref{eq:TgateMagicStateGaussSums1}. The third set of four Gauss sums can be left as is, which corresponds to Eq.~\ref{eq:decomposition1}. The fourth set of Gauss sums can be rewritten as Gauss sums over \(X \equiv x+y\) to produce sets of Gauss sums, which corresponds to Eq.~\ref{eq:TgateMagicStateGaussSums3}. This produces repeated Gauss sums that we can take account of with factors of \(2\):

Inspection reveals that the four non-zero Gauss sums consist of at most three unique Gauss sums for any Pauli since at least one repeats. The identity of the repeating Gauss sum is straightforward to parametrize if we rewrite the four Gauss sums such that every Gauss sum that is not cancelled is real-valued. This must be possible to do since the total expectation value is real-valued. This produces Eq.~\ref{eq:3qubittgatemagicstate} and the repeating Gauss sum can now be excluded from the sum by appropriate linear arguments to the sum range over \(x\) and \(y\).

\section{Six-Qubit Decomposition}
\label{sec:6qubitdecomp}

An optimal stabilizer decomposition for the six-qubit T gate magic state is given in Appendix~\ref{sec:minimalstabdecomp}.

\begin{widetext}
\begin{eqnarray}
\label{eq:6qubittgate1}
  && \Tr \left[ \hat P \hat \rho_{T^6} \right]  \\
  &=& \frac{\omega}{8} \left\{\sum_{y=0}^{1} \sum_{x=0}^{y(\gamma_1+\gamma_2)+(y+1)(\gamma_1+\delta_2)} \sum_{y'=0}^{1} \sum_{x'=0}^{y'(\gamma_4+\gamma_5)+(y'+1)(\gamma_4+\delta_5)} 2^{1-y(\gamma_1-\gamma_2)^2-(y-1)^2(\gamma_1-\delta_2)^2} 2^{1-y'(\gamma_4-\gamma_5)^2-(y'-1)^2(\gamma_4-\delta_5)^2} \right. \nonumber\\
  && \qquad \qquad \qquad \qquad \qquad \qquad \qquad \qquad \qquad \qquad \qquad \qquad \times \left[ \mathcal G_1(x,y) + \mathcal G_2(x,y) \right] \left[ \mathcal G_1(x',y') + \mathcal G_2(x',y') \right] \nonumber\\
  && \qquad + 2 \sum_{y=0}^{1} \sum_{x=0}^{y(\gamma_1+\gamma_2)+(y+1)(\gamma_1+\delta_2)} \sum_{x'=0}^{1+\beta_6} \sum_{y'=0}^{1+\alpha_6} 2^{1-y'(\gamma_4-\gamma_5)^2-(y'-1)^2(\gamma_4-\delta_5)^2} \left[ \mathcal G_1(x,y) + \mathcal G_2(x,y) \right]  \mathcal G_3(x',y') \nonumber\\
  && \qquad + \sum_{y=0}^{1} \sum_{x=0}^{y(\gamma_1+\gamma_2)+(y+1)(\gamma_1+\delta_2)} \sum_{y'=0}^1 \sum_{x'=y'(\delta_4 + \delta_5)}^{1+y'(\gamma_4 + \gamma_5)} 2^{1-y(\gamma_1-\gamma_2)^2-(y-1)^2(\gamma_1-\delta_2)^2} 2^{(y'-1)^2[(x')^2(\gamma_4+\delta_5) + (x'-1)^2(\gamma_5+\delta_4)]} \nonumber\\
  && \qquad \qquad \qquad \qquad \qquad \qquad \qquad \qquad \qquad \qquad \times \left[ \mathcal G_1(x,y) + \mathcal G_2(x,y) \right] \mathcal G_4(x',y') \nonumber\\
  && \qquad + 2 \sum_{x=0}^{1+\beta_3} \sum_{y=0}^{1+\alpha_3} \sum_{y'=0}^{1} \sum_{x'=0}^{y'(\gamma_4+\gamma_5)+(y'+1)(\gamma_4+\delta_5)} 2^{1-y'(\gamma_4-\gamma_5)^2-(y'-1)^2(\gamma_4-\delta_5)^2} \mathcal G_3(x,y) \left[ \mathcal G_1(x',y') + \mathcal G_2(x',y') \right]  \nonumber\\
  && \qquad + 4 \sum_{x=0}^{1+\beta_3} \sum_{y=0}^{1+\alpha_3} \sum_{x'=0}^{1+\beta_6} \sum_{y'=0}^{1+\alpha_6} \mathcal G_3(x,y) \mathcal G_3(x',y') \nonumber\\
  && \qquad + 2 \sum_{x=0}^{1+\beta_3} \sum_{y=0}^{1+\alpha_3} \sum_{y'=0}^1 \sum_{x'=y'(\delta_4 + \delta_5)}^{1+y'(\gamma_4 + \gamma_5)} 2^{(y'-1)^2[(x')^2(\gamma_4+\delta_5) + (x'-1)^2(\gamma_5+\delta_4)]} \mathcal G_3(x,y) \mathcal G_4(x',y') \nonumber\\
  && \qquad + \sum_{y=0}^1 \sum_{x=y(\delta_1 + \delta_2)}^{1+y(\gamma_1 + \gamma_2)} \sum_{y'=0}^{1} \sum_{x'=0}^{y'(\gamma_4+\gamma_5)+(y'+1)(\gamma_4+\delta_5)} 2^{(y-1)^2[(x)^2(\gamma_1+\delta_2) + (x-1)^2(\gamma_2+\delta_1)]} 2^{1-y'(\gamma_4-\gamma_5)^2-(y'-1)^2(\gamma_4-\delta_5)^2} \nonumber\\
  && \qquad \qquad \qquad \qquad \qquad \qquad \qquad \qquad \qquad \qquad \times \mathcal G_4(x,y) \left[ \mathcal G_1(x',y') + \mathcal G_2(x',y') \right] \nonumber\\
  && \qquad + 2 \sum_{y=0}^1 \sum_{x=y(\delta_1 + \delta_2)}^{1+y(\gamma_1 + \gamma_2)} \sum_{x'=0}^{1+\beta_6} \sum_{y'=0}^{1+\alpha_6} 2^{(y-1)^2[(x)^2(\gamma_1+\delta_2) + (x-1)^2(\gamma_2+\delta_1)]} \mathcal G_4(x,y) \mathcal G_3(x',y') \nonumber\\
  && \qquad + \sum_{y=0}^1 \sum_{x=y(\delta_1 + \delta_2)}^{1+y(\gamma_1 + \gamma_2)} \sum_{y'=0}^1 \sum_{x'=y'(\delta_4 + \delta_5)}^{1+y'(\gamma_4 + \gamma_5)} 2^{(y-1)^2[(x)^2(\gamma_1+\delta_2) + (x-1)^2(\gamma_2+\delta_1)]} 2^{(y'-1)^2[(x')^2(\gamma_4+\delta_5) + (x'-1)^2(\gamma_5+\delta_4)]} \nonumber\\
  && \qquad \qquad \qquad \qquad \qquad \qquad \qquad \qquad \left. \times \mathcal G_4(x,y) \mathcal G_4(x',y') \right\}\nonumber\\
\end{eqnarray}
\end{widetext}

The nine lines consists of sets of Gauss sums numbering nine, six, nine, six, four, six, nine, six, and nine, respectively. These are sets \emph{of sets} of Gauss sums, with the prior defined sets containing containing four, two, two, two, one, one, two, one, and one, respectively, smaller sets. As always, the smallest sets of Gauss sums are so chosen so that only one is non-zero for a given \(\hat P_k\).

The sets of nine Gauss sums can be reduced due to the simplification described after Eq.~\ref{eq:permutingGausssums}. This simplification reduces them down to seven Gauss sums and produces the following expression:
\begin{widetext}
\begin{eqnarray}
\label{eq:6qubittgate2}
  && \Tr \left[ \hat P \hat \rho_{T^6} \right]  \\
  &=& \frac{\omega}{8} \left\{\sum_{y=0}^{1} \sum_{x=0}^{y(\gamma_1+\gamma_2)+(y+1)(\gamma_1+\delta_2)} \sum_{y'=0+(\gamma_4+\delta_5)x}^{1-(\gamma_4+\gamma_5)x} \sum_{x'=0}^{y'(\gamma_4+\gamma_5)+(y'+1)(\gamma_4+\delta_5)} 2\times 2^{\left[1-y(\gamma_1-\gamma_2)^2-(y-1)^2(\gamma_1-\delta_2)^2\right]\left[1-y'(\gamma_4-\gamma_5)^2-(y'-1)^2(\gamma_4-\delta_5)^2\right]} \right. \nonumber\\
  && \qquad \qquad \qquad \qquad \qquad \qquad \qquad \qquad \qquad \qquad \qquad \qquad \times \left[ \mathcal G_1(x,y) + \mathcal G_2(x,y) \right] \left[ \mathcal G_1(x',y') + \mathcal G_2(x',y') \right] \nonumber\\
  && \qquad + 2 \sum_{y=0}^{1} \sum_{x=0}^{y(\gamma_1+\gamma_2)+(y+1)(\gamma_1+\delta_2)} \sum_{x'=0}^{1+\beta_6} \sum_{y'=0}^{1+\alpha_6} 2^{1-y'(\gamma_4-\gamma_5)^2-(y'-1)^2(\gamma_4-\delta_5)^2} \left[ \mathcal G_1(x,y) + \mathcal G_2(x,y) \right]  \mathcal G_3(x',y') \nonumber\\
  && \qquad + \sum_{y=0}^{1} \sum_{x=0}^{y(\gamma_1+\gamma_2)+(y+1)(\gamma_1+\delta_2)} \sum_{y'=x}^1 \sum_{x'=y'(\delta_4 + \delta_5)+x}^{1+y'(\gamma_4 + \gamma_5)+x} 2\times 2^{\left[1-y(\gamma_1-\gamma_2)^2-(y-1)^2(\gamma_1-\delta_2)^2\right]\left[(y'-1)^2[(x')^2(\gamma_4+\delta_5) + (x'-1)^2(\gamma_5+\delta_4)]\right]} \nonumber\\
  && \qquad \qquad \qquad \qquad \qquad \qquad \qquad \qquad \qquad \qquad \times \left[ \mathcal G_1(x,y) + \mathcal G_2(x,y) \right] \mathcal G_4(x',y') \nonumber\\
  && \qquad + 2 \sum_{x=0}^{1+\beta_3} \sum_{y=0}^{1+\alpha_3} \sum_{y'=0}^{1} \sum_{x'=0}^{y'(\gamma_4+\gamma_5)+(y'+1)(\gamma_4+\delta_5)} 2^{1-y'(\gamma_4-\gamma_5)^2-(y'-1)^2(\gamma_4-\delta_5)^2} \mathcal G_3(x,y) \left[ \mathcal G_1(x',y') + \mathcal G_2(x',y') \right]  \nonumber\\
  && \qquad + 4 \sum_{x=0}^{1+\beta_3} \sum_{y=0}^{1+\alpha_3} \sum_{x'=0}^{1+\beta_6} \sum_{y'=0}^{1+\alpha_6} \mathcal G_3(x,y) \mathcal G_3(x',y') \nonumber\\
  && \qquad + 2 \sum_{x=0}^{1+\beta_3} \sum_{y=0}^{1+\alpha_3} \sum_{y'=0}^1 \sum_{x'=y'(\delta_4 + \delta_5)}^{1+y'(\gamma_4 + \gamma_5)} 2^{(y'-1)^2[(x')^2(\gamma_4+\delta_5) + (x'-1)^2(\gamma_5+\delta_4)]} \mathcal G_3(x,y) \mathcal G_4(x',y') \nonumber\\
  && \qquad + \sum_{y=0}^1 \sum_{x=y(\delta_1 + \delta_2)}^{1+y(\gamma_1 + \gamma_2)} \sum_{y'=0+(\gamma_4+\delta_5)y}^{1-(\gamma_4+\gamma_5)y} \sum_{x'=0}^{y'(\gamma_4+\gamma_5)+(y'+1)(\gamma_4+\delta_5)} 2 \times 2^{\left[(y-1)^2[(x)^2(\gamma_1+\delta_2) + (x-1)^2(\gamma_2+\delta_1)]\right]\left[1-y'(\gamma_4-\gamma_5)^2-(y'-1)^2(\gamma_4-\delta_5)^2\right]} \nonumber\\
  && \qquad \qquad \qquad \qquad \qquad \qquad \qquad \qquad \qquad \qquad \times \mathcal G_4(x,y) \left[ \mathcal G_1(x',y') + \mathcal G_2(x',y') \right] \nonumber\\
  && \qquad + 2 \sum_{y=0}^1 \sum_{x=y(\delta_1 + \delta_2)}^{1+y(\gamma_1 + \gamma_2)} \sum_{x'=0}^{1+\beta_6} \sum_{y'=0}^{1+\alpha_6} 2^{(y-1)^2[(x)^2(\gamma_1+\delta_2) + (x-1)^2(\gamma_2+\delta_1)]} \mathcal G_4(x,y) \mathcal G_3(x',y') \nonumber\\
  && \qquad + \sum_{y=0}^1 \sum_{x=y(\delta_1 + \delta_2)}^{1+y(\gamma_1 + \gamma_2)} \sum_{y'=y}^1 \sum_{x'=y'(\delta_4 + \delta_5)+y}^{1+y'(\gamma_4 + \gamma_5)+y} 2 \times 2^{\left[(y-1)^2[(x)^2(\gamma_1+\delta_2) + (x-1)^2(\gamma_2+\delta_1)]\right]\left[(y'-1)^2[(x')^2(\gamma_4+\delta_5) + (x'-1)^2(\gamma_5+\delta_4)]\right]} \nonumber\\
  && \qquad \qquad \qquad \qquad \qquad \qquad \qquad \qquad \left. \times \mathcal G_4(x,y) \mathcal G_4(x',y') \right\}\nonumber\\
\end{eqnarray}
\end{widetext}

The largest set of Gauss sums that are non-zero for a given \(\hat P_k\) contains seven Gauss sums. This determines the worst-case scaling of \((\xi_{6})^{t/6} = 2^{\sim 0.468 t}\).

\section{Twelve-Qubit Decomposition}
\label{sec:12qubitdecomp}
\begin{widetext}
The largest sets can be reduced from forty-nine down to thirty-one Gauss sums by following the same simplification strategy employed on the six qubit cases. This can be seen by examining the first line of Eq.~\ref{eq:12qubittgate} below:
\begin{eqnarray}
\label{eq:12qubitsimplification1}
  && 8 \sum_{y=0}^{1} \sum_{x=0}^{\substack{y(\gamma_1+\gamma_2)+\\(y+1)(\gamma_1+\delta_2)}} \sum_{y'=(\gamma_4+\delta_5)x}^{1+(\gamma_4+\gamma_5)x} \sum_{x'=0}^{\substack{y'(\gamma_4+\gamma_5)+\\(y'+1)(\gamma_4+\delta_5)}} \sum_{y''=0}^{1} \sum_{x''=0}^{\substack{y''(\gamma_7+\gamma_8)+\\(y''+1)(\gamma_7+\delta_8)}} \sum_{y'''=(\gamma_{10}+\delta_{11})x''}^{1+(\gamma_{10}+\gamma_{11})x''} \sum_{x'''=0}^{\substack{y'''(\gamma_{10}+\gamma_{11})+\\(y'''+1)(\gamma_{10}+\delta_{11})}} \\
  && \times 2^{\left[1-y(\gamma_1-\gamma_2)^2-(y-1)^2(\gamma_1-\delta_2)^2\right]\left[1-y'(\gamma_4-\gamma_5)^2-(y'-1)^2(\gamma_4-\delta_5)^2\right]} 2^{\left[1-y''(\gamma_7-\gamma_8)^2-(y''-1)^2(\gamma_7-\delta_8)^2\right]\left[1-y'''(\gamma_{10}-\gamma_{11})^2-(y'''-1)^2(\gamma_{10}-\delta_{11})^2\right]} \nonumber\\
  && \times \left[ \mathcal G_1(x,y) + \mathcal G_2(x,y) \right] \left[ \mathcal G_1(x',y') + \mathcal G_2(x',y') \right] \left[ \mathcal G_1(x'',y'') + \mathcal G_2(x'',y'') \right] \left[ \mathcal G_1(x''',y''') + \mathcal G_2(x''',y''') \right] \nonumber\\
\label{eq:12qubitsimplification2}
  &=& 8 \sum_{y=0}^{1} \sum_{x=y(\beta_1+\alpha_2+\beta_2)}^{\substack{y(\beta_1+\alpha_2+\beta_2)+\\y(\gamma_1+\gamma_2)+\\(y+1)(\gamma_1+\delta_2)}} \sum_{y'=(\gamma_4+\delta_5)x}^{1-(\gamma_4+\gamma_5)x} \sum_{x'=y'(\beta_4+\alpha_5+\beta_5)}^{\substack{y'(\beta_4+\alpha_5+\beta_5)+\\y'(\gamma_4+\gamma_5)+\\(y'+1)(\gamma_4+\delta_5)}} \sum_{\substack{y''=\\(\gamma_7+\delta_8)(x+x')}}^{\substack{1+\\(\gamma_7+\gamma_8)(x+x')}} \sum_{x''=y''(\beta_7+\alpha_8+\beta_8)}^{\substack{y''(\beta_7+\alpha_8+\beta_8)+\\y''(\gamma_7+\gamma_8)+\\(y''+1)(\gamma_7+\delta_8)}} \sum_{\substack{y'''=\\(\gamma_{10}+\delta_{11})(x+x'+x'')}}^{\substack{1+\\(\gamma_{10}+\gamma_{11})(x+x'+x'')}} \sum_{x'''=y'''(\beta_{10}+\alpha_{11}+\beta_{12})}^{\substack{y'''(\beta_{10}+\alpha_{11}+\beta_{12})+\\y'''(\gamma_{10}+\gamma_{11})+\\(y'''+1)(\gamma_{10}+\delta_{11})}}  \\
  && \times 2^{\left[1-y(\gamma_1-\gamma_2)^2-(y-1)^2(\gamma_1-\delta_2)^2\right]\left[1-y'(\gamma_4-\gamma_5)^2-(y'-1)^2(\gamma_4-\delta_5)^2\right]\left[1-y''(\gamma_7-\gamma_8)^2-(y''-1)^2(\gamma_7-\delta_8)^2\right]\left[1-y'''(\gamma_{10}-\gamma_{11})^2-(y'''-1)^2(\gamma_{10}-\delta_{11})^2\right]} \nonumber\\
  && \times \left[ \mathcal G_1(x,y) + \mathcal G_2(x,y) \right] \left[ \mathcal G_1(x',y') + \mathcal G_2(x',y') \right] \left[ \mathcal G_1(x'',y'') + \mathcal G_2(x'',y'') \right] \left[ \mathcal G_1(x''',y''') + \mathcal G_2(x''',y''') \right] \nonumber
\end{eqnarray}
Notice that there exists a factor of two whose argument is simply the product of the arguments in Eq.~\ref{eq:12qubitsimplification2}'s powers of two. Furthermore, the domains of the second half of the sum in Eq.~\ref{eq:12qubitsimplification2} have an added factor of \(x+x'\) compared to Eq.~\ref{eq:12qubitsimplification1}. These are the same modifications that were made to reduce the number of Gauss sums in the largest set for the six-qubit \(T\) gate magic state.

The full equation follows:
{\tiny
\begin{eqnarray}
\label{eq:12qubittgate}
  && \Tr \left[ \hat P \hat \rho_{T^{12}} \right] \\
\label{eq:12qubittgate1}
  &=& \frac{\omega}{64} \left\{\sum_{y=0}^{1} \sum_{x=y(\beta_1+\alpha_2+\beta_2)}^{\substack{y(\beta_1+\alpha_2+\beta_2+\gamma_1+\gamma_2)+\\(y+1)(\gamma_1+\delta_2)}} \sum_{y'=(\gamma_4+\delta_5)x}^{1-(\gamma_4+\gamma_5)x} \sum_{x'=y'(\beta_4+\alpha_5+\beta_5)}^{\substack{y'(\beta_4+\alpha_5+\beta_5+\gamma_4+\gamma_5)+\\(y'+1)(\gamma_4+\delta_5)}} \right. \nonumber\\
  && \left[ \sum_{y''=(\gamma_7+\delta_8)(x+x')}^{1+(\gamma_7+\gamma_8)(x+x')} \sum_{x''=y''(\beta_7+\alpha_8+\beta_8)}^{\substack{y''(\beta_7+\alpha_8+\beta_8)+\\y''(\gamma_7+\gamma_8)+\\(y''+1)(\gamma_7+\delta_8)}} \sum_{y'''=(\gamma_{10}+\delta_{11})(x+x'+x'')}^{1+(\gamma_{10}+\gamma_{11})(x+x'+x'')} \sum_{x'''=y'''(\beta_{10}+\alpha_{11}+\beta_{12})}^{\substack{y'''(\beta_{10}+\alpha_{11}+\beta_{12})+\\y'''(\gamma_{10}+\gamma_{11})+\\(y'''+1)(\gamma_{10}+\delta_{11})}} \right. \nonumber\\
  && \times 4\times 2\times 2^{\left[1-y(\gamma_1-\gamma_2)^2-(y-1)^2(\gamma_1-\delta_2)^2\right]\left[1-y'(\gamma_4-\gamma_5)^2-(y'-1)^2(\gamma_4-\delta_5)^2\right]\left[1-y''(\gamma_7-\gamma_8)^2-(y''-1)^2(\gamma_7-\delta_8)^2\right]\left[1-y'''(\gamma_{10}-\gamma_{11})^2-(y'''-1)^2(\gamma_{10}-\delta_{11})^2\right]} \nonumber\\
  && \times \left[ \mathcal G_1(x,y) + \mathcal G_2(x,y) \right] \left[ \mathcal G_1(x',y') + \mathcal G_2(x',y') \right] \left[ \mathcal G_1(x'',y'') + \mathcal G_2(x'',y'') \right] \left[ \mathcal G_1(x''',y''') + \mathcal G_2(x''',y''') \right] \nonumber\\ %
  && + \sum_{y''=0}^{1} \sum_{x''=y''(\beta_7+\alpha_8+\beta_8)}^{\substack{y''(\beta_7+\alpha_8+\beta_8)+\\y''(\gamma_7+\gamma_8)+\\(y''+1)(\gamma_7+\delta_8)}} \sum_{x'''=0}^{1+\beta_{12}} \sum_{y'''=0}^{1+\alpha_{12}} \nonumber\\
  && \times 2\times 2^{\left[1-y(\gamma_1-\gamma_2)^2-(y-1)^2(\gamma_1-\delta_2)^2\right]\left[1-y'(\gamma_4-\gamma_5)^2-(y'-1)^2(\gamma_4-\delta_5)^2\right]} 2^{1-y'''(\gamma_{10}-\gamma_{11})^2-(y'''-1)^2(\gamma_{10}-\delta_{11})^2} \nonumber\\
  && \times \left[ \mathcal G_1(x,y) + \mathcal G_2(x,y) \right] \left[ \mathcal G_1(x',y') + \mathcal G_2(x',y') \right] \left[ \mathcal G_1(x'',y'') + \mathcal G_2(x'',y'') \right]  \mathcal G_3(x''',y''') \nonumber\\
  && + \sum_{y''=(\gamma_7+\delta_8)(x+x')}^{1+(\gamma_7+\gamma_8)(x+x')} \sum_{x''=y''(\beta_7+\alpha_8+\beta_8)}^{\substack{y''(\beta_7+\alpha_8+\beta_8)+\\y''(\gamma_7+\gamma_8)+\\(y''+1)(\gamma_7+\delta_8)}} \sum_{y'''=x(1+\alpha_1+\beta_1)+x'(1+\alpha_4+\beta_4)+x''(1+\alpha_7+\beta_7)}^1 \sum_{x'''=y'''(\delta_{10} + \delta_{11})+x(1+\alpha_1+\beta_1)+x'(1+\alpha_4+\beta_4)+x''(1+\alpha_7+\beta_7)}^{1+y'''(\gamma_{10} + \gamma_{11})+x(1+\alpha_1+\beta_1)+x'(1+\alpha_4+\beta_4)+x''(1+\alpha_7+\beta_7)}\nonumber\\
  && \times 4\times 2 \times 2^{\left[1-y(\gamma_1-\gamma_2)^2-(y-1)^2(\gamma_1-\delta_2)^2\right]\left[1-y'(\gamma_4-\gamma_5)^2-(y'-1)^2(\gamma_4-\delta_5)^2\right] \left[1-y''(\gamma_7-\gamma_8)^2-(y''-1)^2(\gamma_7-\delta_8)^2\right]\left[(y'''-1)^2[(x''')^2(\gamma_{10}+\delta_{11}) + (x'''-1)^2(\gamma_{11}+\delta_{10})]\right]} \nonumber\\
  && \times \left[ \mathcal G_1(x,y) + \mathcal G_2(x,y) \right] \left[ \mathcal G_1(x',y') + \mathcal G_2(x',y') \right] \left[ \mathcal G_1(x'',y'') + \mathcal G_2(x'',y'') \right] \mathcal G_4(x''',y''')\nonumber\\
  && + \sum_{x''=0}^{1+\beta_9} \sum_{y''=0}^{1+\alpha_9} \sum_{y'''=0}^{1} \sum_{x'''=y'''(\beta_{10}+\alpha_{11}+\beta_{11})}^{\substack{y'''(\beta_{10}+\alpha_{11}+\beta_{11})+\\y'''(\gamma_{10}+\gamma_{11})+\\(y'''+1)(\gamma_{10}+\delta_{11})}} \nonumber\\
  && \times 2\times 2^{\left[1-y(\gamma_1-\gamma_2)^2-(y-1)^2(\gamma_1-\delta_2)^2\right]\left[1-y'(\gamma_4-\gamma_5)^2-(y'-1)^2(\gamma_4-\delta_5)^2\right]} 2^{1-y'''(\gamma_{10}-\gamma_{11})^2-(y'''-1)^2(\gamma_{10}-\delta_{11})^2} \nonumber\\
  && \times \left[ \mathcal G_1(x,y) + \mathcal G_2(x,y) \right] \left[ \mathcal G_1(x',y') + \mathcal G_2(x',y') \right] \mathcal G_3(x'',y'') \left[ \mathcal G_1(x''',y''') + \mathcal G_2(x''',y''') \right] \nonumber\\
  && + \sum_{x''=0}^{1+\beta_9} \sum_{y''=0}^{1+\alpha_9} \sum_{x'''=0}^{1+\beta_{12}} \sum_{y'''=0}^{1+\alpha_{12}} \nonumber\\
  && \times 2\times 2^{\left[1-y(\gamma_1-\gamma_2)^2-(y-1)^2(\gamma_1-\delta_2)^2\right]\left[1-y'(\gamma_4-\gamma_5)^2-(y'-1)^2(\gamma_4-\delta_5)^2\right]} \nonumber\\
  && \times \left[ \mathcal G_1(x,y) + \mathcal G_2(x,y) \right] \left[ \mathcal G_1(x',y') + \mathcal G_2(x',y') \right] \mathcal G_3(x'',y'') \mathcal G_3(x''',y''')\nonumber\\
  && + \sum_{x''=0}^{1+\beta_9} \sum_{y''=0}^{1+\alpha_9} \sum_{y'''=0}^1 \sum_{x'''=y'''(\delta_{10} + \delta_{11})}^{1+y'''(\gamma_{10} + \gamma_{11})} \nonumber\\
  && \times 2\times 2^{\left[1-y(\gamma_1-\gamma_2)^2-(y-1)^2(\gamma_1-\delta_2)^2\right]\left[1-y'(\gamma_4-\gamma_5)^2-(y'-1)^2(\gamma_4-\delta_5)^2\right]} 2^{(y'''-1)^2[(x''')^2(\gamma_{10}+\delta_{11}) + (x'''-1)^2(\gamma_{11}+\delta_{10})]} \nonumber\\
  && \times \left[ \mathcal G_1(x,y) + \mathcal G_2(x,y) \right] \left[ \mathcal G_1(x',y') + \mathcal G_2(x',y') \right] \mathcal G_3(x'',y'') \mathcal G_4(x''',y''') \nonumber\\
  && + \sum_{y''=x(1+\alpha_1+\beta_1)+x'(1+\alpha_4+\beta_4)}^1 \sum_{x''=y''(\delta_7 + \delta_8)+x(1+\alpha_1+\beta_1)+x'(1+\alpha_4+\beta_4)}^{1+y''(\gamma_7 + \gamma_8)+x(1+\alpha_1+\beta_1)+x'(1+\alpha_4+\beta_4)} \sum_{y'''=(\gamma_{10}+\delta_{11})(x+x'+y'')}^{1-(\gamma_{10}+\gamma_{11})(x+x'+y'')} \sum_{x'''=y'''(\beta_{10}+\alpha_{11}+\beta_{11})}^{y'''(\beta_{10}+\alpha_{11}+\beta_{11})+y'''(\gamma_{10}+\gamma_{11})+(y'''+1)(\gamma_{10}+\delta_{11})} \nonumber\\
  && \times 4\times 2 \times 2^{\left[1-y(\gamma_1-\gamma_2)^2-(y-1)^2(\gamma_1-\delta_2)^2\right]\left[1-y'(\gamma_4-\gamma_5)^2-(y'-1)^2(\gamma_4-\delta_5)^2\right] \left[(y''-1)^2[(x'')^2(\gamma_7+\delta_8) + (x''-1)^2(\gamma_8+\delta_7)]\right]\left[1-y'''(\gamma_{10}-\gamma_{11})^2-(y'''-1)^2(\gamma_{10}-\delta_{11})^2\right]} \nonumber\\
  && \times \left[ \mathcal G_1(x,y) + \mathcal G_2(x,y) \right] \left[ \mathcal G_1(x',y') + \mathcal G_2(x',y') \right] \mathcal G_4(x'',y'') \left[ \mathcal G_1(x''',y''') + \mathcal G_2(x''',y''') \right] \nonumber\\
  && + \sum_{y''=0}^1 \sum_{x''=y''(\delta_7 + \delta_8)}^{1+y''(\gamma_7 + \gamma_8)} \sum_{x'''=0}^{1+\beta_{12}} \sum_{y'''=0}^{1+\alpha_{12}} \nonumber\\
  && \times 2\times 2^{\left[1-y(\gamma_1-\gamma_2)^2-(y-1)^2(\gamma_1-\delta_2)^2\right]\left[1-y'(\gamma_4-\gamma_5)^2-(y'-1)^2(\gamma_4-\delta_5)^2\right]} 2^{(y''-1)^2[(x'')^2(\gamma_7+\delta_8) + (x''-1)^2(\gamma_8+\delta_7)]} \nonumber\\
  && \times \left[ \mathcal G_1(x,y) + \mathcal G_2(x,y) \right] \left[ \mathcal G_1(x',y') + \mathcal G_2(x',y') \right] \mathcal G_4(x'',y'') \mathcal G_3(x''',y''') \nonumber\\
  && + \sum_{y''=x(1+\alpha_1+\beta_1)+x'(1+\alpha_4+\beta_4)}^1 \sum_{x''=y''(\delta_7 + \delta_8)+x(1+\alpha_1+\beta_1)+x'(1+\alpha_4+\beta_4)}^{1+y''(\gamma_7 + \gamma_8)+x(1+\alpha_1+\beta_1)+x'(1+\alpha_4+\beta_4)} \sum_{y'''=x(1+\alpha_1+\beta_1)+x'(1+\alpha_4+\beta_4)+y''}^1 \sum_{x'''=y'''(\delta_{10} + \delta_{11})+(x(1+\alpha_1+\beta_1)+x'(1+\alpha_4+\beta_4)+y'')}^{1+y'''(\gamma_{10} + \gamma_{11})+(x(1+\alpha_1+\beta_1)+x'(1+\alpha_4+\beta_4)+y'')} \nonumber\\
  && \times 4\times 2 \times 2^{\left[1-y(\gamma_1-\gamma_2)^2-(y-1)^2(\gamma_1-\delta_2)^2\right]\left[1-y'(\gamma_4-\gamma_5)^2-(y'-1)^2(\gamma_4-\delta_5)^2\right] \left[(y''-1)^2[(x'')^2(\gamma_7+\delta_8) + (x''-1)^2(\gamma_8+\delta_7)]\right]\left[(y'''-1)^2[(x''')^2(\gamma_{10}+\delta_{11}) + (x'''-1)^2(\gamma_{11}+\delta_{10})]\right]} \nonumber\\
  && \times \left[ \mathcal G_1(x,y) + \mathcal G_2(x,y) \right] \left[ \mathcal G_1(x',y') + \mathcal G_2(x',y') \right] \mathcal G_4(x'',y'') \mathcal G_4(x''',y''') \Bigg] \Bigg\} \nonumber
\end{eqnarray}

\begin{eqnarray}
\label{eq:12qubittgate2}
  &+& \frac{\omega}{64} \left\{2 \sum_{y=0}^{1} \sum_{x=y(\beta_1+\alpha_2+\beta_2)}^{y(\beta_1+\alpha_2+\beta_2)+y(\gamma_1+\gamma_2)+(y+1)(\gamma_1+\delta_2)} \sum_{x'=0}^{1+\beta_6} \sum_{y'=0}^{1+\alpha_6} \right.\\
  && \left[ \sum_{y''=0}^{1} \sum_{x''=y''(\beta_7+\alpha_8+\beta_8)}^{\substack{y''(\beta_7+\alpha_8+\beta_8)+\\y''(\gamma_7+\gamma_8)+\\(y''+1)(\gamma_7+\delta_8)}} \sum_{y'''=(\gamma_{10}+\delta_{11})x''}^{1-(\gamma_{10}+\gamma_{11})x''} \sum_{x'''=y'''(\beta_{10}+\alpha_{11}+\beta_{11})}^{\substack{y'''(\beta_{10}+\alpha_{11}+\beta_{11})+\\y'''(\gamma_{10}+\gamma_{11})+\\(y'''+1)(\gamma_{10}+\delta_{11})}} \right. \nonumber\\
  && \times 2^{1-y'(\gamma_4-\gamma_5)^2-(y'-1)^2(\gamma_4-\delta_5)^2} 2\times 2^{\left[1-y''(\gamma_7-\gamma_8)^2-(y''-1)^2(\gamma_7-\delta_8)^2\right]\left[1-y'''(\gamma_{10}-\gamma_{11})^2-(y'''-1)^2(\gamma_{10}-\delta_{11})^2\right]} \nonumber\\
  && \times \left[ \mathcal G_1(x,y) + \mathcal G_2(x,y) \right]  \mathcal G_3(x',y') \left[ \mathcal G_1(x'',y'') + \mathcal G_2(x'',y'') \right] \left[ \mathcal G_1(x''',y''') + \mathcal G_2(x''',y''') \right] \nonumber\\
  && + 2 \sum_{y''=0}^{1} \sum_{x''=y''(\beta_7+\alpha_8+\beta_8)}^{\substack{y''(\beta_7+\alpha_8+\beta_8)+\\y''(\gamma_7+\gamma_8)+\\(y''+1)(\gamma_7+\delta_8)}} \sum_{x'''=0}^{1+\beta_{12}} \sum_{y'''=0}^{1+\alpha_{12}} \nonumber\\
  && \times 2^{1-y'(\gamma_4-\gamma_5)^2-(y'-1)^2(\gamma_4-\delta_5)^2} 2^{1-y'''(\gamma_{10}-\gamma_{11})^2-(y'''-1)^2(\gamma_{10}-\delta_{11})^2} \nonumber\\
  && \times \left[ \mathcal G_1(x,y) + \mathcal G_2(x,y) \right]  \mathcal G_3(x',y') \left[ \mathcal G_1(x'',y'') + \mathcal G_2(x'',y'') \right]  \mathcal G_3(x''',y''') \nonumber\\
  && + \sum_{y''=0}^{1} \sum_{x''=y''(\beta_7+\alpha_8+\beta_8)}^{\substack{y''(\beta_7+\alpha_8+\beta_8)+\\y''(\gamma_7+\gamma_8)+\\(y''+1)(\gamma_7+\delta_8)}} \sum_{y'''=x''(1+\alpha_7+\beta_7)}^1 \sum_{x'''=y'''(\delta_{10} + \delta_{11})+x''(1+\alpha_7+\beta_7)}^{1+y'''(\gamma_{10} + \gamma_{11})+x''(1+\alpha_7+\beta_7)}\nonumber\\
  && \times 2^{1-y'(\gamma_4-\gamma_5)^2-(y'-1)^2(\gamma_4-\delta_5)^2} 2\times 2^{\left[1-y''(\gamma_7-\gamma_8)^2-(y''-1)^2(\gamma_7-\delta_8)^2\right]\left[(y'''-1)^2[(x''')^2(\gamma_{10}+\delta_{11}) + (x'''-1)^2(\gamma_{11}+\delta_{10})]\right]} \nonumber\\
  && \times \left[ \mathcal G_1(x,y) + \mathcal G_2(x,y) \right]  \mathcal G_3(x',y') \left[ \mathcal G_1(x'',y'') + \mathcal G_2(x'',y'') \right] \mathcal G_4(x''',y''')\nonumber\\
  && + 2 \sum_{x''=0}^{1+\beta_9} \sum_{y''=0}^{1+\alpha_9} \sum_{y'''=0}^{1} \sum_{x'''=y'''(\beta_{10}+\alpha_{11}+\beta_{11})}^{\substack{y'''(\beta_{10}+\alpha_{11}+\beta_{11})+\\y'''(\gamma_{10}+\gamma_{11})+\\(y'''+1)(\gamma_{10}+\delta_{11})}} \nonumber\\
  && \times 2^{1-y'(\gamma_4-\gamma_5)^2-(y'-1)^2(\gamma_4-\delta_5)^2} 2^{1-y'''(\gamma_{10}-\gamma_{11})^2-(y'''-1)^2(\gamma_{10}-\delta_{11})^2} \nonumber\\
  && \times \left[ \mathcal G_1(x,y) + \mathcal G_2(x,y) \right]  \mathcal G_3(x',y') \mathcal G_3(x'',y'') \left[ \mathcal G_1(x''',y''') + \mathcal G_2(x''',y''') \right] \nonumber\\
  && + 4 \sum_{x''=0}^{1+\beta_9} \sum_{y''=0}^{1+\alpha_9} \sum_{x'''=0}^{1+\beta_{12}} \sum_{y'''=0}^{1+\alpha_{12}} \nonumber\\
  && \times 2^{1-y'(\gamma_4-\gamma_5)^2-(y'-1)^2(\gamma_4-\delta_5)^2} \nonumber\\
  && \times \left[ \mathcal G_1(x,y) + \mathcal G_2(x,y) \right]  \mathcal G_3(x',y') \mathcal G_3(x'',y'') \mathcal G_3(x''',y''')\nonumber\\
  && + 2 \sum_{x''=0}^{1+\beta_9} \sum_{y''=0}^{1+\alpha_9} \sum_{y'''=0}^1 \sum_{x'''=y'''(\delta_{10} + \delta_{11})}^{1+y'''(\gamma_{10} + \gamma_{11})} \nonumber\\
  && \times 2^{1-y'(\gamma_4-\gamma_5)^2-(y'-1)^2(\gamma_4-\delta_5)^2} 2^{(y'''-1)^2[(x''')^2(\gamma_{10}+\delta_{11}) + (x'''-1)^2(\gamma_{11}+\delta_{10})]} \nonumber\\
  && \times \left[ \mathcal G_1(x,y) + \mathcal G_2(x,y) \right]  \mathcal G_3(x',y') \mathcal G_3(x'',y'') \mathcal G_4(x''',y''') \nonumber\\
  && + \sum_{y''=0}^1 \sum_{x''=y''(\delta_7 + \delta_8)}^{1+y''(\gamma_7 + \gamma_8)} \sum_{y'''=(\gamma_{10}+\delta_{11})y''}^{1-(\gamma_{10}+\gamma_{11})y''} \sum_{x'''=y'''(\beta_{10}+\alpha_{11}+\beta_{11})}^{y'''(\beta_{10}+\alpha_{11}+\beta_{11})+y'''(\gamma_{10}+\gamma_{11})+(y'''+1)(\gamma_{10}+\delta_{11})} \nonumber\\
  && \times 2^{1-y'(\gamma_4-\gamma_5)^2-(y'-1)^2(\gamma_4-\delta_5)^2} 2 \times 2^{\left[(y''-1)^2[(x'')^2(\gamma_7+\delta_8) + (x''-1)^2(\gamma_8+\delta_7)]\right]\left[1-y'''(\gamma_{10}-\gamma_{11})^2-(y'''-1)^2(\gamma_{10}-\delta_{11})^2\right]} \nonumber\\
  && \times \left[ \mathcal G_1(x,y) + \mathcal G_2(x,y) \right]  \mathcal G_3(x',y') \mathcal G_4(x'',y'') \left[ \mathcal G_1(x''',y''') + \mathcal G_2(x''',y''') \right] \nonumber\\
  && + 2 \sum_{y''=0}^1 \sum_{x''=y''(\delta_7 + \delta_8)}^{1+y''(\gamma_7 + \gamma_8)} \sum_{x'''=0}^{1+\beta_{12}} \sum_{y'''=0}^{1+\alpha_{12}} \nonumber\\
  && \times 2^{1-y'(\gamma_4-\gamma_5)^2-(y'-1)^2(\gamma_4-\delta_5)^2} 2^{(y''-1)^2[(x'')^2(\gamma_7+\delta_8) + (x''-1)^2(\gamma_8+\delta_7)]} \nonumber\\
  && \times \left[ \mathcal G_1(x,y) + \mathcal G_2(x,y) \right]  \mathcal G_3(x',y') \mathcal G_4(x'',y'') \mathcal G_3(x''',y''') \nonumber\\
  && + \sum_{y''=0}^1 \sum_{x''=y''(\delta_7 + \delta_8)}^{1+y''(\gamma_7 + \gamma_8)} \sum_{y'''=y''}^1 \sum_{x'''=y'''(\delta_{10} + \delta_{11})+y''}^{1+y'''(\gamma_{10} + \gamma_{11})+y''} \nonumber\\
  && \times 2^{1-y'(\gamma_4-\gamma_5)^2-(y'-1)^2(\gamma_4-\delta_5)^2} 2\times 2^{\left[(y''-1)^2[(x'')^2(\gamma_7+\delta_8) + (x''-1)^2(\gamma_8+\delta_7)]\right]\left[(y'''-1)^2[(x''')^2(\gamma_{10}+\delta_{11}) + (x'''-1)^2(\gamma_{11}+\delta_{10})]\right]} \nonumber\\
  && \times \left[ \mathcal G_1(x,y) + \mathcal G_2(x,y) \right]  \mathcal G_3(x',y') \mathcal G_4(x'',y'') \mathcal G_4(x''',y''') \Bigg] \Bigg\}\nonumber
\end{eqnarray}

\begin{eqnarray}
\label{eq:12qubittgate3}
  &+& \frac{\omega}{64} \left\{\sum_{y=0}^{1} \sum_{x=y(\beta_1+\alpha_2+\beta_2)}^{\substack{y(\beta_1+\alpha_2+\beta_2)+\\y(\gamma_1+\gamma_2)+\\(y+1)(\gamma_1+\delta_2)}} \sum_{y'=x(1+\alpha_1+\beta_1)}^1 \sum_{x'=y'(\delta_4 + \delta_5)+x(1+\alpha_1+\beta_1)}^{1+y'(\gamma_4 + \gamma_5)+x(1+\alpha_1+\beta_1)} \right. \\
  && \left[ \sum_{y''=(\gamma_7+\delta_8)(x+y')}^{1+(\gamma_7+\gamma_8)(x+y')} \sum_{x''=y''(\beta_7+\alpha_8+\beta_8)}^{\substack{y''(\beta_7+\alpha_8+\beta_8)+\\y''(\gamma_7+\gamma_8)+\\(y''+1)(\gamma_7+\delta_8)}} \sum_{y'''=(\gamma_{10}+\delta_{11})(x+y'+x'')}^{1+(\gamma_{10}+\gamma_{11})(x+y'+x'')} \sum_{x'''=y'''(\beta_{10}+\alpha_{11}+\beta_{11})}^{\substack{y'''(\beta_{10}+\alpha_{11}+\beta_{11})+\\y'''(\gamma_{10}+\gamma_{11})+\\(y'''+1)(\gamma_{10}+\delta_{11})}} \right. \nonumber\\
  && \times 4 \times 2 \times 2^{\left[1-y(\gamma_1-\gamma_2)^2-(y-1)^2(\gamma_1-\delta_2)^2\right]\left[(y'-1)^2[(x')^2(\gamma_4+\delta_5) + (x'-1)^2(\gamma_5+\delta_4)]\right] \left[1-y''(\gamma_7-\gamma_8)^2-(y''-1)^2(\gamma_7-\delta_8)^2\right]\left[1-y'''(\gamma_{10}-\gamma_{11})^2-(y'''-1)^2(\gamma_{10}-\delta_{11})^2\right]} \nonumber\\
  && \times \left[ \mathcal G_1(x,y) + \mathcal G_2(x,y) \right] \mathcal G_4(x',y') \left[ \mathcal G_1(x'',y'') + \mathcal G_2(x'',y'') \right] \left[ \mathcal G_1(x''',y''') + \mathcal G_2(x''',y''') \right] \nonumber\\
  && + 2 \sum_{y''=(\gamma_7+\delta_8)(x+y')}^{1+(\gamma_7+\gamma_8)(x+y')} \sum_{x''=y''(\beta_7+\alpha_8+\beta_8)}^{\substack{y''(\beta_7+\alpha_8+\beta_8)+\\y''(\gamma_7+\gamma_8)+\\(y''+1)(\gamma_7+\delta_8)}} \sum_{x'''=0}^{1+\beta_{12}} \sum_{y'''=0}^{1+\alpha_{12}} \nonumber\\
  && \times 2\times 2^{\left[1-y(\gamma_1-\gamma_2)^2-(y-1)^2(\gamma_1-\delta_2)^2\right]\left[(y'-1)^2[(x')^2(\gamma_4+\delta_5) + (x'-1)^2(\gamma_5+\delta_4)]\right]} 2^{1-y'''(\gamma_{10}-\gamma_{11})^2-(y'''-1)^2(\gamma_{10}-\delta_{11})^2} \nonumber\\
  && \times \left[ \mathcal G_1(x,y) + \mathcal G_2(x,y) \right] \mathcal G_4(x',y') \left[ \mathcal G_1(x'',y'') + \mathcal G_2(x'',y'') \right]  \mathcal G_3(x''',y''') \nonumber\\
  && + \sum_{y''=(\gamma_7+\delta_8)(x+y')}^{1+(\gamma_7+\gamma_8)(x+y')} \sum_{x''=y''(\beta_7+\alpha_8+\beta_8)}^{\substack{y''(\beta_7+\alpha_8+\beta_8)+\\y''(\gamma_7+\gamma_8)+\\(y''+1)(\gamma_7+\delta_8)}} \sum_{y'''=(x(1+\alpha_1+\beta_1)+y'+x''(1+\alpha_7+\beta_7))}^1 \sum_{x'''=y'''(\delta_{10} + \delta_{11})+(x(1+\alpha_1+\beta_1)+y'+x''(1+\alpha_7+\beta_7))}^{1+y'''(\gamma_{10} + \gamma_{11})+(x(1+\alpha_1+\beta_1)+y'+x''(1+\alpha_7+\beta_7))}\nonumber\\
  && \times 4 \times 2 \times 2^{\left[1-y(\gamma_1-\gamma_2)^2-(y-1)^2(\gamma_1-\delta_2)^2\right]\left[(y'-1)^2[(x')^2(\gamma_4+\delta_5) + (x'-1)^2(\gamma_5+\delta_4)]\right] \left[1-y''(\gamma_7-\gamma_8)^2-(y''-1)^2(\gamma_7-\delta_8)^2\right]\left[(y'''-1)^2[(x''')^2(\gamma_{10}+\delta_{11}) + (x'''-1)^2(\gamma_{11}+\delta_{10})]\right]} \nonumber\\
  && \times \left[ \mathcal G_1(x,y) + \mathcal G_2(x,y) \right] \mathcal G_4(x',y') \left[ \mathcal G_1(x'',y'') + \mathcal G_2(x'',y'') \right] \mathcal G_4(x''',y''')\nonumber\\
  && + 2 \sum_{x''=0}^{1+\beta_9} \sum_{y''=0}^{1+\alpha_9} \sum_{y'''=0}^{1} \sum_{x'''=y'''(\beta_{10}+\alpha_{11}+\beta_{11})}^{\substack{y'''(\beta_{10}+\alpha_{11}+\beta_{11})+\\y'''(\gamma_{10}+\gamma_{11})+\\(y'''+1)(\gamma_{10}+\delta_{11})}} \nonumber\\
  && \times 2\times 2^{\left[1-y(\gamma_1-\gamma_2)^2-(y-1)^2(\gamma_1-\delta_2)^2\right]\left[(y'-1)^2[(x')^2(\gamma_4+\delta_5) + (x'-1)^2(\gamma_5+\delta_4)]\right]} 2^{1-y'''(\gamma_{10}-\gamma_{11})^2-(y'''-1)^2(\gamma_{10}-\delta_{11})^2} \nonumber\\
  && \times \left[ \mathcal G_1(x,y) + \mathcal G_2(x,y) \right] \mathcal G_4(x',y') \mathcal G_3(x'',y'') \left[ \mathcal G_1(x''',y''') + \mathcal G_2(x''',y''') \right] \nonumber\\
  && + 4 \sum_{x''=0}^{1+\beta_9} \sum_{y''=0}^{1+\alpha_9} \sum_{x'''=0}^{1+\beta_{12}} \sum_{y'''=0}^{1+\alpha_{12}} \nonumber\\
  && \times 2\times 2^{\left[1-y(\gamma_1-\gamma_2)^2-(y-1)^2(\gamma_1-\delta_2)^2\right]\left[(y'-1)^2[(x')^2(\gamma_4+\delta_5) + (x'-1)^2(\gamma_5+\delta_4)]\right]} \nonumber\\
  && \times \left[ \mathcal G_1(x,y) + \mathcal G_2(x,y) \right] \mathcal G_4(x',y') \mathcal G_3(x'',y'') \mathcal G_3(x''',y''')\nonumber\\
  && + 2 \sum_{x''=0}^{1+\beta_9} \sum_{y''=0}^{1+\alpha_9} \sum_{y'''=0}^1 \sum_{x'''=y'''(\delta_{10} + \delta_{11})}^{1+y'''(\gamma_{10} + \gamma_{11})} \nonumber\\
  && \times 2\times 2^{\left[1-y(\gamma_1-\gamma_2)^2-(y-1)^2(\gamma_1-\delta_2)^2\right]\left[(y'-1)^2[(x')^2(\gamma_4+\delta_5) + (x'-1)^2(\gamma_5+\delta_4)]\right]} 2^{(y'''-1)^2[(x''')^2(\gamma_{10}+\delta_{11}) + (x'''-1)^2(\gamma_{11}+\delta_{10})]} \nonumber\\
  && \times \left[ \mathcal G_1(x,y) + \mathcal G_2(x,y) \right] \mathcal G_4(x',y') \mathcal G_3(x'',y'') \mathcal G_4(x''',y''') \nonumber\\
  && + \sum_{y''=x(1+\alpha_1+\beta_1)+y'}^1 \sum_{x''=y''(\delta_7 + \delta_8)+(x(1+\alpha_1+\beta_1)+y')}^{1+y''(\gamma_7 + \gamma_8)+(x(1+\alpha_1+\beta_1)+y')} \sum_{y'''=(\gamma_{10}+\delta_{11})(x+y'+y'')}^{1-(\gamma_{10}+\gamma_{11})(x+y'+y'')} \sum_{x'''=y'''(\beta_{10}+\alpha_{11}+\beta_{11})}^{y'''(\beta_{10}+\alpha_{11}+\beta_{11})+y'''(\gamma_{10}+\gamma_{11})+(y'''+1)(\gamma_{10}+\delta_{11})} \nonumber\\
  && \times 4 \times 2 \times 2^{\left[1-y(\gamma_1-\gamma_2)^2-(y-1)^2(\gamma_1-\delta_2)^2\right]\left[(y'-1)^2[(x')^2(\gamma_4+\delta_5) + (x'-1)^2(\gamma_5+\delta_4)]\right] \left[(y''-1)^2[(x'')^2(\gamma_7+\delta_8) + (x''-1)^2(\gamma_8+\delta_7)]\right]\left[1-y'''(\gamma_{10}-\gamma_{11})^2-(y'''-1)^2(\gamma_{10}-\delta_{11})^2\right]} \nonumber\\
  && \times \left[ \mathcal G_1(x,y) + \mathcal G_2(x,y) \right] \mathcal G_4(x',y') \mathcal G_4(x'',y'') \left[ \mathcal G_1(x''',y''') + \mathcal G_2(x''',y''') \right] \nonumber\\
  && + 2 \sum_{y''=0}^1 \sum_{x''=y''(\delta_7 + \delta_8)}^{1+y''(\gamma_7 + \gamma_8)} \sum_{x'''=0}^{1+\beta_{12}} \sum_{y'''=0}^{1+\alpha_{12}} \nonumber\\
  && \times 2\times 2^{\left[1-y(\gamma_1-\gamma_2)^2-(y-1)^2(\gamma_1-\delta_2)^2\right]\left[(y'-1)^2[(x')^2(\gamma_4+\delta_5) + (x'-1)^2(\gamma_5+\delta_4)]\right]} 2^{(y''-1)^2[(x'')^2(\gamma_7+\delta_8) + (x''-1)^2(\gamma_8+\delta_7)]} \nonumber\\
  && \times \left[ \mathcal G_1(x,y) + \mathcal G_2(x,y) \right] \mathcal G_4(x',y') \mathcal G_4(x'',y'') \mathcal G_3(x''',y''') \nonumber\\
  && + \sum_{y''=x(1+\alpha_1+\beta_1)+y'}^1 \sum_{x''=y''(\delta_7 + \delta_8)+(x(1+\alpha_1+\beta_1)+y')}^{1+y''(\gamma_7 + \gamma_8)+(x(1+\alpha_1+\beta_1)+y')} \sum_{y'''=x(1+\alpha_1+\beta_1)+y'+y''}^1 \sum_{x'''=y'''(\delta_{10} + \delta_{11})+(x(1+\alpha_1+\beta_1)+y'+y'')}^{1+y'''(\gamma_{10} + \gamma_{11})+(x(1+\alpha_1+\beta_1)+y'+y'')} \nonumber\\
  && \times 4 \times 2\times 2^{\left[1-y(\gamma_1-\gamma_2)^2-(y-1)^2(\gamma_1-\delta_2)^2\right]\left[(y'-1)^2[(x')^2(\gamma_4+\delta_5) + (x'-1)^2(\gamma_5+\delta_4)]\right] \left[(y''-1)^2[(x'')^2(\gamma_7+\delta_8) + (x''-1)^2(\gamma_8+\delta_7)]\right]\left[(y'''-1)^2[(x''')^2(\gamma_{10}+\delta_{11}) + (x'''-1)^2(\gamma_{11}+\delta_{10})]\right]} \nonumber\\
  && \times \left[ \mathcal G_1(x,y) + \mathcal G_2(x,y) \right] \mathcal G_4(x',y') \mathcal G_4(x'',y'') \mathcal G_4(x''',y''') \Bigg] \Bigg\}\nonumber
\end{eqnarray}

\begin{eqnarray}
\label{eq:12qubittgate4}
  &+& \frac{\omega}{64} \left\{2 \sum_{x=0}^{1+\beta_3} \sum_{y=0}^{1+\alpha_3} \sum_{y'=0}^{1} \sum_{x'=y'(\beta_4+\alpha_5+\beta_5)}^{\substack{y'(\beta_4+\alpha_5+\beta_5)+\\y'(\gamma_4+\gamma_5)+\\(y'+1)(\gamma_4+\delta_5)}} \right.\\
  && \left[ \sum_{y''=0}^{1} \sum_{x''=y''(\beta_7+\alpha_8+\beta_8)}^{\substack{y''(\beta_7+\alpha_8+\beta_8)+\\y''(\gamma_7+\gamma_8)+\\(y''+1)(\gamma_7+\delta_8)}} \sum_{y'''=(\gamma_{10}+\delta_{11})x''}^{1-(\gamma_{10}+\gamma_{11})x''} \sum_{x'''=y'''(\beta_{10}+\alpha_{11}+\beta_{11})}^{\substack{y'''(\beta_{10}+\alpha_{11}+\beta_{11})+\\y'''(\gamma_{10}+\gamma_{11})+\\(y'''+1)(\gamma_{10}+\delta_{11})}} \right. \nonumber\\
  && \times 2^{1-y'(\gamma_4-\gamma_5)^2-(y'-1)^2(\gamma_4-\delta_5)^2} 2\times 2^{\left[1-y''(\gamma_7-\gamma_8)^2-(y''-1)^2(\gamma_7-\delta_8)^2\right]\left[1-y'''(\gamma_{10}-\gamma_{11})^2-(y'''-1)^2(\gamma_{10}-\delta_{11})^2\right]} \nonumber\\
  && \times \mathcal G_3(x,y) \left[ \mathcal G_1(x',y') + \mathcal G_2(x',y') \right] \left[ \mathcal G_1(x'',y'') + \mathcal G_2(x'',y'') \right] \left[ \mathcal G_1(x''',y''') + \mathcal G_2(x''',y''') \right] \nonumber\\
  && + 2 \sum_{y''=0}^{1} \sum_{x''=y''(\beta_7+\alpha_8+\beta_8)}^{\substack{y''(\beta_7+\alpha_8+\beta_8)+\\y''(\gamma_7+\gamma_8)+\\(y''+1)(\gamma_7+\delta_8)}} \sum_{x'''=0}^{1+\beta_{12}} \sum_{y'''=0}^{1+\alpha_{12}} \nonumber\\
  && \times 2^{1-y'(\gamma_4-\gamma_5)^2-(y'-1)^2(\gamma_4-\delta_5)^2} 2^{1-y'''(\gamma_{10}-\gamma_{11})^2-(y'''-1)^2(\gamma_{10}-\delta_{11})^2} \nonumber\\
  && \times \mathcal G_3(x,y) \left[ \mathcal G_1(x',y') + \mathcal G_2(x',y') \right] \left[ \mathcal G_1(x'',y'') + \mathcal G_2(x'',y'') \right]  \mathcal G_3(x''',y''') \nonumber\\
  && + \sum_{y''=0}^{1} \sum_{x''=y''(\beta_7+\alpha_8+\beta_8)}^{\substack{y''(\beta_7+\alpha_8+\beta_8)+\\y''(\gamma_7+\gamma_8)+\\(y''+1)(\gamma_7+\delta_8)}} \sum_{y'''=x''(1+\alpha_7+\beta_7)}^1 \sum_{x'''=y'''(\delta_{10} + \delta_{11})+x''(1+\alpha_7+\beta_7)}^{1+y'''(\gamma_{10} + \gamma_{11})+x''(1+\alpha_7+\beta_7)}\nonumber\\
  && \times 2^{1-y'(\gamma_4-\gamma_5)^2-(y'-1)^2(\gamma_4-\delta_5)^2} 2\times 2^{\left[1-y''(\gamma_7-\gamma_8)^2-(y''-1)^2(\gamma_7-\delta_8)^2\right]\left[(y'''-1)^2[(x''')^2(\gamma_{10}+\delta_{11}) + (x'''-1)^2(\gamma_{11}+\delta_{10})]\right]} \nonumber\\
  && \times \mathcal G_3(x,y) \left[ \mathcal G_1(x',y') + \mathcal G_2(x',y') \right] \left[ \mathcal G_1(x'',y'') + \mathcal G_2(x'',y'') \right] \mathcal G_4(x''',y''')\nonumber\\
  && + 2 \sum_{x''=0}^{1+\beta_9} \sum_{y''=0}^{1+\alpha_9} \sum_{y'''=0}^{1} \sum_{x'''=y'''(\beta_{10}+\alpha_{11}+\beta_{11})}^{\substack{y'''(\beta_{10}+\alpha_{11}+\beta_{11})+\\y'''(\gamma_{10}+\gamma_{11})+\\(y'''+1)(\gamma_{10}+\delta_{11})}} \nonumber\\
  && \times 2^{1-y'(\gamma_4-\gamma_5)^2-(y'-1)^2(\gamma_4-\delta_5)^2} 2^{1-y'''(\gamma_{10}-\gamma_{11})^2-(y'''-1)^2(\gamma_{10}-\delta_{11})^2} \nonumber\\
  && \times \mathcal G_3(x,y) \left[ \mathcal G_1(x',y') + \mathcal G_2(x',y') \right] \mathcal G_3(x'',y'') \left[ \mathcal G_1(x''',y''') + \mathcal G_2(x''',y''') \right] \nonumber\\
  && + 4 \sum_{x''=0}^{1+\beta_9} \sum_{y''=0}^{1+\alpha_9} \sum_{x'''=0}^{1+\beta_{12}} \sum_{y'''=0}^{1+\alpha_{12}} \nonumber\\
  && \times 2^{1-y'(\gamma_4-\gamma_5)^2-(y'-1)^2(\gamma_4-\delta_5)^2} \nonumber\\
  && \times \mathcal G_3(x,y) \left[ \mathcal G_1(x',y') + \mathcal G_2(x',y') \right] \mathcal G_3(x'',y'') \mathcal G_3(x''',y''')\nonumber\\
  && + 2 \sum_{x''=0}^{1+\beta_9} \sum_{y''=0}^{1+\alpha_9} \sum_{y'''=0}^1 \sum_{x'''=y'''(\delta_{10} + \delta_{11})}^{1+y'''(\gamma_{10} + \gamma_{11})} \nonumber\\
  && \times 2^{1-y'(\gamma_4-\gamma_5)^2-(y'-1)^2(\gamma_4-\delta_5)^2} 2^{(y'''-1)^2[(x''')^2(\gamma_{10}+\delta_{11}) + (x'''-1)^2(\gamma_{11}+\delta_{10})]} \nonumber\\
  && \times \mathcal G_3(x,y) \left[ \mathcal G_1(x',y') + \mathcal G_2(x',y') \right] \mathcal G_3(x'',y'') \mathcal G_4(x''',y''') \nonumber\\
  && + \sum_{y''=0}^1 \sum_{x''=y''(\delta_7 + \delta_8)}^{1+y''(\gamma_7 + \gamma_8)} \sum_{y'''=(\gamma_{10}+\delta_{11})y''}^{1-(\gamma_{10}+\gamma_{11})y''} \sum_{x'''=y'''(\beta_{10}+\alpha_{11}+\beta_{11})}^{y'''(\beta_{10}+\alpha_{11}+\beta_{11})+y'''(\gamma_{10}+\gamma_{11})+(y'''+1)(\gamma_{10}+\delta_{11})} \nonumber\\
  && \times 2^{1-y'(\gamma_4-\gamma_5)^2-(y'-1)^2(\gamma_4-\delta_5)^2} 2 \times 2^{\left[(y''-1)^2[(x'')^2(\gamma_7+\delta_8) + (x''-1)^2(\gamma_8+\delta_7)]\right]\left[1-y'''(\gamma_{10}-\gamma_{11})^2-(y'''-1)^2(\gamma_{10}-\delta_{11})^2\right]} \nonumber\\
  && \times \mathcal G_3(x,y) \left[ \mathcal G_1(x',y') + \mathcal G_2(x',y') \right] \mathcal G_4(x'',y'') \left[ \mathcal G_1(x''',y''') + \mathcal G_2(x''',y''') \right] \nonumber\\
  && + 2 \sum_{y''=0}^1 \sum_{x''=y''(\delta_7 + \delta_8)}^{1+y''(\gamma_7 + \gamma_8)} \sum_{x'''=0}^{1+\beta_{12}} \sum_{y'''=0}^{1+\alpha_{12}} \nonumber\\
  && \times 2^{1-y'(\gamma_4-\gamma_5)^2-(y'-1)^2(\gamma_4-\delta_5)^2} 2^{(y''-1)^2[(x'')^2(\gamma_7+\delta_8) + (x''-1)^2(\gamma_8+\delta_7)]} \nonumber\\
  && \times \mathcal G_3(x,y) \left[ \mathcal G_1(x',y') + \mathcal G_2(x',y') \right] \mathcal G_4(x'',y'') \mathcal G_3(x''',y''') \nonumber\\
  && + \sum_{y''=0}^1 \sum_{x''=y''(\delta_7 + \delta_8)}^{1+y''(\gamma_7 + \gamma_8)} \sum_{y'''=y''}^1 \sum_{x'''=y'''(\delta_{10} + \delta_{11})+y''}^{1+y'''(\gamma_{10} + \gamma_{11})+y''} \nonumber\\
  && \times 2^{1-y'(\gamma_4-\gamma_5)^2-(y'-1)^2(\gamma_4-\delta_5)^2} 2\times 2^{\left[(y''-1)^2[(x'')^2(\gamma_7+\delta_8) + (x''-1)^2(\gamma_8+\delta_7)]\right]\left[(y'''-1)^2[(x''')^2(\gamma_{10}+\delta_{11}) + (x'''-1)^2(\gamma_{11}+\delta_{10})]\right]} \nonumber\\
  && \times \mathcal G_3(x,y) \left[ \mathcal G_1(x',y') + \mathcal G_2(x',y') \right] \mathcal G_4(x'',y'') \mathcal G_4(x''',y''') \Bigg] \Bigg\}\nonumber
\end{eqnarray}

\begin{eqnarray}
\label{eq:12qubittgate5}
  &+& \frac{\omega}{64} \left\{4 \sum_{x=0}^{1+\beta_3} \sum_{y=0}^{1+\alpha_3} \sum_{x'=0}^{1+\beta_6} \sum_{y'=0}^{1+\alpha_6} \right. \\
  && \left[ \sum_{y''=0}^{1} \sum_{x''=y''(\beta_7+\alpha_8+\beta_8)}^{\substack{y''(\beta_7+\alpha_8+\beta_8)+\\y''(\gamma_7+\gamma_8)+\\(y''+1)(\gamma_7+\delta_8)}} \sum_{y'''=(\gamma_{10}+\delta_{11})x''}^{1-(\gamma_{10}+\gamma_{11})x''} \sum_{x'''=y'''(\beta_{10}+\alpha_{11}+\beta_{11})}^{\substack{y'''(\beta_{10}+\alpha_{11}+\beta_{11})+\\y'''(\gamma_{10}+\gamma_{11})+\\(y'''+1)(\gamma_{10}+\delta_{11})}} \right. \nonumber\\
  && \times 2\times 2^{\left[1-y''(\gamma_7-\gamma_8)^2-(y''-1)^2(\gamma_7-\delta_8)^2\right]\left[1-y'''(\gamma_{10}-\gamma_{11})^2-(y'''-1)^2(\gamma_{10}-\delta_{11})^2\right]} \nonumber\\
  && \times \mathcal G_3(x,y) \mathcal G_3(x',y') \left[ \mathcal G_1(x'',y'') + \mathcal G_2(x'',y'') \right] \left[ \mathcal G_1(x''',y''') + \mathcal G_2(x''',y''') \right] \nonumber\\
  && + 2 \sum_{y''=0}^{1} \sum_{x''=y''(\beta_7+\alpha_8+\beta_8)}^{\substack{y''(\beta_7+\alpha_8+\beta_8)+\\y''(\gamma_7+\gamma_8)+\\(y''+1)(\gamma_7+\delta_8)}} \sum_{x'''=0}^{1+\beta_{12}} \sum_{y'''=0}^{1+\alpha_{12}} \nonumber\\
  && \times 2^{1-y'''(\gamma_{10}-\gamma_{11})^2-(y'''-1)^2(\gamma_{10}-\delta_{11})^2} \nonumber\\
  && \times \mathcal G_3(x,y) \mathcal G_3(x',y') \left[ \mathcal G_1(x'',y'') + \mathcal G_2(x'',y'') \right]  \mathcal G_3(x''',y''') \nonumber\\
  && + \sum_{y''=0}^{1} \sum_{x''=y''(\beta_7+\alpha_8+\beta_8)}^{\substack{y''(\beta_7+\alpha_8+\beta_8)+\\y''(\gamma_7+\gamma_8)+\\(y''+1)(\gamma_7+\delta_8)}} \sum_{y'''=x''(1+\alpha_7+\beta_7)}^1 \sum_{x'''=y'''(\delta_{10} + \delta_{11})+x''(1+\alpha_7+\beta_7)}^{1+y'''(\gamma_{10} + \gamma_{11})+x''(1+\alpha_7+\beta_7)}\nonumber\\
  && \times 2\times 2^{\left[1-y''(\gamma_7-\gamma_8)^2-(y''-1)^2(\gamma_7-\delta_8)^2\right]\left[(y'''-1)^2[(x''')^2(\gamma_{10}+\delta_{11}) + (x'''-1)^2(\gamma_{11}+\delta_{10})]\right]} \nonumber\\
  && \times \mathcal G_3(x,y) \mathcal G_3(x',y') \left[ \mathcal G_1(x'',y'') + \mathcal G_2(x'',y'') \right] \mathcal G_4(x''',y''')\nonumber\\
  && + 2 \sum_{x''=0}^{1+\beta_9} \sum_{y''=0}^{1+\alpha_9} \sum_{y'''=0}^{1} \sum_{x'''=y'''(\beta_{10}+\alpha_{11}+\beta_{11})}^{\substack{y'''(\beta_{10}+\alpha_{11}+\beta_{11})+\\y'''(\gamma_{10}+\gamma_{11})+\\(y'''+1)(\gamma_{10}+\delta_{11})}} \nonumber\\
  && \times 2^{1-y'''(\gamma_{10}-\gamma_{11})^2-(y'''-1)^2(\gamma_{10}-\delta_{11})^2} \nonumber\\
  && \times \mathcal G_3(x,y) \mathcal G_3(x',y') \mathcal G_3(x'',y'') \left[ \mathcal G_1(x''',y''') + \mathcal G_2(x''',y''') \right] \nonumber\\
  && + 4 \sum_{x''=0}^{1+\beta_9} \sum_{y''=0}^{1+\alpha_9} \sum_{x'''=0}^{1+\beta_{12}} \sum_{y'''=0}^{1+\alpha_{12}} \nonumber\\
%  && \nonumber\\
  && \times \mathcal G_3(x,y) \mathcal G_3(x',y') \mathcal G_3(x'',y'') \mathcal G_3(x''',y''')\nonumber\\
  && + 2 \sum_{x''=0}^{1+\beta_9} \sum_{y''=0}^{1+\alpha_9} \sum_{y'''=0}^1 \sum_{x'''=y'''(\delta_{10} + \delta_{11})}^{1+y'''(\gamma_{10} + \gamma_{11})} \nonumber\\
  && \times 2^{(y'''-1)^2[(x''')^2(\gamma_{10}+\delta_{11}) + (x'''-1)^2(\gamma_{11}+\delta_{10})]} \nonumber\\
  && \times \mathcal G_3(x,y) \mathcal G_3(x',y') \mathcal G_3(x'',y'') \mathcal G_4(x''',y''') \nonumber\\
  && + \sum_{y''=0}^1 \sum_{x''=y''(\delta_7 + \delta_8)}^{1+y''(\gamma_7 + \gamma_8)} \sum_{y'''=(\gamma_{10}+\delta_{11})y''}^{1-(\gamma_{10}+\gamma_{11})y''} \sum_{x'''=y'''(\beta_{10}+\alpha_{11}+\beta_{11})}^{y'''(\beta_{10}+\alpha_{11}+\beta_{11})+y'''(\gamma_{10}+\gamma_{11})+(y'''+1)(\gamma_{10}+\delta_{11})} \nonumber\\
  && \times 2 \times 2^{\left[(y''-1)^2[(x'')^2(\gamma_7+\delta_8) + (x''-1)^2(\gamma_8+\delta_7)]\right]\left[1-y'''(\gamma_{10}-\gamma_{11})^2-(y'''-1)^2(\gamma_{10}-\delta_{11})^2\right]} \nonumber\\
  && \times \mathcal G_3(x,y) \mathcal G_3(x',y') \mathcal G_4(x'',y'') \left[ \mathcal G_1(x''',y''') + \mathcal G_2(x''',y''') \right] \nonumber\\
  && + 2 \sum_{y''=0}^1 \sum_{x''=y''(\delta_7 + \delta_8)}^{1+y''(\gamma_7 + \gamma_8)} \sum_{x'''=0}^{1+\beta_{12}} \sum_{y'''=0}^{1+\alpha_{12}} \nonumber\\
  && \times 2^{(y''-1)^2[(x'')^2(\gamma_7+\delta_8) + (x''-1)^2(\gamma_8+\delta_7)]} \nonumber\\
  && \times \mathcal G_3(x,y) \mathcal G_3(x',y') \mathcal G_4(x'',y'') \mathcal G_3(x''',y''') \nonumber\\
  && + \sum_{y''=0}^1 \sum_{x''=y''(\delta_7 + \delta_8)}^{1+y''(\gamma_7 + \gamma_8)} \sum_{y'''=y''}^1 \sum_{x'''=y'''(\delta_{10} + \delta_{11})+y''}^{1+y'''(\gamma_{10} + \gamma_{11})+y''} \nonumber\\
  && \times 2\times 2^{\left[(y''-1)^2[(x'')^2(\gamma_7+\delta_8) + (x''-1)^2(\gamma_8+\delta_7)]\right]\left[(y'''-1)^2[(x''')^2(\gamma_{10}+\delta_{11}) + (x'''-1)^2(\gamma_{11}+\delta_{10})]\right]} \nonumber\\
  && \times \mathcal G_3(x,y) \mathcal G_3(x',y') \mathcal G_4(x'',y'') \mathcal G_4(x''',y''') \Bigg] \Bigg\} \nonumber
\end{eqnarray}

\begin{eqnarray}
\label{eq:12qubittgate6}
  &+& \frac{\omega}{64} \left\{2 \sum_{x=0}^{1+\beta_3} \sum_{y=0}^{1+\alpha_3} \sum_{y'=0}^1 \sum_{x'=y'(\delta_4 + \delta_5)}^{1+y'(\gamma_4 + \gamma_5)} \right.\\
  && \left[ \sum_{y''=0}^{1} \sum_{x''=y''(\beta_7+\alpha_8+\beta_8)}^{\substack{y''(\beta_7+\alpha_8+\beta_8)+\\y''(\gamma_7+\gamma_8)+\\(y''+1)(\gamma_7+\delta_8)}} \sum_{y'''=(\gamma_{10}+\delta_{11})x''}^{1-(\gamma_{10}+\gamma_{11})x''} \sum_{x'''=y'''(\beta_{10}+\alpha_{11}+\beta_{11})}^{\substack{y'''(\beta_{10}+\alpha_{11}+\beta_{11})+\\y'''(\gamma_{10}+\gamma_{11})+\\(y'''+1)(\gamma_{10}+\delta_{11})}} \right. \nonumber\\
  && \times 2^{(y'-1)^2[(x')^2(\gamma_4+\delta_5) + (x'-1)^2(\gamma_5+\delta_4)]} 2\times 2^{\left[1-y''(\gamma_7-\gamma_8)^2-(y''-1)^2(\gamma_7-\delta_8)^2\right]\left[1-y'''(\gamma_{10}-\gamma_{11})^2-(y'''-1)^2(\gamma_{10}-\delta_{11})^2\right]} \nonumber\\
  && \times \mathcal G_3(x,y) \mathcal G_4(x',y') \left[ \mathcal G_1(x'',y'') + \mathcal G_2(x'',y'') \right] \left[ \mathcal G_1(x''',y''') + \mathcal G_2(x''',y''') \right] \nonumber\\
  && + 2 \sum_{y''=0}^{1} \sum_{x''=y''(\beta_7+\alpha_8+\beta_8)}^{\substack{y''(\beta_7+\alpha_8+\beta_8)+\\y''(\gamma_7+\gamma_8)+\\(y''+1)(\gamma_7+\delta_8)}} \sum_{x'''=0}^{1+\beta_{12}} \sum_{y'''=0}^{1+\alpha_{12}} \nonumber\\
  && \times 2^{(y'-1)^2[(x')^2(\gamma_4+\delta_5) + (x'-1)^2(\gamma_5+\delta_4)]} 2^{1-y'''(\gamma_{10}-\gamma_{11})^2-(y'''-1)^2(\gamma_{10}-\delta_{11})^2} \nonumber\\
  && \times \mathcal G_3(x,y) \mathcal G_4(x',y') \left[ \mathcal G_1(x'',y'') + \mathcal G_2(x'',y'') \right]  \mathcal G_3(x''',y''') \nonumber\\
  && + \sum_{y''=0}^{1} \sum_{x''=y''(\beta_7+\alpha_8+\beta_8)}^{\substack{y''(\beta_7+\alpha_8+\beta_8)+\\y''(\gamma_7+\gamma_8)+\\(y''+1)(\gamma_7+\delta_8)}} \sum_{y'''=x''(1+\alpha_7+\beta_7)}^1 \sum_{x'''=y'''(\delta_{10} + \delta_{11})+x''(1+\alpha_7+\beta_7)}^{1+y'''(\gamma_{10} + \gamma_{11})+x''(1+\alpha_7+\beta_7)}\nonumber\\
  && \times 2^{(y'-1)^2[(x')^2(\gamma_4+\delta_5) + (x'-1)^2(\gamma_5+\delta_4)]} 2\times 2^{\left[1-y''(\gamma_7-\gamma_8)^2-(y''-1)^2(\gamma_7-\delta_8)^2\right]\left[(y'''-1)^2[(x''')^2(\gamma_{10}+\delta_{11}) + (x'''-1)^2(\gamma_{11}+\delta_{10})]\right]} \nonumber\\
  && \times \mathcal G_3(x,y) \mathcal G_4(x',y') \left[ \mathcal G_1(x'',y'') + \mathcal G_2(x'',y'') \right] \mathcal G_4(x''',y''')\nonumber\\
  && + 2 \sum_{x''=0}^{1+\beta_9} \sum_{y''=0}^{1+\alpha_9} \sum_{y'''=0}^{1} \sum_{x'''=y'''(\beta_{10}+\alpha_{11}+\beta_{11})}^{\substack{y'''(\beta_{10}+\alpha_{11}+\beta_{11})+\\y'''(\gamma_{10}+\gamma_{11})+\\(y'''+1)(\gamma_{10}+\delta_{11})}} \nonumber\\
  && \times 2^{(y'-1)^2[(x')^2(\gamma_4+\delta_5) + (x'-1)^2(\gamma_5+\delta_4)]} 2^{1-y'''(\gamma_{10}-\gamma_{11})^2-(y'''-1)^2(\gamma_{10}-\delta_{11})^2} \nonumber\\
  && \times \mathcal G_3(x,y) \mathcal G_4(x',y') \mathcal G_3(x'',y'') \left[ \mathcal G_1(x''',y''') + \mathcal G_2(x''',y''') \right] \nonumber\\
  && + 4 \sum_{x''=0}^{1+\beta_9} \sum_{y''=0}^{1+\alpha_9} \sum_{x'''=0}^{1+\beta_{12}} \sum_{y'''=0}^{1+\alpha_{12}} \nonumber\\
  && \times 2^{(y'-1)^2[(x')^2(\gamma_4+\delta_5) + (x'-1)^2(\gamma_5+\delta_4)]} \nonumber\\
  && \times \mathcal G_3(x,y) \mathcal G_4(x',y') \mathcal G_3(x'',y'') \mathcal G_3(x''',y''')\nonumber\\
  && + 2 \sum_{x''=0}^{1+\beta_9} \sum_{y''=0}^{1+\alpha_9} \sum_{y'''=0}^1 \sum_{x'''=y'''(\delta_{10} + \delta_{11})}^{1+y'''(\gamma_{10} + \gamma_{11})} \nonumber\\
  && \times 2^{(y'-1)^2[(x')^2(\gamma_4+\delta_5) + (x'-1)^2(\gamma_5+\delta_4)]} 2^{(y'''-1)^2[(x''')^2(\gamma_{10}+\delta_{11}) + (x'''-1)^2(\gamma_{11}+\delta_{10})]} \nonumber\\
  && \times \mathcal G_3(x,y) \mathcal G_4(x',y') \mathcal G_3(x'',y'') \mathcal G_4(x''',y''') \nonumber\\
  && + \sum_{y''=0}^1 \sum_{x''=y''(\delta_7 + \delta_8)}^{1+y''(\gamma_7 + \gamma_8)} \sum_{y'''=(\gamma_{10}+\delta_{11})y''}^{1-(\gamma_{10}+\gamma_{11})y''} \sum_{x'''=y'''(\beta_{10}+\alpha_{11}+\beta_{11})}^{y'''(\beta_{10}+\alpha_{11}+\beta_{11})+y'''(\gamma_{10}+\gamma_{11})+(y'''+1)(\gamma_{10}+\delta_{11})} \nonumber\\
  && \times 2^{(y'-1)^2[(x')^2(\gamma_4+\delta_5) + (x'-1)^2(\gamma_5+\delta_4)]} 2 \times 2^{\left[(y''-1)^2[(x'')^2(\gamma_7+\delta_8) + (x''-1)^2(\gamma_8+\delta_7)]\right]\left[1-y'''(\gamma_{10}-\gamma_{11})^2-(y'''-1)^2(\gamma_{10}-\delta_{11})^2\right]} \nonumber\\
  && \times \mathcal G_3(x,y) \mathcal G_4(x',y') \mathcal G_4(x'',y'') \left[ \mathcal G_1(x''',y''') + \mathcal G_2(x''',y''') \right] \nonumber\\
  && + 2 \sum_{y''=0}^1 \sum_{x''=y''(\delta_7 + \delta_8)}^{1+y''(\gamma_7 + \gamma_8)} \sum_{x'''=0}^{1+\beta_{12}} \sum_{y'''=0}^{1+\alpha_{12}} \nonumber\\
  && \times 2^{(y'-1)^2[(x')^2(\gamma_4+\delta_5) + (x'-1)^2(\gamma_5+\delta_4)]} 2^{(y''-1)^2[(x'')^2(\gamma_7+\delta_8) + (x''-1)^2(\gamma_8+\delta_7)]} \nonumber\\
  && \times \mathcal G_3(x,y) \mathcal G_4(x',y') \mathcal G_4(x'',y'') \mathcal G_3(x''',y''') \nonumber\\
  && + \sum_{y''=0}^1 \sum_{x''=y''(\delta_7 + \delta_8)}^{1+y''(\gamma_7 + \gamma_8)} \sum_{y'''=y''}^1 \sum_{x'''=y'''(\delta_{10} + \delta_{11})+y''}^{1+y'''(\gamma_{10} + \gamma_{11})+y''} \nonumber\\
  && \times 2^{(y'-1)^2[(x')^2(\gamma_4+\delta_5) + (x'-1)^2(\gamma_5+\delta_4)]} 2 \times 2^{\left[(y''-1)^2[(x'')^2(\gamma_7+\delta_8) + (x''-1)^2(\gamma_8+\delta_7)]\right]\left[(y'''-1)^2[(x''')^2(\gamma_{10}+\delta_{11}) + (x'''-1)^2(\gamma_{11}+\delta_{10})]\right]} \nonumber\\
  && \times \mathcal G_3(x,y) \mathcal G_4(x',y') \mathcal G_4(x'',y'') \mathcal G_4(x''',y''') \Bigg] \Bigg\} \nonumber
\end{eqnarray}

\begin{eqnarray}
\label{eq:12qubittgate7}
  &+& \frac{\omega}{64} \left\{\sum_{y=0}^1 \sum_{x=y(\delta_1 + \delta_2)}^{1+y(\gamma_1 + \gamma_2)} \sum_{y'=(\gamma_4+\delta_5)y}^{1-(\gamma_4+\gamma_5)y} \sum_{x'=y'(\beta_4+\alpha_5+\beta_5)}^{y'(\beta_4+\alpha_5+\beta_5)+y'(\gamma_4+\gamma_5)+(y'+1)(\gamma_4+\delta_5)} \right. \\
  && \left[ \sum_{y''=(\gamma_7+\delta_8)(y+x')}^{1+(\gamma_7+\gamma_8)(y+x')} \sum_{x''=y''(\beta_7+\alpha_8+\beta_8)}^{\substack{y''(\beta_7+\alpha_8+\beta_8)+\\y''(\gamma_7+\gamma_8)+\\(y''+1)(\gamma_7+\delta_8)}} \sum_{y'''=(\gamma_{10}+\delta_{11})(y+x'+x'')}^{1+(\gamma_{10}+\gamma_{11})(y+x'+x'')} \sum_{x'''=y'''(\beta_{10}+\alpha_{11}+\beta_{11})}^{\substack{y'''(\beta_{10}+\alpha_{11}+\beta_{11})+\\y'''(\gamma_{10}+\gamma_{11})+\\(y'''+1)(\gamma_{10}+\delta_{11})}} \right. \nonumber\\
  && \times 4 \times 2 \times 2^{\left[(y-1)^2[(x)^2(\gamma_1+\delta_2) + (x-1)^2(\gamma_2+\delta_1)]\right]\left[1-y'(\gamma_4-\gamma_5)^2-(y'-1)^2(\gamma_4-\delta_5)^2\right] \left[1-y''(\gamma_7-\gamma_8)^2-(y''-1)^2(\gamma_7-\delta_8)^2\right]\left[1-y'''(\gamma_{10}-\gamma_{11})^2-(y'''-1)^2(\gamma_{10}-\delta_{11})^2\right]} \nonumber\\
  && \times \mathcal G_4(x,y) \left[ \mathcal G_1(x',y') + \mathcal G_2(x',y') \right] \left[ \mathcal G_1(x'',y'') + \mathcal G_2(x'',y'') \right] \left[ \mathcal G_1(x''',y''') + \mathcal G_2(x''',y''') \right] \nonumber\\
  && + 2 \sum_{y''=0}^{1} \sum_{x''=y''(\beta_7+\alpha_8+\beta_8)}^{\substack{y''(\beta_7+\alpha_8+\beta_8)+\\y''(\gamma_7+\gamma_8)+\\(y''+1)(\gamma_7+\delta_8)}} \sum_{x'''=0}^{1+\beta_{12}} \sum_{y'''=0}^{1+\alpha_{12}} \nonumber\\
  && \times 2 \times 2^{\left[(y-1)^2[(x)^2(\gamma_1+\delta_2) + (x-1)^2(\gamma_2+\delta_1)]\right]\left[1-y'(\gamma_4-\gamma_5)^2-(y'-1)^2(\gamma_4-\delta_5)^2\right]} 2^{1-y'''(\gamma_{10}-\gamma_{11})^2-(y'''-1)^2(\gamma_{10}-\delta_{11})^2} \nonumber\\
  && \times \mathcal G_4(x,y) \left[ \mathcal G_1(x',y') + \mathcal G_2(x',y') \right] \left[ \mathcal G_1(x'',y'') + \mathcal G_2(x'',y'') \right]  \mathcal G_3(x''',y''') \nonumber\\
  && + \sum_{y''=(\gamma_7+\delta_8)(y+x')}^{1+(\gamma_7+\gamma_8)(y+x')} \sum_{x''=y''(\beta_7+\alpha_8+\beta_8)}^{\substack{y''(\beta_7+\alpha_8+\beta_8)+\\y''(\gamma_7+\gamma_8)+\\(y''+1)(\gamma_7+\delta_8)}} \sum_{y'''=y+x'(1+\alpha_4+\beta_4)+x''(1+\alpha_7+\beta_7)}^1 \sum_{x'''=y'''(\delta_{10} + \delta_{11})+(y+x'(1+\alpha_4+\beta_4)+x''(1+\alpha_7+\beta_7))}^{1+y'''(\gamma_{10} + \gamma_{11})+(y+x'(1+\alpha_4+\beta_4)+x''(1+\alpha_7+\beta_7))}\nonumber\\
  && \times 4 \times 2 \times 2^{\left[(y-1)^2[(x)^2(\gamma_1+\delta_2) + (x-1)^2(\gamma_2+\delta_1)]\right]\left[1-y'(\gamma_4-\gamma_5)^2-(y'-1)^2(\gamma_4-\delta_5)^2\right] \left[1-y''(\gamma_7-\gamma_8)^2-(y''-1)^2(\gamma_7-\delta_8)^2\right]\left[(y'''-1)^2[(x''')^2(\gamma_{10}+\delta_{11}) + (x'''-1)^2(\gamma_{11}+\delta_{10})]\right]} \nonumber\\
  && \times \mathcal G_4(x,y) \left[ \mathcal G_1(x',y') + \mathcal G_2(x',y') \right] \left[ \mathcal G_1(x'',y'') + \mathcal G_2(x'',y'') \right] \mathcal G_4(x''',y''')\nonumber\\
  && + 2 \sum_{x''=0}^{1+\beta_9} \sum_{y''=0}^{1+\alpha_9} \sum_{y'''=0}^{1} \sum_{x'''=y'''(\beta_{10}+\alpha_{11}+\beta_{11})}^{\substack{y'''(\beta_{10}+\alpha_{11}+\beta_{11})+\\y'''(\gamma_{10}+\gamma_{11})+\\(y'''+1)(\gamma_{10}+\delta_{11})}} \nonumber\\
  && \times 2 \times 2^{\left[(y-1)^2[(x)^2(\gamma_1+\delta_2) + (x-1)^2(\gamma_2+\delta_1)]\right]\left[1-y'(\gamma_4-\gamma_5)^2-(y'-1)^2(\gamma_4-\delta_5)^2\right]} 2^{1-y'''(\gamma_{10}-\gamma_{11})^2-(y'''-1)^2(\gamma_{10}-\delta_{11})^2} \nonumber\\
  && \times \mathcal G_4(x,y) \left[ \mathcal G_1(x',y') + \mathcal G_2(x',y') \right] \mathcal G_3(x'',y'') \left[ \mathcal G_1(x''',y''') + \mathcal G_2(x''',y''') \right] \nonumber\\
  && + 4 \sum_{x''=0}^{1+\beta_9} \sum_{y''=0}^{1+\alpha_9} \sum_{x'''=0}^{1+\beta_{12}} \sum_{y'''=0}^{1+\alpha_{12}} \nonumber\\
  && \times 2 \times 2^{\left[(y-1)^2[(x)^2(\gamma_1+\delta_2) + (x-1)^2(\gamma_2+\delta_1)]\right]\left[1-y'(\gamma_4-\gamma_5)^2-(y'-1)^2(\gamma_4-\delta_5)^2\right]} \nonumber\\
  && \times \mathcal G_4(x,y) \left[ \mathcal G_1(x',y') + \mathcal G_2(x',y') \right] \mathcal G_3(x'',y'') \mathcal G_3(x''',y''')\nonumber\\
  && + 2 \sum_{x''=0}^{1+\beta_9} \sum_{y''=0}^{1+\alpha_9} \sum_{y'''=0}^1 \sum_{x'''=y'''(\delta_{10} + \delta_{11})}^{1+y'''(\gamma_{10} + \gamma_{11})} \nonumber\\
  && \times 2 \times 2^{\left[(y-1)^2[(x)^2(\gamma_1+\delta_2) + (x-1)^2(\gamma_2+\delta_1)]\right]\left[1-y'(\gamma_4-\gamma_5)^2-(y'-1)^2(\gamma_4-\delta_5)^2\right]} 2^{(y'''-1)^2[(x''')^2(\gamma_{10}+\delta_{11}) + (x'''-1)^2(\gamma_{11}+\delta_{10})]} \nonumber\\
  && \times \mathcal G_4(x,y) \left[ \mathcal G_1(x',y') + \mathcal G_2(x',y') \right] \mathcal G_3(x'',y'') \mathcal G_4(x''',y''') \nonumber\\
  && + \sum_{y''=y+x'(1+\alpha_4+\beta_4)}^1 \sum_{x''=y''(\delta_7 + \delta_8)+(y+x'(1+\alpha_4+\beta_4))}^{1+y''(\gamma_7 + \gamma_8)+(y+x'(1+\alpha_4+\beta_4))} \sum_{y'''=(\gamma_{10}+\delta_{11})(y+x'+y'')}^{1-(\gamma_{10}+\gamma_{11})(y+x'+y'')} \sum_{x'''=y'''(\beta_{10}+\alpha_{11}+\beta_{11})}^{y'''(\beta_{10}+\alpha_{11}+\beta_{11})+y'''(\gamma_{10}+\gamma_{11})+(y'''+1)(\gamma_{10}+\delta_{11})} \nonumber\\
  && \times 4 \times 2 \times 2^{\left[(y-1)^2[(x)^2(\gamma_1+\delta_2) + (x-1)^2(\gamma_2+\delta_1)]\right]\left[1-y'(\gamma_4-\gamma_5)^2-(y'-1)^2(\gamma_4-\delta_5)^2\right] \left[(y''-1)^2[(x'')^2(\gamma_7+\delta_8) + (x''-1)^2(\gamma_8+\delta_7)]\right]\left[1-y'''(\gamma_{10}-\gamma_{11})^2-(y'''-1)^2(\gamma_{10}-\delta_{11})^2\right]} \nonumber\\
  && \times \mathcal G_4(x,y) \left[ \mathcal G_1(x',y') + \mathcal G_2(x',y') \right] \mathcal G_4(x'',y'') \left[ \mathcal G_1(x''',y''') + \mathcal G_2(x''',y''') \right] \nonumber\\
  && + 2 \sum_{y''=0}^1 \sum_{x''=y''(\delta_7 + \delta_8)}^{1+y''(\gamma_7 + \gamma_8)} \sum_{x'''=0}^{1+\beta_{12}} \sum_{y'''=0}^{1+\alpha_{12}} \nonumber\\
  && \times 2 \times 2^{\left[(y-1)^2[(x)^2(\gamma_1+\delta_2) + (x-1)^2(\gamma_2+\delta_1)]\right]\left[1-y'(\gamma_4-\gamma_5)^2-(y'-1)^2(\gamma_4-\delta_5)^2\right]} 2^{(y''-1)^2[(x'')^2(\gamma_7+\delta_8) + (x''-1)^2(\gamma_8+\delta_7)]} \nonumber\\
  && \times \mathcal G_4(x,y) \left[ \mathcal G_1(x',y') + \mathcal G_2(x',y') \right] \mathcal G_4(x'',y'') \mathcal G_3(x''',y''') \nonumber\\
  && + \sum_{y''=y+x'(1+\alpha_4+\beta_4)}^1 \sum_{x''=y''(\delta_7 + \delta_8)+(y+x'(1+\alpha_4+\beta_4))}^{1+y''(\gamma_7 + \gamma_8)+(y+x'(1+\alpha_4+\beta_4))} \sum_{y'''=y+x'(1+\alpha_4+\beta_4)+y''}^1 \sum_{x'''=y'''(\delta_{10} + \delta_{11})+(y+x'(1+\alpha_4+\beta_4)+y'')}^{1+y'''(\gamma_{10} + \gamma_{11})+(y+x'(1+\alpha_4+\beta_4)+y'')} \nonumber\\
  && \times 4 \times 2 \times 2^{\left[(y-1)^2[(x)^2(\gamma_1+\delta_2) + (x-1)^2(\gamma_2+\delta_1)]\right]\left[1-y'(\gamma_4-\gamma_5)^2-(y'-1)^2(\gamma_4-\delta_5)^2\right] \left[(y''-1)^2[(x'')^2(\gamma_7+\delta_8) + (x''-1)^2(\gamma_8+\delta_7)]\right]\left[(y'''-1)^2[(x''')^2(\gamma_{10}+\delta_{11}) + (x'''-1)^2(\gamma_{11}+\delta_{10})]\right]} \nonumber\\
  && \times \mathcal G_4(x,y) \left[ \mathcal G_1(x',y') + \mathcal G_2(x',y') \right] \mathcal G_4(x'',y'') \mathcal G_4(x''',y''') \Bigg] \Bigg\} \nonumber
\end{eqnarray}

\begin{eqnarray}
\label{eq:12qubittgate8}
  &+& \frac{\omega}{64} \left\{2 \sum_{y=0}^1 \sum_{x=y(\delta_1 + \delta_2)}^{1+y(\gamma_1 + \gamma_2)} \sum_{x'=0}^{1+\beta_6} \sum_{y'=0}^{1+\alpha_6} \right. \\
  && \left[ \sum_{y''=0}^{1} \sum_{x''=y''(\beta_7+\alpha_8+\beta_8)}^{\substack{y''(\beta_7+\alpha_8+\beta_8)+\\y''(\gamma_7+\gamma_8)+\\(y''+1)(\gamma_7+\delta_8)}} \sum_{y'''=(\gamma_{10}+\delta_{11})x''}^{1-(\gamma_{10}+\gamma_{11})x''} \sum_{x'''=y'''(\beta_{10}+\alpha_{11}+\beta_{11})}^{\substack{y'''(\beta_{10}+\alpha_{11}+\beta_{11})+\\y'''(\gamma_{10}+\gamma_{11})+\\(y'''+1)(\gamma_{10}+\delta_{11})}} \right. \nonumber\\
  && \times 2^{(y-1)^2[(x)^2(\gamma_1+\delta_2) + (x-1)^2(\gamma_2+\delta_1)]} 2\times 2^{\left[1-y''(\gamma_7-\gamma_8)^2-(y''-1)^2(\gamma_7-\delta_8)^2\right]\left[1-y'''(\gamma_{10}-\gamma_{11})^2-(y'''-1)^2(\gamma_{10}-\delta_{11})^2\right]} \nonumber\\
  && \times \mathcal G_4(x,y) \mathcal G_3(x',y') \left[ \mathcal G_1(x'',y'') + \mathcal G_2(x'',y'') \right] \left[ \mathcal G_1(x''',y''') + \mathcal G_2(x''',y''') \right] \nonumber\\
  && + 2 \sum_{y''=0}^{1} \sum_{x''=y''(\beta_7+\alpha_8+\beta_8)}^{\substack{y''(\beta_7+\alpha_8+\beta_8)+\\y''(\gamma_7+\gamma_8)+\\(y''+1)(\gamma_7+\delta_8)}} \sum_{x'''=0}^{1+\beta_{12}} \sum_{y'''=0}^{1+\alpha_{12}} \nonumber\\
  && \times 2^{(y-1)^2[(x)^2(\gamma_1+\delta_2) + (x-1)^2(\gamma_2+\delta_1)]} 2^{1-y'''(\gamma_{10}-\gamma_{11})^2-(y'''-1)^2(\gamma_{10}-\delta_{11})^2} \nonumber\\
  && \times \mathcal G_4(x,y) \mathcal G_3(x',y') \left[ \mathcal G_1(x'',y'') + \mathcal G_2(x'',y'') \right]  \mathcal G_3(x''',y''') \nonumber\\
  && + \sum_{y''=0}^{1} \sum_{x''=y''(\beta_7+\alpha_8+\beta_8)}^{\substack{y''(\beta_7+\alpha_8+\beta_8)+\\y''(\gamma_7+\gamma_8)+\\(y''+1)(\gamma_7+\delta_8)}} \sum_{y'''=x''(1+\alpha_7+\beta_7)}^1 \sum_{x'''=y'''(\delta_{10} + \delta_{11})+x''(1+\alpha_7+\beta_7)}^{1+y'''(\gamma_{10} + \gamma_{11})+x''(1+\alpha_7+\beta_7)}\nonumber\\
  && \times 2^{(y-1)^2[(x)^2(\gamma_1+\delta_2) + (x-1)^2(\gamma_2+\delta_1)]} 2\times 2^{\left[1-y''(\gamma_7-\gamma_8)^2-(y''-1)^2(\gamma_7-\delta_8)^2\right]\left[(y'''-1)^2[(x''')^2(\gamma_{10}+\delta_{11}) + (x'''-1)^2(\gamma_{11}+\delta_{10})]\right]} \nonumber\\
  && \times \mathcal G_4(x,y) \mathcal G_3(x',y') \left[ \mathcal G_1(x'',y'') + \mathcal G_2(x'',y'') \right] \mathcal G_4(x''',y''')\nonumber\\
  && + 2 \sum_{x''=0}^{1+\beta_9} \sum_{y''=0}^{1+\alpha_9} \sum_{y'''=0}^{1} \sum_{x'''=y'''(\beta_{10}+\alpha_{11}+\beta_{11})}^{\substack{y'''(\beta_{10}+\alpha_{11}+\beta_{11})+\\y'''(\gamma_{10}+\gamma_{11})+\\(y'''+1)(\gamma_{10}+\delta_{11})}} \nonumber\\
  && \times 2^{(y-1)^2[(x)^2(\gamma_1+\delta_2) + (x-1)^2(\gamma_2+\delta_1)]} 2^{1-y'''(\gamma_{10}-\gamma_{11})^2-(y'''-1)^2(\gamma_{10}-\delta_{11})^2} \nonumber\\
  && \times \mathcal G_4(x,y) \mathcal G_3(x',y') \mathcal G_3(x'',y'') \left[ \mathcal G_1(x''',y''') + \mathcal G_2(x''',y''') \right] \nonumber\\
  && + 4 \sum_{x''=0}^{1+\beta_9} \sum_{y''=0}^{1+\alpha_9} \sum_{x'''=0}^{1+\beta_{12}} \sum_{y'''=0}^{1+\alpha_{12}} \nonumber\\
  && \times 2^{(y-1)^2[(x)^2(\gamma_1+\delta_2) + (x-1)^2(\gamma_2+\delta_1)]} \nonumber\\
  && \times \mathcal G_4(x,y) \mathcal G_3(x',y') \mathcal G_3(x'',y'') \mathcal G_3(x''',y''')\nonumber\\
  && + 2 \sum_{x''=0}^{1+\beta_9} \sum_{y''=0}^{1+\alpha_9} \sum_{y'''=0}^1 \sum_{x'''=y'''(\delta_{10} + \delta_{11})}^{1+y'''(\gamma_{10} + \gamma_{11})} \nonumber\\
  && \times 2^{(y-1)^2[(x)^2(\gamma_1+\delta_2) + (x-1)^2(\gamma_2+\delta_1)]} 2^{(y'''-1)^2[(x''')^2(\gamma_{10}+\delta_{11}) + (x'''-1)^2(\gamma_{11}+\delta_{10})]} \nonumber\\
  && \times \mathcal G_4(x,y) \mathcal G_3(x',y') \mathcal G_3(x'',y'') \mathcal G_4(x''',y''') \nonumber\\
  && + \sum_{y''=0}^1 \sum_{x''=y''(\delta_7 + \delta_8)}^{1+y''(\gamma_7 + \gamma_8)} \sum_{y'''=(\gamma_{10}+\delta_{11})y''}^{1-(\gamma_{10}+\gamma_{11})y''} \sum_{x'''=y'''(\beta_{10}+\alpha_{11}+\beta_{11})}^{y'''(\beta_{10}+\alpha_{11}+\beta_{11})+y'''(\gamma_{10}+\gamma_{11})+(y'''+1)(\gamma_{10}+\delta_{11})} \nonumber\\
  && \times 2^{(y-1)^2[(x)^2(\gamma_1+\delta_2) + (x-1)^2(\gamma_2+\delta_1)]} 2 \times 2^{\left[(y''-1)^2[(x'')^2(\gamma_7+\delta_8) + (x''-1)^2(\gamma_8+\delta_7)]\right]\left[1-y'''(\gamma_{10}-\gamma_{11})^2-(y'''-1)^2(\gamma_{10}-\delta_{11})^2\right]} \nonumber\\
  && \times \mathcal G_4(x,y) \mathcal G_3(x',y') \mathcal G_4(x'',y'') \left[ \mathcal G_1(x''',y''') + \mathcal G_2(x''',y''') \right] \nonumber\\
  && + 2 \sum_{y''=0}^1 \sum_{x''=y''(\delta_7 + \delta_8)}^{1+y''(\gamma_7 + \gamma_8)} \sum_{x'''=0}^{1+\beta_{12}} \sum_{y'''=0}^{1+\alpha_{12}} \nonumber\\
  && \times 2^{(y-1)^2[(x)^2(\gamma_1+\delta_2) + (x-1)^2(\gamma_2+\delta_1)]} 2^{(y''-1)^2[(x'')^2(\gamma_7+\delta_8) + (x''-1)^2(\gamma_8+\delta_7)]} \nonumber\\
  && \times \mathcal G_4(x,y) \mathcal G_3(x',y') \mathcal G_4(x'',y'') \mathcal G_3(x''',y''') \nonumber\\
  && + \sum_{y''=0}^1 \sum_{x''=y''(\delta_7 + \delta_8)}^{1+y''(\gamma_7 + \gamma_8)} \sum_{y'''=y''}^1 \sum_{x'''=y'''(\delta_{10} + \delta_{11})+y''}^{1+y'''(\gamma_{10} + \gamma_{11})+y''} \nonumber\\
  && \times 2^{(y-1)^2[(x)^2(\gamma_1+\delta_2) + (x-1)^2(\gamma_2+\delta_1)]} 2\times 2^{\left[(y''-1)^2[(x'')^2(\gamma_7+\delta_8) + (x''-1)^2(\gamma_8+\delta_7)]\right]\left[(y'''-1)^2[(x''')^2(\gamma_{10}+\delta_{11}) + (x'''-1)^2(\gamma_{11}+\delta_{10})]\right]} \nonumber\\
  && \times \mathcal G_4(x,y) \mathcal G_3(x',y') \mathcal G_4(x'',y'') \mathcal G_4(x''',y''') \Bigg] \Bigg\} \nonumber
\end{eqnarray}

\begin{eqnarray}
\label{eq:12qubittgate9}
  &+& \frac{\omega}{64} \left\{\sum_{y=0}^1 \sum_{x=y(\delta_1 + \delta_2)}^{1+y(\gamma_1 + \gamma_2)} \sum_{y'=y}^1 \sum_{x'=y'(\delta_4 + \delta_5)+y}^{1+y'(\gamma_4 + \gamma_5)+y} \right.\\
  && \left[ \sum_{y''=(\gamma_7+\delta_8)(y+y')}^{1+(\gamma_7+\gamma_8)(y+y')} \sum_{x''=y''(\beta_7+\alpha_8+\beta_8)}^{\substack{y''(\beta_7+\alpha_8+\beta_8)+\\y''(\gamma_7+\gamma_8)+\\(y''+1)(\gamma_7+\delta_8)}} \sum_{y'''=(\gamma_{10}+\delta_{11})(y+y'+x'')}^{1-(\gamma_{10}+\gamma_{11})(y+y'+x'')} \sum_{x'''=y'''(\beta_{10}+\alpha_{11}+\beta_{11})}^{\substack{y'''(\beta_{10}+\alpha_{11}+\beta_{11})+\\y'''(\gamma_{10}+\gamma_{11})+\\(y'''+1)(\gamma_{10}+\delta_{11})}} \right. \nonumber\\
  && \times 4 \times 2 \times 2^{\left[(y-1)^2[(x)^2(\gamma_1+\delta_2) + (x-1)^2(\gamma_2+\delta_1)]\right]\left[(y'-1)^2[(x')^2(\gamma_4+\delta_5) + (x'-1)^2(\gamma_5+\delta_4)]\right] \left[1-y''(\gamma_7-\gamma_8)^2-(y''-1)^2(\gamma_7-\delta_8)^2\right]\left[1-y'''(\gamma_{10}-\gamma_{11})^2-(y'''-1)^2(\gamma_{10}-\delta_{11})^2\right]} \nonumber\\
  && \times \left[ \mathcal G_4(x,y) \mathcal G_4(x',y') \mathcal G_1(x'',y'') + \mathcal G_2(x'',y'') \right] \left[ \mathcal G_1(x''',y''') + \mathcal G_2(x''',y''') \right] \nonumber\\
  && + 2 \sum_{y''=0}^{1} \sum_{x''=y''(\beta_7+\alpha_8+\beta_8)}^{\substack{y''(\beta_7+\alpha_8+\beta_8)+\\y''(\gamma_7+\gamma_8)+\\(y''+1)(\gamma_7+\delta_8)}} \sum_{x'''=0}^{1+\beta_{12}} \sum_{y'''=0}^{1+\alpha_{12}} \nonumber\\
  && \times 2 \times 2^{\left[(y-1)^2[(x)^2(\gamma_1+\delta_2) + (x-1)^2(\gamma_2+\delta_1)]\right]\left[(y'-1)^2[(x')^2(\gamma_4+\delta_5) + (x'-1)^2(\gamma_5+\delta_4)]\right]} 2^{1-y'''(\gamma_{10}-\gamma_{11})^2-(y'''-1)^2(\gamma_{10}-\delta_{11})^2} \nonumber\\
  && \times \mathcal G_4(x,y) \mathcal G_4(x',y') \left[ \mathcal G_1(x'',y'') + \mathcal G_2(x'',y'') \right]  \mathcal G_3(x''',y''') \nonumber\\
  && + \sum_{y''=(\gamma_7+\delta_8)(y+y')}^{1+(\gamma_7+\gamma_8)(y+y')} \sum_{x''=y''(\beta_7+\alpha_8+\beta_8)}^{\substack{y''(\beta_7+\alpha_8+\beta_8)+\\y''(\gamma_7+\gamma_8)+\\(y''+1)(\gamma_7+\delta_8)}} \sum_{y'''=y+y'+x''(1+\alpha_7+\beta_7)}^1 \sum_{x'''=y'''(\delta_{10} + \delta_{11})+(y+y'+x''(1+\alpha_7+\beta_7))}^{1+y'''(\gamma_{10} + \gamma_{11})+(y+y'+x''(1+\alpha_7+\beta_7))}\nonumber\\
  && \times 4 \times 2 \times 2^{\left[(y-1)^2[(x)^2(\gamma_1+\delta_2) + (x-1)^2(\gamma_2+\delta_1)]\right]\left[(y'-1)^2[(x')^2(\gamma_4+\delta_5) + (x'-1)^2(\gamma_5+\delta_4)]\right] \left[1-y''(\gamma_7-\gamma_8)^2-(y''-1)^2(\gamma_7-\delta_8)^2\right]\left[(y'''-1)^2[(x''')^2(\gamma_{10}+\delta_{11}) + (x'''-1)^2(\gamma_{11}+\delta_{10})]\right]} \nonumber\\
  && \times \mathcal G_4(x,y) \mathcal G_4(x',y') \left[ \mathcal G_1(x'',y'') + \mathcal G_2(x'',y'') \right] \mathcal G_4(x''',y''')\nonumber\\
  && + 2 \sum_{x''=0}^{1+\beta_9} \sum_{y''=0}^{1+\alpha_9} \sum_{y'''=0}^{1} \sum_{x'''=y'''(\beta_{10}+\alpha_{11}+\beta_{11})}^{\substack{y'''(\beta_{10}+\alpha_{11}+\beta_{11})+\\y'''(\gamma_{10}+\gamma_{11})+\\(y'''+1)(\gamma_{10}+\delta_{11})}} \nonumber\\
  && \times 2 \times 2^{\left[(y-1)^2[(x)^2(\gamma_1+\delta_2) + (x-1)^2(\gamma_2+\delta_1)]\right]\left[(y'-1)^2[(x')^2(\gamma_4+\delta_5) + (x'-1)^2(\gamma_5+\delta_4)]\right]} 2^{1-y'''(\gamma_{10}-\gamma_{11})^2-(y'''-1)^2(\gamma_{10}-\delta_{11})^2} \nonumber\\
  && \times \mathcal G_4(x,y) \mathcal G_4(x',y') \mathcal G_3(x'',y'') \left[ \mathcal G_1(x''',y''') + \mathcal G_2(x''',y''') \right] \nonumber\\
  && + 4 \sum_{x''=0}^{1+\beta_9} \sum_{y''=0}^{1+\alpha_9} \sum_{x'''=0}^{1+\beta_{12}} \sum_{y'''=0}^{1+\alpha_{12}} \nonumber\\
  && \times 2 \times 2^{\left[(y-1)^2[(x)^2(\gamma_1+\delta_2) + (x-1)^2(\gamma_2+\delta_1)]\right]\left[(y'-1)^2[(x')^2(\gamma_4+\delta_5) + (x'-1)^2(\gamma_5+\delta_4)]\right]} \nonumber\\
  && \times \mathcal G_4(x,y) \mathcal G_4(x',y') \mathcal G_3(x'',y'') \mathcal G_3(x''',y''')\nonumber\\
  && + 2 \sum_{x''=0}^{1+\beta_9} \sum_{y''=0}^{1+\alpha_9} \sum_{y'''=0}^1 \sum_{x'''=y'''(\delta_{10} + \delta_{11})}^{1+y'''(\gamma_{10} + \gamma_{11})} \nonumber\\
  && \times 2 \times 2^{\left[(y-1)^2[(x)^2(\gamma_1+\delta_2) + (x-1)^2(\gamma_2+\delta_1)]\right]\left[(y'-1)^2[(x')^2(\gamma_4+\delta_5) + (x'-1)^2(\gamma_5+\delta_4)]\right]} 2^{(y'''-1)^2[(x''')^2(\gamma_{10}+\delta_{11}) + (x'''-1)^2(\gamma_{11}+\delta_{10})]} \nonumber\\
  && \times \mathcal G_4(x,y) \mathcal G_4(x',y') \mathcal G_3(x'',y'') \mathcal G_4(x''',y''') \nonumber\\
  && + \sum_{y''=y+y'}^1 \sum_{x''=y''(\delta_7 + \delta_8)+(y+y')}^{1+y''(\gamma_7 + \gamma_8)+(y+y')} \sum_{y'''=(\gamma_{10}+\delta_{11})(y+y'+y'')}^{1-(\gamma_{10}+\gamma_{11})(y+y'+y'')} \sum_{x'''=y'''(\beta_{10}+\alpha_{11}+\beta_{11})}^{y'''(\beta_{10}+\alpha_{11}+\beta_{11})+y'''(\gamma_{10}+\gamma_{11})+(y'''+1)(\gamma_{10}+\delta_{11})} \nonumber\\
  && \times 4 \times 2 \times 2^{\left[(y-1)^2[(x)^2(\gamma_1+\delta_2) + (x-1)^2(\gamma_2+\delta_1)]\right]\left[(y'-1)^2[(x')^2(\gamma_4+\delta_5) + (x'-1)^2(\gamma_5+\delta_4)]\right] \left[(y''-1)^2[(x'')^2(\gamma_7+\delta_8) + (x''-1)^2(\gamma_8+\delta_7)]\right]\left[1-y'''(\gamma_{10}-\gamma_{11})^2-(y'''-1)^2(\gamma_{10}-\delta_{11})^2\right]} \nonumber\\
  && \times \mathcal G_4(x,y) \mathcal G_4(x',y') \mathcal G_4(x'',y'') \left[ \mathcal G_1(x''',y''') + \mathcal G_2(x''',y''') \right] \nonumber\\
  && + 2 \sum_{y''=0}^1 \sum_{x''=y''(\delta_7 + \delta_8)}^{1+y''(\gamma_7 + \gamma_8)} \sum_{x'''=0}^{1+\beta_{12}} \sum_{y'''=0}^{1+\alpha_{12}} \nonumber\\
  && \times 2 \times 2^{\left[(y-1)^2[(x)^2(\gamma_1+\delta_2) + (x-1)^2(\gamma_2+\delta_1)]\right]\left[(y'-1)^2[(x')^2(\gamma_4+\delta_5) + (x'-1)^2(\gamma_5+\delta_4)]\right]} 2^{(y''-1)^2[(x'')^2(\gamma_7+\delta_8) + (x''-1)^2(\gamma_8+\delta_7)]} \nonumber\\
  && \times \mathcal G_4(x,y) \mathcal G_4(x',y') \mathcal G_4(x'',y'') \mathcal G_3(x''',y''') \nonumber\\
  && + \sum_{y''=y+y'}^1 \sum_{x''=y''(\delta_7 + \delta_8)+(y+y')}^{1+y''(\gamma_7 + \gamma_8)+(y+y')} \sum_{y'''=y+y'+y''}^1 \sum_{x'''=y'''(\delta_{10} + \delta_{11})+(y+y'+y'')}^{1+y'''(\gamma_{10} + \gamma_{11})+(y+y'+y'')} \nonumber\\
  && \times 4 \times 2 \times 2^{\left[(y-1)^2[(x)^2(\gamma_1+\delta_2) + (x-1)^2(\gamma_2+\delta_1)]\right]\left[(y'-1)^2[(x')^2(\gamma_4+\delta_5) + (x'-1)^2(\gamma_5+\delta_4)]\right] \left[(y''-1)^2[(x'')^2(\gamma_7+\delta_8) + (x''-1)^2(\gamma_8+\delta_7)]\right]\left[(y'''-1)^2[(x''')^2(\gamma_{10}+\delta_{11}) + (x'''-1)^2(\gamma_{11}+\delta_{10})]\right]} \nonumber\\
  && \left. \times \mathcal G_4(x,y) \mathcal G_4(x',y') \mathcal G_4(x'',y'') \mathcal G_4(x''',y''') \right\} \Bigg] \Bigg\} \nonumber
\end{eqnarray}
}
\end{widetext}

Eq.~\ref{eq:12qubittgate} is made up of lines~\ref{eq:12qubittgate1}-~\ref{eq:12qubittgate9}. Before the simplification (described after Eq.~\ref{eq:permutingGausssums}), these correspond to \(256\) sets \emph{of sets} of Gauss sums, which are obtained from taking the sets of Gauss sums for the six-qutrit T-gate magic state and tensoring them with each other. Eq.~\ref{eq:12qubittgate1} consists of sets of Gauss sums numbering seven, six, seven, six, four, six, seven, six, and seven, respectively, all multiplied by seven. Eq.~\ref{eq:12qubittgate2}-\ref{eq:12qubittgate9} consists of the same numbers multiplied by six, seven, six, four, six, seven, six, and seven, respectively. As before, these are sets \emph{of sets} of Gauss sums, so chosen so that only one of these smallest sets of Gauss sums is non-zero for a given \(\hat P_k\).

After simplification, the sets of Gauss sums that consisted of \(7 \times 7 = 49\) Gauss sums are reduced to \(31\) Gauss sums. As a result, the formerly next-largest set consisting of \(6 \times 7=42\) Gauss sums, becomes the largest set and determines the worst-case scaling of \((\xi_{12})^{t/12} = 2^{\sim 0.449 t}\).

\section{Numerical Results}
\label{sec:gausssumnumericalresults}

Code of the equations in Sections~\ref{sec:3qubitdecomp}-\ref{sec:12qubitdecomp} can be found at https://s3miclassical.com/gitweb/. Its performance is plotted in Figure~\ref{fig:scaling}.

\end{document}